\documentclass[a4paper,11pt]{article}
\pdfoutput=1 

\usepackage{jheppub, amsmath, amssymb,amsthm,dsfont,graphicx,mathtools,multirow,rotating,placeins,verbatim} 
\usepackage[T1]{fontenc} 

\usepackage{tensor}
\usepackage[T1]{fontenc} 
\usepackage{float}
\usepackage{todonotes} 
\setuptodonotes{linecolor=black!50,backgroundcolor=yellow!80,size=footnotesize}

\newcommand{\Av}{\mathcal{A}}

\def\G{\mathcal{G}}
\def\be{\begin{equation}}
\def\ee{\end{equation}}
\def\ba{\begin{eqnarray}}
\def\ea{\end{eqnarray}}
\def\bi{\begin{itemize}}
\def\ei{\end{itemize}}
\newcommand{\beq}{\begin{eqnarray}}
\newcommand{\eeq}{\end{eqnarray}}

\allowdisplaybreaks

\usepackage{leftindex}

\def\F{\mathcal{F}}
\def\A{\mathcal{A}}
\def\zb{\bar{z}}

\def\I{\mathcal{I}}

\def\O{\mathcal{O}}

\def\Tr{\text{Tr}}

\newcommand{\R}{\mathbb{R}}
\newcommand{\Z}{\mathbb{Z}}
\newcommand{\N}{\mathbb{N}}

\newcommand{\Op}{\mathcal{O}}

\newcommand{\Ib}{\mathcal{I}}
\newcommand{\tr}{\text{tr}}

\newtheorem{lemma}{Lemma}

\title{Infinite-dimensional hierarchy of recursive extensions for all sub$^n$-leading soft effects in Yang-Mills}

\author[a,1]{Silvia Nagy,\note{Corresponding author.}}
\author[b]{Javier Peraza,}
\author[a]{Giorgio Pizzolo}

\affiliation[a]{Department of Mathematical Sciences, Durham University, Durham, DH1 3LE, UK}
\affiliation[b]{Department of Mathematics and Statistics , Concordia University, Montreal, H3G 1M8, Canada}

\emailAdd{silvia.nagy@durham.ac.uk}
\emailAdd{javier.perazamartiarena@concordia.ca}
\emailAdd{giorgio.pizzolo@durham.ac.uk}

\abstract{Building on our proposal in \cite{Nagy:2024dme}, we present in detail the construction of the extended phase space for Yang-Mills at null infinity, containing the asymptotic symmetries and the charges responsible for sub$^n$-leading soft theorems at all orders. The generality of the procedure allows it to be directly applied to the computation of both tree and loop-level soft limits. We also give a detailed study of Yang-Mills equations under the radial expansion, giving a thorough construction of the radiative phase space for decays compatible with tree-level amplitudes for both light-cone and radial gauges. This gives rise to useful recursion relations at all orders between the field strength and the vector gauge coefficients. We construct the sub$^n$-leading charges recursively, and show a hierarchical truncation such that each charge subalgebra is closed, and their action in the extended phase space is canonical. We relate these results with the infinite-dimensional algebras that have been recently introduced in the context of conformal field theories at null infinity. We also apply our method to the computation of non-universal terms in the sub-leading charges arising in theories with higher derivative interaction terms.
}

\begin{document} 

\maketitle


\section{Introduction}\label{Introduction}
The infrared limit of gauge and gravity theories has allowed for the discovery of deep connections between the soft limits of scattering amplitudes, asymptotic symmetries, and the memory effect, organised in the so-called IR triangle \cite{Strominger:2017zoo}. This has lead to insights which have allowed for the development of a conjectured holographic principle for asymptotically flat spacetimes\footnote{See reviews \cite{Pasterski:2021raf,Raclariu:2021zjz} and references within.}.

In this article we will be focusing specifically on the relation between soft theorems and asymptotic symmetries. The former, first discovered many years ago \cite{Weinberg1965,Low1958}, govern the behaviour of scattering amplitudes in the limit where one or more of the particles have a vanishingly small energy. The latter are local symmetries that survive in the limit where we go to the null boundary of spacetime ($\I$). Then, unlike for standard gauge transformations, we can construct charges, which, via the Ward identities, provide proofs for the aforementioned soft theorems. This relies on constructing a well-defined phase space at $\I$, containing the free data of the theory, and on which symmetries act canonically. The above can be set up perturbatively, via an expansion in the small energy parameter. At leading order, where an elegant universal factorisation rule arises, the result has by now been established \cite{Strominger:2013lka,Strominger:2013jfa,He:2014cra}, with progress to certain sub-leading orders in both gravity and gauge theories \cite{Campiglia2015,Campiglia2020,Campiglia:2021oqz,Strominger2014a,Lysov:2014csa,Donnelly:2016auv,Speranza:2017gxd,Freidel:2020ayo,Freidel:2020svx,Freidel:2020xyx,Freidel:2021dxw,Ciambelli:2021nmv,Freidel:2021cjp,Freidel:2021dfs,Ciambelli:2021vnn,Geiller:2024bgf,Geiller:2022vto,Mao:2017tey,Bern:2014vva,Saha:2019tub}.

In a brief Letter \cite{Nagy:2024dme}, working in Yang-Mills (YM) theory, we proposed a mechanism for extending the canonical phase space to accommodate the construction of asymptotic symmetries and charges required for establishing the soft theorems at arbitrary sub-leading orders\footnote{See also \cite{Peraza:2023ivy} and \cite{Nagy:2022xxs}, where this has been achieved in the simplified set-ups of massless QED and self-dual theories in light-cone gauge, respectively.}. The proposal relies on a modification of the Stueckelberg procedure \cite{Stueckelberg:1938hvi}, originally introduced to restore broken gauge symmetry. This extends our phase space with a new Goldstone-type field, which transforms under the sub-leading asymptotic symmetries. We additionally put forward a set of recursion relation for the e.o.m., valid at any order in the radial expansion, which feed into recursion relations for the charges at all orders. 

In this longer article, we will expand on the proposal in \cite{Nagy:2024dme}, and present some new results that follow from our prescription. Firstly we will give further details on the extended phase space construction, and prove that it leads to appropriate charges, satisfying the expected charge algebra. We consider a very broad set-up here, allowing for the presence of field-dependent parameters with very general polyhomogeneous fall-offs in the radial expansion. We also made no assumptions at this stage about the fall-off in the coordinate dual to the energy ($u$ in Bondi coordinates), which means it can be applied in principle to loop-level soft theorems as well (see e.g. \cite{Sahoo:2018lxl, Pasterski:2022djr, Donnay:2022hkf, Agrawal:2023zea, Choi:2024ygx,Campiglia:2019wxe,AtulBhatkar:2019vcb,Sahoo:2020ryf}). 

We give explicit all-order expressions for the symmetry transformations of the extended phase space Goldstone-like Stueckelberg field, which enter the computation of the sub$^n$-charges. Interestingly, these turn out to be controlled by the Bernoulli numbers, which arise from the perturbative expansion of (an operator version of) the characteristic power series generating the so-called Todd polynomial \cite{ToddPol}.

The general construction is independent of gauge and coordinate choice, as well as the choice of hypersurface on which the initial data is defined. We choose to illustrate it here by working in Bondi coordinates $(u,r,z,\zb)$ around $\I^+$, and we look at two different gauge conditions, namely light-cone and radial gauge. In these two contexts, we explicitly derive recursion relations at all orders in the expansion in $r$, thus formally justifying the choice of free data in the radiative phase space. We additionally study the expansion in the $u$ coordinate, which, via a Fourier transform, is related to the energy expansion relevant for sub$^n$-leading soft theorems.  

Our procedure allows us to give an elegant construction for the charges at null infinity, valid at any sub$^n$-leading order. We render these charges finite on $\Ib^+$, by extending the renormalization procedure from \cite{Campiglia:2016hvg} to all orders. From here, one can then construct finite charges at the corner $\Ib^+_-$ (see e.g. \cite{Freidel:2019ohg,Peraza:2023ivy}). We find that the charges can be neatly organised in a radial expansion, which closely follows the expansion in the new Goldstone field\footnote{See e.g. \cite{Freidel:2023gue,Kampf:2023elx,He:2024skc} for some other works where Goldstone modes arise at null infinity.} introduced via the Stueckelberg trick. We establish the existence of a closed sub-algebra of charges at each order in the perturbative expansion, thus showing how a robust truncation of the procedure can be established at any desired order for practical calculations related to soft theorems. We also match with known results from the literature at level $n=0$ (\cite{Strominger:2013lka}) and $n=1$ (\cite{Campiglia:2021oqz}). Moreover, the charges themselves turn out to satisfy recursions that allow us to construct any order from previous orders (first noted in \cite{Campiglia:2021oqz} for leading and sub-leading), and also straightforwardly match the abelian hierarchy of sub$^n$-leading charges \cite{Campiglia:2018dyi}. The iterative approach to the construction of the charges has the advantage that we only use equations of motion and the phase space formalism, and we do not need to make any direct reference to the scattering language. We establish the correspondence with the usual expressions for the sub-leading charge coming from the Ward identities, showing how it can be indeed carried out at all orders. 

In the context of Ward identities the sub-leading charges generally contain universal and non-universal terms. Interestingly, we can also have \emph{quasi-universal} contributions to the sub-leading soft charge when looking at higher derivative interactions between the gauge vector in QED and a scalar \cite{Elvang:2016qvq, Laddha:2017vfh}. We extend the derivation of the quasi-universal contributions to the sub-leading charge in Yang-Mills for an interaction of the form $\phi \tr ( F^2 )$, using the dressing procedure in our framework. This example illustrates the applicability of our procedure to different theories.

One of the more intriguing discoveries of the celestial holography program has been the presence of certain infinite dimensional algebras in sectors of Yang-Mills and gravity (where it is called the $w_{1+\infty}$ algebra,  first defined by \cite{Bakas:1989xu} and connected to gravity in \cite{Strominger:2021mtt}, see also \cite{Pope:1989ew,Fairlie:1990wv,Pope:1991ig,Bu:2022iak,Bittleston:2023bzp,Bittleston:2024rqe,Taylor:2023ajd,Kmec:2024nmu,Monteiro:2022lwm,Himwich:2021dau,
Lipstein:2023pih,Lipstein:2023tapp,Adamo:2021lrv} for further work and deformations). The quantization of these algebras (known as $W_{1+\infty}$, \cite{Pope:1991zka}) is expected to be a quantum symmetry of the S-matrix, and it is interesting that one can trace its origin to classical symmetries and their charge algebra (see e.g. \cite{Strominger:2021mtt} for YM and \cite{Freidel:2021ytz, Geiller:2024bgf} for gravity). We find that a subset of our recursion relations encompasses the sectors displaying these symmetries.      

The article is organised as follows. In \autoref{sec_ext_phase_space} we give details of the general phase space calculation. In \autoref{eom_scri_plus}, working in Bondi coordinates around null infinity we give the e.o.m. at all orders in $r$, and the proofs of the recursion relations in both light-cone and radial gauge. The construction of the necessary charges and their recursion relations at all orders is given in \autoref{charges}. In \autoref{charge_algebra} the charge algebra valid at any sub-leading order is constructed and shown to close order by order. In \autoref{sec_charges_corners}, we show the link between the recursive formulas in \autoref{eom_scri_plus} and the charges constructed in \autoref{charges}, as well as showing that our results for the leading and sub-leading charges agree with previous proposals from the literature \cite{Strominger:2013lka,Campiglia:2021oqz} and the full abelian case \cite{Campiglia:2018dyi,Peraza:2023ivy}. In \autoref{Relations with infinite algebras} we identify, via the recursion relations, the subsectors of our theory that carry the infinite symmetry algebras along the lines of the discussion in \cite{Freidel:2023gue}. In \autoref{higher_oerder_der_int}, we discuss the quasi-universal contributions to the charges from higher order derivative interactions, using our prescription in a particular example for the case of sub-leading Yang-Mills (following the QED example given in \cite{Laddha:2017vfh}). In \autoref{Conclusions} we provide an outlook of the results, and propose future explorations. The appendices contain detailed computations, useful formulae and discussions that complement the main text. In \autoref{Extended phase space calculations}, we give details of calculations pertaining to the general extended phase space and we collect some useful formulae. In \autoref{app_YM_eq_r_exp}, we provide proofs for the derivation of the equations of motion and Bianchi identities at arbitrary order in the radial expansion, which were used in finding the recursive formulae of \autoref{eom_scri_plus}. In \autoref{YM e.o.m. in flat Bondi coordinates}, we turn to the so-called flat Bondi coordinates, which differ from the standard Bondi through the vanishing of the $g_{uu}$ component of the metric, as well as the celestial sphere being mapped into the complex plane $\mathbb{C}^2$ via the stereographic projection. We see that our procedure continues to hold, and find simplified recursion relations in these coordinates. In \autoref{sub_n_leading_charge_limit} we provide an explicit derivation of the charges via the limit in $t\rightarrow +\infty$, showing that the renormalization procedure can be applied. Finally, in \autoref{last_appendix} we show how our general large gauge parameters can be chosen to coincide with the particular generators of charges coming from the usual derivation of the soft theorems from the Ward identity. 

\section{Extended phase space in arbitrary gauge and coordinates}  \label{sec_ext_phase_space}
We begin by reviewing the general procedure introduced in \cite{Nagy:2024dme}. The starting point is YM theory with the standard equation of motion  
\be 
\label{eom_gen}
\mathcal{E}_\nu\equiv D^\mu \mathcal{F}_{\mu\nu}=0
\ee 
where
\be 
\mathcal{F}_{\mu\nu}=\partial_\mu \mathcal{A}_\nu-\partial_\nu \mathcal{A}_\mu - i [\mathcal{A}_\mu,\mathcal{A}_\nu]
\ee
The gauge fields transform in the usual way:
\be\label{gauge_gen}
\mathcal{A}'_\mu=e^{i\Lambda}\mathcal{A}_\mu e^{-i\Lambda}+ie^{i\Lambda}\partial_\mu e^{-i\Lambda}.
\ee
and we assume that $\mathcal{A}_\mu$ satisfies some gauge condition
\be \label{gf_gen}
\mathcal{G}(\A_\mu)=0
\ee
We will restrict to the residual gauge symmetries preserving the condition above, as is standard when studying asymptotic symmetries. For now we do not need to specify the coordinates. For explicit examples later on, we will work in some appropriately chosen coordinates (these could be e.g. Bondi, flat Bondi, or even light-cone coordinates). We find it convenient to adopt the following notation for the coordinates:
\be 
x^\mu=(\mathfrak{r},\vec{\mathbf{y}}), 
\ee 
where $\mathfrak{r}$ is the expansion parameter, and $\vec{\mathbf{y}}$ are the the coordinates of the hypersurface on which we define our free data\footnote{Most commonly, in Bondi coordinates, we have $\mathfrak{r}=r$ and $\vec{\mathbf{y}}=(u,z,\bar{z})$ describing $\mathcal{I}^+$. However, the procedure applies in other coordinates as well.}.

We consider a very general expansion for our gauge field, in terms of polyhomogeneous functions
\be \label{A_st_exp}
\mathcal{A}_\mu=\sum_{n,k}A_\mu^{(-n;k)}(\vec{\mathbf{y}})\frac{\text{log}^k \mathfrak{r}}{\mathfrak{r}^n} \ ,
\ee 
where $n$ and $k$ are chosen such that $\lim_{\mathfrak{r}\to\infty}\frac{\text{log}^k \mathfrak{r}}{\mathfrak{r}^n}=\mathcal{O}(1)$. In practice, this means that we only allow terms with $n\geq 0$, and $k$ is unconstrained, except for $n=0$, where we only allow a contribution from $k=0$. We remark that upon summing all orders in the above, one could of course use the Taylor series for the log, and ultimately have an expression of the form $\sum_{n=0}^{\infty}B_\mu^{(-n)}(\vec{\mathbf{y}})\frac{1}{\mathfrak{r}^n}$ . However, in practice we will stop at a finite order in $n$ and $k$, in which case the log terms become necessary (see e.g. \cite{Campiglia:2021oqz} in Lorenz gauge). Certain gauge choices have the benefit of avoiding the log terms, as we will see in \autoref{sec_extended_phase_space_LC}, where we work in axial gauge.

Then the phase space is described by
\be
\label{gen_ph_sp} 
\Gamma^0=\left\{\mathfrak{A}^{0} \text{\ satisfying \eqref{eom_gen} and \eqref{gf_gen}}
\right\}
\ee
where $\mathfrak{A}^{0}$ comprises the necessary free data which determines the radiative fields to all orders in $\mathfrak{r}$, via the equations of motion and Bianchi identities. 

In \autoref{eom_scri_plus}, working in Bondi coordinates, we give expressions for the equations of motion and Bianchi identities at all orders in $r$. Then, for both radial and light-cone gauge, we give recursion relations valid at all orders in the perturbation parameter, proving that the fields in the bulk can be entirely reconstructed from the boundary data at null infinity. 

The large gauge transformations responsible for the leading order soft theorems are characterised by a parameter $\lambda \equiv \Lambda^{(0)} (\vec{\mathbf{y}})$, subject to a constraint coming from the gauge condition \eqref{gf_gen}:
\be\label{gen_gf_L_0} 
\mathcal{G}(D_\mu \Lambda^{(0)}(\vec{\mathbf{y}}))=0 \ , 
\ee
thus preserving $\Gamma^0$.

We now wish to construct an extended phase space, able to accommodate additional large gauge transformations with a fall-off of the form\footnote{Note that small gauge transformations, i.e. those which vanish in the limit $\mathfrak{r}\to\infty$, have been set to $0$ for simplicity, since they do not contribute to the soft theorems.}
\be \label{Lambda_plus_exp}
\Lambda_+(x)=\sum_{n,k}\mathfrak{r}^n \text{log}^k\mathfrak{r}\ \Lambda^{(n;k)}(\vec{\mathbf{y}}) \ ,
\ee
where $n$ and $k$ are chosen such that $\mathfrak{r}^n \text{log}^k\mathfrak{r}$ diverges in the limit $\mathfrak{r}\to\infty$. In \cite{Campiglia:2021oqz} this was achieved at first order in the radial expansion in YM in Lorenz gauge, working in Bondi coordinates. In \cite{Nagy:2022xxs}, a construction was given to all orders, for the self-dual YM theory, working in light-cone gauge. Here we will show that the construction extends to full YM theory, for $n=\infty$, with an arbitrary gauge and coordinate choice. 

First we note that in some cases one need to generalise \eqref{Lambda_plus_exp} to a field-dependent parameter\footnote{This might arise because we are imposing further constraints on the transformation. For example, requiring that they preserve Lorenz gauge can lead to field-dependent transformation parameters, see \cite{Campiglia:2021oqz}.}, and we consider here the most general situation:
\be \label{gauge_par_field_dep}
\begin{aligned}
\breve{\Lambda}_+= \breve{\Lambda}_+(\mathcal{A}_\mu(x),\Lambda_+(x))
=\sum_{n,k}
\mathfrak{r}^n \text{log}^k\mathfrak{r}\ f^{(n;k)}(A_\mu^{(0;0)}(\vec{\mathbf{y}}),...,\Lambda^{(0;1)}(\vec{\mathbf{y}}),...) \ ,
\end{aligned}
\ee 
where $n$ and $k$ are such that $\mathfrak{r}^n \text{log}^k\mathfrak{r}$ diverges as $\mathfrak{r}\to\infty$. Here $\A_\mu(x)$ is expanded as in \eqref{A_st_exp}, and the $\Lambda_+(x)$ expansion is given in \eqref{Lambda_plus_exp}. Later in the article, we will specialise to cases where $\breve{\Lambda}_+$ is independent of the fields, which will simplify calculations considerably. However, for now we continue with the most general solution as described in \eqref{gauge_par_field_dep}.

The next step is to define an extended phase space capable of accommodating a gauge transformation with a parameter given by \eqref{gauge_par_field_dep}. These transformations obviously violate the fall-off \eqref{A_st_exp}. We can interpret these violations as a form of symmetry breaking, which leads us to using a modified Stueckelberg trick in the construction of our phase space. The Stueckelberg procedure\footnote{Initially introduced in the context of electromagnetism \cite{Stueckelberg:1938hvi}, also extended to (super)gravity \cite{Kuchar:1991xd,Nagy:2019ywi,Bansal:2020krz}.} is usually employed when a gauge symmetry is broken e.g. due to the presence of mass. It allows us to reinstate the symmetry, by promoting the local transformation parameter to a new field, which transforms non-linearly. We will adapt it to our situation by defining the object: 
\be 
\breve{\Psi}= \breve{\Psi}(\mathcal{A}_\mu(x),\Psi(x))=
\sum_{n,k}
\mathfrak{r}^n \text{log}^k\mathfrak{r}\ f^{(n;k)}(A^{(0;0)}_\mu(\vec{\mathbf{y}}),...,\Psi^{(0;1)}(\vec{\mathbf{y}}),...)
\ee
with
\be \label{gen_psi_no_fields}
\Psi(x)=\sum_{n,k}\mathfrak{r}^n \text{log}^k\mathfrak{r}\ \Psi^{(n;k)}(\vec{\mathbf{y}}).
\ee 
with $n,k$ as in \eqref{gauge_par_field_dep} and where $\breve{\Psi}$ has the same functional dependence on $(A_\mu,\Psi)$ as $\breve{\Lambda}_+$ does on $(A_\mu,\Lambda_+)$. Equivalently, $\breve\Psi$ is obtained from $\breve{\Lambda}_+$ by applying the replacement rule
\be \label{stueck_trick}
\Lambda_+(x)\quad \to \quad \Psi(x),
\ee 
We are now ready to define our extended phase space as: 
\be
\label{extended_phase_space_YM}
\Gamma_{\infty}^{\text{ext}} := \Gamma^0\times\{\Psi(x)   \text{ as defined in \eqref{gen_psi_no_fields}} \}
\ee
where $\Gamma^0$ is the original phase space defined in \eqref{gen_ph_sp}. We will see that the object $\breve{\Psi}$ plays a crucial role in our construction. In order to proceed with the construction of charges on our extended phase space, we must first determine the transformation properties of our new field $\Psi$. We generalise the procedure employed in \cite{Campiglia:2021oqz,Nagy:2022xxs} and start by defining the extended gauge field 
\be\label{stuec_gen}
\tilde{\mathcal{A}}_\mu = e^{i\breve{\Psi}}\mathcal{A}_\mu e^{-i\breve{\Psi}}+ie^{i\breve{\Psi}}\partial_\mu e^{-i\breve{\Psi}}
\ee
This mimics a gauge-transformed gauge field, with the parameter replaced by our new field, which is a hallmark of the Stueckelberg procedure. We require that it satisfies the following consistency condition:  
\be
\label{delta_Lambda_HatA}
\delta_{\breve\Lambda}\tilde{\mathcal{A}}_\mu = \tilde{D}_\mu\breve\Lambda 
\ee
where 
\be 
\breve\Lambda=\Lambda^{(0)}+\breve\Lambda_+
\ee 
and $\tilde{D}_\mu$ is the gauge-covariant derivative w.r.t. $\tilde{A}_\mu$. In the above, $\Lambda^{(0)}$ is the standard (leading order) large gauge parameter. Equation \eqref{delta_Lambda_HatA} is tantamount to saying that the extended field transforms as a gauge field in the extended space. Note that when writing the actual transformation, we are working only to linear order in the parameter, as this will be sufficient for the construction of the charges. This is in contrast with the Stueckelberg step \eqref{stuec_gen}, where all orders are necessary. The consistency condition essentially states that the additional fields on the boundary do not correspond to new fields in the bulk, but rather are coming from pieces of the bulk field which were discarded when assuming the fall-off in \eqref{A_st_exp}. Finally, as is standard in the Stueckelberg procedure, the transformation of the original field $\A_\mu$ is unchanged.   

Following a calculation similar to that in \cite{Nagy:2022xxs} (see \autoref{Extended phase space calculations} for details), we rewrite the LHS of \eqref{delta_Lambda_HatA} as
\be \label{LHS_consistency}
\delta_{\breve\Lambda}\tilde{\mathcal{A}}_\mu =
e^{i\breve{\Psi}}\left\{\delta_{\breve\Lambda}\mathcal{A}_\mu+D_\mu\left[e^{-i\breve{\Psi}}\mathcal{O}_{-i\breve{\Psi}}(\delta_{\breve\Lambda}\breve{\Psi})e^{i\breve{\Psi}}\right]\right\}e^{-i\breve{\Psi}}
\ee
where we have used that
\be \label{delta_exp_O_op}
\delta e^X = e^X\mathcal{O}_X(\delta X)
\ee
with 
\be \label{O_def_gen}
\mathcal{O}_X  := \frac{1-e^{-ad_X}}{ad_X}  \ .
\ee 
The above is defined via its series expansion
\be \label{O_series_exp}
\mathcal{O}_X= \sum_{k=0}^\infty\frac{(-1)^k}{(k+1)!}(ad_X)^k \ ,
\ee 
such that
\begin{equation}
    \begin{aligned}
        \mathcal{O}_{-i\breve{\Psi}}(\delta_{\breve\Lambda}\breve{\Psi}) &= \sum_{k=0}^{\infty}\frac{(-1)^k}{(k+1)!}(ad_{-i\breve{\Psi}})^k(\delta_{\breve\Lambda}\breve{\Psi}) \\
        &= \delta_{\breve\Lambda}\breve{\Psi}+\frac{i}{2}[\breve{\Psi},\delta_{\breve\Lambda}\breve{\Psi}]-\frac{1}{6}[\breve{\Psi},[\breve{\Psi},\delta_{\breve\Lambda}\breve{\Psi}]]-\frac{i}{24}[\breve{\Psi},[\breve{\Psi},[\breve{\Psi},\delta_{\breve\Lambda}\breve{\Psi}]]]+\dots\;.
    \end{aligned}
\end{equation}
The RHS of \eqref{delta_Lambda_HatA} becomes
\be \label{RHS_consistency}
\tilde{D}_\mu{\breve\Lambda}=e^{i\breve{\Psi}}D_\mu(e^{-i\breve\Psi}{\breve\Lambda} e^{i\breve\Psi})e^{-i\breve\Psi}
\ee 
Then, plugging \eqref{LHS_consistency} and \eqref{RHS_consistency} into \eqref{delta_Lambda_HatA} we obtain:
\be
\label{B_first_result}
\delta_{\breve\Lambda}\mathcal{A}_\mu+D_\mu\left[e^{-i\breve\Psi}\mathcal{O}_{-i\breve\Psi}(\delta_{\breve\Lambda}\breve\Psi)e^{i\breve\Psi}\right] = D_\mu\left(e^{-i\breve\Psi}{\breve\Lambda} e^{i\breve\Psi}\right)
\ee
Let us now focus on one of the components of the equation above, call this $w$. In any gauge, we can choose a $w$ such at $A_w^{(0;0)}\neq 0$. Then we have 
\be
\label{deltaLambda_AZ0}
\delta_{\breve\Lambda} A_w^{(-n;k)} := \left(D_w{\breve\Lambda}\right)^{(-n;k)} = \delta_{n,0}\delta_{k,0}\partial_w\Lambda^{(0)}-i[A_w^{(-n;k)},\Lambda^{(0)}] \ ,
\ee
where we have imposed that only the leading order gauge parameter $\Lambda^{(0)}$ contributes (the sub-leading gauge parameters in \eqref{Lambda_plus_exp} will only appear in the transformation of $\Psi$, and we recall that we are ignoring the small gauge transformations). Making use of the above, we can rearrange \eqref{B_first_result} into 
\be\label{def_delta_breve_Psi} 
\mathcal{O}_{-i\breve\Psi}(\delta_{\breve\Lambda}\breve\Psi) = {\breve\Lambda}-e^{i\breve\Psi}\Lambda^{(0)}e^{-i\breve\Psi} 
\ee 
This can be solved to any order in $\mathfrak{r}$ by inverting $\mathcal{O}_{-i\Psi}$:
\be\label{psi_inv_O} 
\delta_{\breve\Lambda}\breve\Psi = \mathcal{O}_{-i\breve\Psi}^{-1}({\breve\Lambda}-e^{i\breve\Psi}\Lambda^{(0)}e^{-i\breve\Psi})
\ee 
We will make use of the formal expansion
\be
\left(\frac{1-e^{-X}}{X}\right)^{-1} = \sum_{m=0}^{\infty}\frac{B^+_m X^m}{m!}
\ee
where $B^+_m$ are the Bernoulli numbers
\be \label{Bernoulli_def}
B^+_0=1,\quad B^+_1=\frac{1}{2},\quad B^+_2=\frac{1}{6},\quad ...
\ee
Using this in \eqref{psi_inv_O} we find\footnote{The proof is completely analogous to the one in Appendix A3 of \cite{Nagy:2022xxs}.}:
\be \label{trans_stueck_perturbative_field}
\delta_{\breve\Lambda}^{[m]}\breve\Psi = \frac{B_m^+}{m!}(ad_{-i\breve\Psi})^m\left[{\breve\Lambda}+(-1+2\delta_{m,1})\Lambda^{(0)}\right]
\ee
where $\delta_{\breve\Lambda}^{[m]}\breve\Psi$ denotes the variation of $\breve\Psi$ at order $m$ in $\breve\Psi$. The next step is to extract the transformation of $\Psi$ from the above. This will of course depend on the exact form of $\breve\Psi$ as a function of $\Psi$ and $\mathcal{A}_\mu$, so it will be done case by case, depending on the gauge choice \eqref{gf_gen}. Nevertheless, we can proceed in the general case, as it is the $\breve{\Psi}$ that will appear in the construction of the charges, and so it is its transformation that will be relevant for computing the charge algebra.  Before continuing with the charge algebra, we pause to make some remarks about the transformation of the Stueckelberg fields. Firstly, we note that, at 0th order in the field, it transforms via a shift
\be 
\delta_{\breve\Lambda}^{[0]}\breve\Psi =\breve\Lambda-\Lambda^{(0)}=\breve\Lambda_+
\ee 
so they have an interpretation as Goldstone modes for the symmetry breaking in the bulk. 

Secondly, recall that 
\be
B_{2k+1}=0,\quad \text{for}\quad k>0
\ee
i.e. only even powers of the fields will contribute to the transformation \eqref{trans_stueck_perturbative_field} beyond $m=1$. We leave the interpretation of this observation for future work. 

Finally, there is a close relation between the Bernoulli numbers and the Riemann $\zeta$-function
\be 
\zeta(m) = \frac{(-1)^{\frac{m}{2}+1}}{2} \frac{B_m^+ (2 \pi)^m}{m!},
\ee 
for $m$ even and larger than $1$. Then we can alternatively write \eqref{trans_stueck_perturbative_field}, for $m>1$, as
\be 
\delta_{\breve\Lambda}^{[m]}\breve\Psi = \frac{2 (-1)^{\tfrac{m}{2}+1}\zeta(m)}{(2 \pi)^m}(ad_{-i\breve\Psi})^m\left[{\breve\Lambda}+(-1+2\delta_{m,1})\Lambda^{(0)}\right].
\ee

\subsection{Charge algebra}
\label{subsec:Charge Algebra}
The charge, computed via the covariant space space formalism\footnote{More details on the computation of charges are given in \autoref{charges}.}, can be expressed as a two-dimensional integral, over a hypersurface $\mathcal{B}^2$ 
\be 
\tilde{Q}_{\breve\Lambda} = \int_{\mathcal{B}^2} \tr ({\breve\Lambda} \tilde{\F}^{\mu \nu}) dS_{\mu \nu},
\ee
where $\tilde{\F}^{\mu \nu}$ is constructed from $\tilde\A_\mu$, and takes the form
\be \label{def_c_F}
\tilde{\F}_{\mu \nu}
\equiv \partial_\mu \tilde{\mathcal{A}}_\nu-\partial_\nu \tilde{\mathcal{A}}_\mu - i [\tilde{\mathcal{A}}_\mu,\tilde{\mathcal{A}}_\nu]
=e^{i\breve\Psi}\F_{\mu\nu}e^{-i\breve\Psi}
\ee
Let $x_B$ and $y_B$ the the coordinates on $\mathcal{B}_2$, then we have
\be 
d S_{\mu\nu}=dx_B dy_B \sqrt{g_B}m_\mu n_\nu
\ee 
where $g_B$ is the induced metric on $\mathcal{B}_2$ and $n_\mu$,$m_\mu$ are unit vectors orthogonal to $\mathcal{B}_2$. We can now define the charge density
\be 
\tilde q_{\breve\Lambda}=\tr\left(\sqrt{g_B} {\breve\Lambda} \tilde{\F}_{mn}\right)^{(0)}
\ee 
where we have denoted
\be 
\tilde{\F}_{mn}=\tilde{\F}^{\mu \nu}m_\mu n_\nu
\ee 
In terms of the original field strength, the charge density is
\be \label{gen_comm_qs} 
\begin{aligned} 
\tilde q_{\breve\Lambda}=&\tr\left(\sqrt{g_B} {\breve\Lambda} e^{i\breve\Psi}\F_{mn}e^{-i\breve\Psi}\right)^{(0)}\\
=&\tr\left(\sqrt{g_B}e^{-i\breve\Psi}{\breve\Lambda} e^{i\breve\Psi}\F_{mn}\right)^{(0)}
\end{aligned} 
\ee 
We now want to check the charge algebra 
\be 
\{\tilde  q_{{\breve\Lambda}_1},\tilde q_{{\breve\Lambda}_2}\}_*=\frac{1}{2}\left[\delta_{{\breve\Lambda}_1}\tilde q_{{\breve\Lambda}_2}+\tilde q_{\delta_{{\breve\Lambda}_1}{\breve\Lambda}_2}-(1\leftrightarrow2)\right]
\ee 
Note that we are using a modified bracket \cite{Barnich:2010eb,Barnich:2010xq,Barnich:2013sxa,Fiorucci:2021pha,Campiglia:2021oqz}, as $\breve\Lambda$ can be field dependent in general. Explicitly we then have:
\be \label{commGenStart}
\begin{aligned}
\{ \tilde q_{{\breve\Lambda}_1},\tilde q_{{\breve\Lambda}_2}\}_*
=&\frac{1}{2}\Bigg[\tr\left(\sqrt{g_B}\left(\delta_{{\breve\Lambda}_1}e^{-i\breve\Psi}\right){\breve\Lambda}_2 e^{i\breve\Psi}\F_{mn}\right)^{(0)}
+2\ \tr\left(\sqrt{g_B}e^{-i\breve\Psi}\left(\delta_{{\breve\Lambda}_1}{\breve\Lambda}_2\right) e^{i\breve\Psi}\F_{mn}\right)^{(0)}\\
&+\tr\left(\sqrt{g_B}e^{-i\breve\Psi}{\breve\Lambda}_2 \left(\delta_{{\breve\Lambda}_1}e^{i\breve\Psi}\right)\F_{mn}\right)^{(0)}
+\tr\left(\sqrt{g_B}e^{-i\breve\Psi}{\breve\Lambda}_2 e^{i\breve\Psi}\delta_{{\breve\Lambda}_1}\F_{mn}\right)^{(0)} \Bigg] \\
& -
(1\leftrightarrow2)
\end{aligned}
\ee 
Using \eqref{delta_exp_O_op} and \eqref{O_def_gen}, we have
\be \label{commGen1}
\begin{aligned}
\delta_{{\breve\Lambda}_1}e^{-i\breve\Psi}=&-ie^{-i\breve\Psi}\mathcal{O}_{-i\breve\Psi}\left(\delta_{{\breve\Lambda}_1}\breve\Psi \right)\\
=&-ie^{-i\breve\Psi}\left({\breve\Lambda_1}-e^{i\breve\Psi}\Lambda_1^{(0)}e^{-i\breve\Psi}  \right)\\
=&-ie^{-i\breve\Psi}{\breve\Lambda_1} +i\Lambda_1^{(0)}e^{-i\breve\Psi}
\end{aligned}
\ee 
where to get to the second line we used \eqref{def_delta_breve_Psi}. Similarly, using \eqref{delta_e_neg}, we write 
\be \label{commGen2}
\begin{aligned}
\delta_{{\breve\Lambda}_1}e^{i\breve\Psi}
=&i\mathcal{O}_{-i\breve\Psi}\left(\delta_{{\breve\Lambda}_1}\breve\Psi \right)e^{i\breve\Psi}\\
=&i\left({\breve\Lambda_1}-e^{i\breve\Psi}\Lambda_1^{(0)}e^{-i\breve\Psi} \right)e^{i\breve\Psi}\\
=&i\breve\Lambda_1e^{i\breve\Psi}-i e^{i\breve\Psi}\Lambda_1^{(0)}
\end{aligned}
\ee
To compute the final term in \eqref{gen_comm_qs}  we use \eqref{deltaLambda_AZ0} to write
\be \label{commGen3}
\delta_{\breve\Lambda_1}\F_{mn}=-i[\F_{mn},\Lambda_1^{(0)}]
\ee 
remembering that we have set the small gauge transformations to vanish. Finally, plugging \eqref{commGen1}, \eqref{commGen2}, and \eqref{commGen3} into \eqref{commGenStart}, we get
\be 
\{\tilde  q_{{\breve\Lambda}_1},\tilde q_{{\breve\Lambda}_2}\}_*=\tilde q_{[{\breve\Lambda}_1,{\breve\Lambda}_2]_*}
\ee
as required, where the deformed bracket is as expected for field dependent parameters \cite{Barnich:2010eb,Barnich:2010xq,Barnich:2013sxa,Fiorucci:2021pha,Campiglia:2021oqz}: 
\be \label{mod_bracket}
[{\breve\Lambda}_1,{\breve\Lambda}_2]_*=-i[{\breve\Lambda}_1,{\breve\Lambda}_2]+\delta_{{\breve\Lambda}_1}{\breve\Lambda}_2-\delta_{{\breve\Lambda}_2}{\breve\Lambda}_1
\ee


\subsection{Case study: Bondi coordinates in axial gauge} \label{sec_extended_phase_space_LC}
Let us make the previous section more concrete, by looking at a specific example of coordinates, namely the Bondi coordinates $(u,r,z,\zb)$ around $\mathcal{I}^+$, with line element 
\begin{equation}\label{bondi_coord_def1}
   ds^2=-du^2-2dudr+2r^2\gamma dzd\bar{z}
\end{equation}
where
\be \label{bondi_coord_def2}
\gamma=\frac{2}{(1+z\bar{z})^2} \ .
\ee 
Let us choose an axial gauge, i.e.
\be \label{axial_gauge}
\mathcal{G}(\A_\mu)=a^\mu \A_\mu\equiv 0 \ , 
\ee
where $a^\mu$ is some constant vector. Different choices of $a^\mu$ give different gauges of interest for various applications. Later in the article we will work in the light-cone, and radial gauges respectively. For now we proceed with a general $a^\mu$. In this gauge, the expansion of the gauge fields \eqref{A_st_exp} simplifies to
\be 
\mathcal{A}_\mu=\sum_{n=0}^{\infty}\frac{A_\mu^{(-n)}(u,z,\zb)}{r^n} \ ,
\ee 
i.e. the log terms are no longer present.  The radiative phase space will be as in \eqref{gen_ph_sp}. The leading order large gauge transformations are constrained by the gauge choice, and explicitly we have:
\be 
a^\mu D_\mu \Lambda^{(0)}=a^\mu\partial_\mu \Lambda^{(0)}\equiv 0 \ ,
\ee 
where in the first equality we made use of \eqref{axial_gauge}. We see that the only constraint on the gauge parameter is that it does not depend on the linear combination of coordinates $a_\mu x^\mu$.  We are now looking to construct an extended phase space, able to accommodate a gauge parameter
\be \label{simple_lambda_plus}
\Lambda_+(x)=\sum_{n=1}^{\infty}r^n\Lambda^{(n)}(u,z,\zb) \ ,
\ee 
where we highlight the absence of the logarithmic terms, reflecting the absence of these terms in the expansion of $\A_\mu$ above. At this stage, one could choose to impose the gauge fixing condition on the new terms, i.e.
\be \label{extra_gauge_cond}
a^\mu\partial_\mu \Lambda_+(x)=0
\ee 
though we emphasize that this is not required by our procedure. In either case, the upshot is that we do not need a field-dependent parameter in this case. Thus \eqref{gauge_par_field_dep} simplifies significantly, and we have
\be 
\breve\Lambda_+(x)=\Lambda_+(x)
\ee 
We are now ready to construct our Stueckelberg field, which is simply
\be 
\breve\Psi(x)=\Psi(x)=\sum_{k=1}^{\infty}r^k\Psi^{(k)}(u,z,\zb)
\ee 
The extended gauge field will be as in \eqref{stuec_gen}. Upon applying the consistency condition \eqref{delta_Lambda_HatA}, and following the steps in the previous section, we now explicitly obtain the transformation rule for the Stueckelberg field
\be 
\delta_{\Lambda}\Psi = \mathcal{O}_{-i\Psi}^{-1}({\Lambda}-e^{i\Psi}\Lambda^{(0)}e^{-i\Psi})
\ee
with $\mathcal{O}$ as in \eqref{O_def_gen}. Working at order $m$ in $\Psi$, this can be written in terms of the Bernoulli numbers \eqref{Bernoulli_def} as 
\be
\delta_{\Lambda}^{[m]}\Psi = \frac{B_m^+}{m!}(ad_{-i\Psi})^m\left[{\Lambda}+(-1+2\delta_{m,1})\Lambda^{(0)}\right]
\ee
Let us write the first few orders explicitly. At $[m]=0$ we have
\be \label{zeroth_order_variation}
\delta_\Lambda^{[0]}\Psi=\Lambda-\Lambda^{(0)}=\Lambda_+
\ee
as expected for a Stuckelberg-type field. At linear order in $\Psi$ we have
 
\be
\delta_\Lambda^{[1]}\Psi = -\frac{i}{2}[\Psi,\Lambda+\Lambda^{(0)}] \ ,
\ee 
while at second order the transformation becomes
\be
\delta_\Lambda^{[2]}\Psi = -\frac{1}{12}[\Psi,[\Psi,\Lambda-\Lambda^{(0)}]] \ .
\ee
As we mentioned in the previous section, all odd orders beyond $m=1$ will vanish
\be
\delta_\Lambda^{[2k+1]}\Psi=0,\qquad k\geq 1 
\ee
The results above are exact in the coordinates, so in general we can expand them to any order in $r$. We remark that we have 
\be
\left(\delta_\Lambda^{[m]}\Psi\right)^{(n)}=0\qquad\text{for}\  n<m \ ,
\ee
due to the fact that $\Psi$ starts at linear order in $r$.

For future use, let us also introduce the notation
\be 
\overset{j}{\delta}:= \sum_{k=0}^j \delta^{[k]}
\ee 
which gives the variation up to order $j$ in the Stueckelberg field.

In Bondi coordinates, the charge will be defined on the celestial sphere at $\Ib^+_-$, via 
\be 
\tilde{Q}_\Lambda = \int_{S^2} \tr (\Lambda \tilde{\F}_{ru}) r^2 \gamma dz d\zb \ ,
\ee 
with
\be 
\tilde{\F}_{ru}=e^{i\Psi}\F_{ru}e^{-i\Psi}
\ee
We can now define the charge aspect as
\be 
\begin{aligned}
\tilde{q}_\Lambda&=\tr (\Lambda r^2\tilde{\F}_{ru})^{(0)}\\
&=\tr (e^{-i\Psi}\Lambda e^{i\Psi}\F_{ru})^{(-2)}\\
&=\tr \left[\left(e^{ad_{-i\Psi}}\Lambda\right) \F_{ru}\right]^{(-2)}
\end{aligned}
\ee
The charge algebra now becomes
\be 
\begin{aligned}
\{ \tilde q_{{\Lambda}_1},\tilde q_{{\Lambda}_2}\}_*=&\frac{1}{2}\left[\delta_{{\Lambda}_1}\tilde q_{{\Lambda}_2}-(1\leftrightarrow2)\right]\\
=&-i\tilde{q}_{[{\Lambda}_1,{\Lambda}_2]}
\end{aligned}
\ee 
as needed. This is of course simpler than the expression in \autoref{subsec:Charge Algebra}, because we no longer have field-dependent parameters. See \autoref{charges} for a more detailed discussion of the charge construction and algebra at various orders.

The gauge choice will play an important role in our proof of the recursion relations for the equations of motion and charges at all orders, in the following sections. We will focus on two particular gauge choices. One is the light-cone gauge
\be 
a^\mu=(1,0,0,0)\quad \Rightarrow\quad \A_u=0
\ee 
together with a constraint to ensure radiative solutions
\be 
A_r^{(0)}=A_r^{(-1)}=0 \ .
\ee 
At this stage, one may choose to impose \eqref{extra_gauge_cond} , which results in\footnote{See \cite{Nagy:2022xxs} for an example of this in self-dual YM.} 
\be 
\Lambda_+ = \Lambda_+(r,z,\zb),\quad \Rightarrow \Psi=\Psi(r,z,\zb)
\ee 
though as we mentioned above this is not needed for our procedure, and we can simply work with a general
\be 
\Lambda_+ =\Lambda_+(u,r,z,\zb),\quad \Rightarrow \Psi=\Psi(u,r,z,\zb)
\ee 
The other gauge choice is the radial gauge, defined by 
\be 
\A_r=0
\ee 
where we additionally impose the usual constraint at $\I^+$
\be 
A_u^{(0)}=0 \ .
\ee 

\section{Equations of motion and recursion relations at all orders} \label{eom_scri_plus}

In this section we study in detail the properties of the radiative phase space, first providing the YM e.o.m. to all orders in the radial expansion in Bondi coordinates. Working in two different gauge choices, namely light-cone gauge and radial gauge, we give recursion relations that allow us to construct the fields to all orders. The equations in this section, together with all the expressions in \autoref{app_YM_eq_r_exp} and \autoref{YM e.o.m. in flat Bondi coordinates}, can be checked explicitly to arbitrarily high order using the Mathematica package developed by one of the authors, which will be presented in \cite{Giorgio:2025}.

We work with the Bondi coordinates \((u,r,z,\zb)\) defined in \eqref{bondi_coord_def1} and \eqref{bondi_coord_def2}. We want to study Yang-Mills theory at \(\mathcal{I}^+\), which in Bondi coordinates is obtained by taking the limit \(r\rightarrow\infty\) and it is parametrized by \((u,z,\bar{z})\). We will use the radial coordinate as the variable for the power expansions\footnote{We use Mathematica to compute most of the calculations regarding the expansions and manipulations with coefficients.}.

Let \(M := \R^{1,3}\) be the smooth manifold considered with coordinates defined \eqref{bondi_coord_def1} and \eqref{bondi_coord_def2}. The components of a generic \((k,l)\)-tensor \(\mathcal{T}\) on \(M\) can be expanded in a $r-$power series near \(\mathcal{I}^+\) as:
\begin{equation}
    \label{eq:generic tensor expansion}
    {\mathcal{T}^{\mu_1\dots\mu_k}}_{\nu_1\dots\nu_l}(u,r,z,\Bar{z})=\sum_{n\in\mathbb{Z}}\frac{{{T^{(-n)}}^{\mu_1\dots\mu_k}}_{\nu_1\dots\nu_l}(u,z,\Bar{z})}{r^n},
\end{equation}

  where each component ${\mathcal{T}^{\mu_1\dots\mu_k}}_{\nu_1\dots\nu_l}(u,r,z,\Bar{z})$ is a \(C^\infty\) function on \(M\) that can be written as a power series expansion in \(1/r\). We are interested in the subset $C^\infty_r(M)$ of functions with finite limit when $r \rightarrow +\infty$, defined as follows\footnote{Note that in the sum in \eqref{def:C^infty_r} we let \(n\) assume negative values and we kill all the coefficients \(f^{(-n)}(u,z,\bar{z})\) with \(n<0\). The reason why we do not define \(C^\infty_r(M)\) by summing only on \(n\in\mathbb{N}\) is that we want to construct a projector \(\mathcal{P}_{-n}\) that acts on \(C^\infty_r(M)\) with both positive and negative values of \(n\) (see \eqref{def:projector P}), so that equations like \eqref{prop:projector_1}-\eqref{prop:projector_3} are well defined. Thus, we want the set \eqref{def:C^infty_r} to take into account all \(n\in\mathbb{Z}\).}:
\begin{equation}
    \label{def:C^infty_r}
    \begin{aligned}
        C^\infty_r(M):=\{f\in C^\infty(M) \ | \ f(u,r,z,\bar{z})=\sum_{n\in\mathbb{Z}}r^{-n}f^{(-n)}(u,z,\Bar{z}), \ f^{(-n)}(u,z,\bar{z})=0 \ \forall n<0 \}
    \end{aligned}
\end{equation}

Therefore, given $f \in  C^\infty_r(M)$, the limit when $r\rightarrow +\infty$ is simply given by the smooth function $f^{(0)}$, which can be thought of as taking values in $\Ib^+$. We also assume that $f^{(n)} \in C^{\infty}(\Ib^+)$ for each $n \in \Z$.


\subsection{YM e.o.m. in Bondi coordinates}

The YM equations of motion are given by the vanishing of the following tensor, 
\begin{equation}
    \label{def:EOM}
    \mathcal{E}_\mu:=D^\nu\mathcal{F}_{\nu\mu}=\nabla^\nu\mathcal{F}_{\nu\mu}-i[\mathcal{A}^\nu,\mathcal{F}_{\nu\mu}].
\end{equation}
The vanishing of $\mathcal{E}_\mu$ implies in particular that $\mathcal{E}_\mu \in C^{\infty}_r (M)$ in \eqref{def:C^infty_r}, when expressed in Bondi coordinates $(u,r,z, \zb)$. Then, one can compute the components of $\mathcal{E}_\mu$ at each order \(1/r^n\), with \(n\in\mathbb{N}\). Let
\be
\mathcal{E}_\mu = \sum_{n\in\mathbb{N}}r^{-n}E_\mu^{(-n)}(u,z,\Bar{z})
\ee
denote the $r$ power expansion. A complete derivation of the following equations can be found in \autoref{app_YM_eq_r_exp}. Here, we present the result for each $E_\mu^{(-n)}$,
\begin{equation}
    \label{eq:YM EOM at every order in general Bondi}
    \begin{aligned}
    E_\mu^{(-n)}&=-\partial_{u}F^{(-n)}_{r\mu}+\tilde{\delta}_\mu^\rho(n-1)(F^{(1-n)}_{u\rho}-F^{(1-n)}_{r\rho})+2\partial_{(z}(\gamma^{-1}F^{(2-n)}_{\Bar{z})\mu})-2\bar{\delta}_\mu^\rho(\partial_{(z}\gamma^{-1})F^{(2-n)}_{\Bar{z})\rho} \\
    &\quad +i\sum_{k=0}^n\left(2[A^{(k-n)}_{(u},F^{(-k)}_{r)\mu}]-[A^{(k-n)}_r,F^{(-k)}_{r\mu}]-2\gamma^{-1}[A^{(k-n)}_{(z},F^{(2-k)}_{\Bar{z})\mu}]\right),
    \end{aligned}
\end{equation}
where,
\begin{align}
    \bar{\delta}_\mu^\nu&:=\mathrm{diag}(1,1,0,0), \\
    \label{def:delta_tilde}
    \tilde{\delta}_\mu^\nu(n)&:=n\delta_\mu^\nu-2\bar{\delta}_\mu^\nu.
\end{align}
More explicitly, each of the components of $E_\mu^{(-n)}$ is given by
\begin{align}
    E_u^{(-n)}&=\partial_{u}F^{(-n)}_{ur}+(n-3)F^{(1-n)}_{ur}+2\gamma^{-1}\partial_{(z}F^{(2-n)}_{\Bar{z})u} \label{eom_full_u}\\
    &\quad +i\sum_{k=0}^n\left([A^{(k-n)}_r-A^{(k-n)}_{u},F^{(-k)}_{ur}]-2\gamma^{-1}[A^{(k-n)}_{(z},F^{(2-k)}_{\Bar{z})u}]\right), \nonumber \\
    E_r^{(-n)}&=(n-3)F^{(1-n)}_{ur}+2\gamma^{-1}\partial_{(z}F^{(2-n)}_{\Bar{z})r} \label{eom_full_r}\\
    &\quad +i\sum_{k=0}^n\left([A^{(k-n)}_{r},F^{(-k)}_{ur}]-2\gamma^{-1}[A^{(k-n)}_{(z},F^{(2-k)}_{\Bar{z})r}]\right), \nonumber \\
    E_z^{(-n)}&=-\partial_{u}F^{(-n)}_{rz}+(n-1)(F^{(1-n)}_{uz}-F^{(1-n)}_{rz})-\partial_z(\gamma^{-1}F^{(2-n)}_{z\bar{z}}) \label{eom_full_z} \\
    &\quad +i\sum_{k=0}^n\left(2[A^{(k-n)}_{(u},F^{(-k)}_{r)z}]-[A^{(k-n)}_r,F^{(-k)}_{rz}]+\gamma^{-1}[A^{(k-n)}_{z},F^{(2-k)}_{z\Bar{z}}]\right), \nonumber\\
    E_{\bar{z}}^{(-n)}&=-\partial_{u}F^{(-n)}_{r\bar{z}}+(n-1)(F^{(1-n)}_{u\bar{z}}-F^{(1-n)}_{r\bar{z}})+\partial_{\bar{z}}(\gamma^{-1}F^{(2-n)}_{z\bar{z}}) \label{eom_full_zb}\\
    &\quad +i\sum_{k=0}^n\left(2[A^{(k-n)}_{(u},F^{(-k)}_{r)\bar{z}}]-[A^{(k-n)}_r,F^{(-k)}_{r\bar{z}}]-\gamma^{-1}[A^{(k-n)}_{\bar{z}},F^{(2-k)}_{z\Bar{z}}]\right). \nonumber
\end{align}

In the following, we study these equations for two particular gauge choices, light-cone and radial gauges. First, we have to define a subset of solutions, given by certain fall-offs compatible with radiative fields. Then, we compute the equations of motion ($E_\mu^{(-n)} = 0$), and finally we deduce a recursive formula that allows to compute the components of the field strength up to functions on the sphere $S^2$ with coordinates $(z,\zb)$, i.e., the celestial sphere. Schematically, one has 
\be 
F_{ur}^{(-n)} = G[\partial_u^{-1},\gamma,\partial_z,\partial_{\zb}, A_z^{(0)}, A_{\zb}^{(0)} , F_{ur}^{(-n+1)},..., F_{ur}^{(-2)}], 
\ee
where $\partial_u^{-1} f(u , z ,\zb) := \int_{-\infty}^{u} f(\tilde{u},z,\zb) d\tilde{u}$ and $G$ is a functional that contains the differential operators $\partial_u^{-1},\partial_z,\partial_{\zb}$ and depends on the field strength via $F_{ur}^{(-n+1)},..., F_{ur}^{(-2)}$.

\subsection{Light-cone gauge }

Our first case study is the light-cone gauge. This gauge is characterised by the vanishing of the $u$-component of the gauge vector,
\be \label{light_cone_gauge}
    \mathcal{A}_u=0.
\ee 
As we mentioned above, we are interested in a subset of solutions that are radiative. This can be achieved by taking the following conditions on the $r$-component of $\mathcal{A}_\mu$,
\begin{equation}
    \label{light_cone_falloff}
     \qquad A_r^{(0)}=A_r^{(-1)}=0.
\end{equation}
From \eqref{light_cone_gauge} and \eqref{light_cone_falloff} we get the following conditions on $\mathcal{F}_{\mu \nu}$,
\begin{equation}
    F_{ur}^{(-n)}=F_{rz}^{(-n)}=F_{r\bar{z}}^{(-n)}=0 \qquad \mathrm{for} \ n=0,1.
\end{equation}
Equations \eqref{eom_full_u}-\eqref{eom_full_zb} simplify to the following set of identities,
\begin{align}
    \label{eq:EOM full Bondi light-cone u component}
    E_u^{(-n)}&=\partial_{u}F^{(-n)}_{ur}+(n-3)F^{(1-n)}_{ur}+2\gamma^{-1}\partial_{(z}F^{(2-n)}_{\Bar{z})u} \\
    &\quad +i\sum_{k=0}^{n-2}\left([A^{(k-n)}_r,F^{(-k)}_{ur}]-2\gamma^{-1}[A^{(2+k-n)}_{(z},F^{(-k)}_{\Bar{z})u}]\right), \qquad n\geq2 \nonumber \\
    \label{eq:EOM full Bondi light-cone r component}
    E_r^{(-n)}&=(n-3)F^{(1-n)}_{ur}+2\gamma^{-1}\partial_{(z}F^{(2-n)}_{\Bar{z})r} \\
    &\quad +i\sum_{k=2}^{n-2}\left([A^{(k-n)}_{r},F^{(-k)}_{ur}]-2\gamma^{-1}[A^{(2+k-n)}_{(z},F^{(-k)}_{\Bar{z})r}]\right), \qquad n\geq4 \nonumber \\
    \label{eq:EOM full Bondi light-cone z component}
    E_z^{(-n)}&=-\partial_{u}F^{(-n)}_{rz}+(n-1)(F^{(1-n)}_{uz}-F^{(1-n)}_{rz})-\partial_z(\gamma^{-1}F^{(2-n)}_{z\bar{z}}) \\
    &\quad +i\sum_{k=0}^{n-2}\left([A^{(k-n)}_{r},F^{(-k)}_{uz}-F^{(-k)}_{rz}]+\gamma^{-1}[A^{(2+k-n)}_{z},F^{(-k)}_{z\Bar{z}}]\right), \qquad n\geq2 \nonumber\\
    \label{eq:EOM full Bondi light-cone zb component}
    E_{\bar{z}}^{(-n)}&=-\partial_{u}F^{(-n)}_{r\bar{z}}+(n-1)(F^{(1-n)}_{u\bar{z}}-F^{(1-n)}_{r\bar{z}})+\partial_{\bar{z}}(\gamma^{-1}F^{(2-n)}_{z\bar{z}}) \\
    &\quad +i\sum_{k=0}^{n-2}\left([A^{(k-n)}_{r},F^{(-k)}_{u\bar{z}}-F^{(-k)}_{r\bar{z}}]-\gamma^{-1}[A^{(2+k-n)}_{\bar{z}},F^{(-k)}_{z\Bar{z}}]\right), \qquad n\geq2 \nonumber
\end{align}

\subsubsection{Recursive formula for the gauge vector} \label{subsubsec:recursive_formula_A_light_cone_full}
In what follows we prove that given the initial data 
\be \label{in_data}
A_z^{(0)} (u,z,\zb),\quad A_{\bar{z}}^{(0)}(u,z,\zb) \ , 
\ee
and families of functions on the sphere 
\be \label{not_r_u_exp}
\{A_z^{(-n,0)},A_{\zb}^{(-n,0)},A_r^{(-n-2,0)} \in C^\infty(S^2)\}_{n\geq 1}  \ ,
\ee
where $T^{(-n,k)}$ denotes the coefficient of $\frac{u^k}{r^n}$ in a formal expansion of the tensor $T$ in $r$ and $u$, it is possible to find the solution of the equation of motion at arbitrary order in the $1/r$-expansion. In particular, we present a recursive formula that allows us to compute $A_\mu^{(-n)}$ for all $n\in\mathbb{N}$, thus providing the explicit expression of the solution of the e.o.m. at each order.

First of all, we write explicitly the definition of the field strength tensor in light-cone gauge \eqref{light_cone_gauge} and with fall-off \eqref{light_cone_falloff} at a given order $n$:
\begin{align}
    \label{F_ur in light-cone full Bondi}
    F_{ur}^{(-n)}&=\partial_uA_r^{(-n)}, \qquad n\geq2 \\
    \label{F_uz in light-cone full Bondi}
    F_{uz}^{(-n)}&=\partial_uA_z^{(-n)} , \qquad n\geq2 \\
    \label{F_rz in light-cone full Bondi}
    F_{rz}^{(-n)}&=(1-n)A_z^{(1-n)}-\partial_zA_r^{(-n)}-i\sum_{k=0}^{n-2}[A_r^{(k-n)},A_z^{(-k)}], \qquad n\geq2 \\
    \label{F_zzb in light-cone full Bondi}
    F_{z\bar{z}}^{(-n)}&=2\partial_{[z}A_{\bar{z}]}^{(-n)}-i\sum_{k=0}^{n}[A_z^{(k-n)},A_{\bar{z}}^{(-k)}] , \qquad n\geq2 \ ,
\end{align}
where $F_{u\zb}^{(-n)}$ and $F_{r\zb}^{(-n)}$ are completely analogous to  \eqref{F_uz in light-cone full Bondi} and \eqref{F_rz in light-cone full Bondi}, respectively. Recall that $F_{ur}^{(-n)}=F_{rz}^{(-n)}=0$ for $n=0,1$. From the above equations we notice that
\begin{align}
    \label{F_umu dependence in light-cone full Bondi}
    F_{u\mu}^{(-n)}& \quad \mathrm{depends} \ \mathrm{on} \quad A_\mu^{(-n)} \quad \text{for} \quad \mu = r,z,\zb,\\
    \label{F_rz dependence in light-cone full Bondi}
    F_{rz}^{(-n)}& \quad \mathrm{depends} \ \mathrm{on} \quad \{A_r^{(-k)}\}_{k\leq n},\{A_z^{(-k)}\}_{k<n}, \\
    \label{F_zzb dependence in light-cone full Bondi}
    F_{z\bar{z}}^{(-n)}& \quad \mathrm{depends} \ \mathrm{on} \quad \{A_z^{(-k)},A_{\bar{z}}^{(-k)}\}_{k\leq n}.
\end{align}

Suppose we know $\{A_z^{(-k)}(u,z,\zb),A_{\zb}^{(-k)}(u,z,\zb),A_r^{(-k-2)}(u,z,\zb)\}_{k<n}$ for some arbitrary $n\in\mathbb{N}$; we want to find $\{A_z^{(-n)}(u,z,\zb),A_{\zb}^{(-n)}(u,z,\zb),A_r^{(-n-2)}(u,z,\zb)\}$. Although it is useful to express the recursion step as above, we emphasize that $A_r^{(-2)}$ is not part of the initial data, but can be obtained from \eqref{in_data} and \eqref{not_r_u_exp}.


Indeed, $A_r^{(-2)}$ can be obtained from \eqref{F_ur in light-cone full Bondi} and \eqref{eq:EOM full Bondi light-cone u component} with $n=2$:
\begin{align}  \partial_{u}F^{(-2)}_{ur}&=2\gamma^{-1}\partial_{(z}F^{(0)}_{u|\Bar{z})}+2i\gamma^{-1}[A^{(0)}_{(z},F^{(0)}_{\Bar{z})u}] ,\\
    \label{eq:A_r(-2) in light-cone gauge full Bondi}
    \Longrightarrow A_r^{(-2)}&=2\gamma^{-1}\partial_u^{-1}\left(\partial_{(z}A_{\bar{z})}^{(0)}+i\partial_u^{-1}[\partial_uA_{(z}^{(0)},A_{\bar z)}^{(0)}]\right),
\end{align}
where in the second line we used \eqref{F_uz in light-cone full Bondi} and the symbol $\partial_u^{-1} = \int du$ defined previously. For $n\geq3$ we use \eqref{F_ur in light-cone full Bondi} and \eqref{eq:EOM full Bondi light-cone r component} to find:
\begin{align}
    \label{eq:A_r(-n) in light-cone gauge full Bondi}
    A_r^{(-n)}&=\frac{1}{2-n}\partial_u^{-1}\left[2\gamma^{-1}\partial_{(z}F_{\bar{z})r}^{(1-n)}+i\sum_{k=2}^{n-1}\left([A_r^{(k-n-1)},F_{ur}^{(-k)}]-2\gamma^{-1}[A_{(z}^{(1+k-n)},F_{\bar{z})r}^{(-k)}]\right)\right]
\end{align}
By observing the equations above and recalling \eqref{F_umu dependence in light-cone full Bondi}-\eqref{F_rz dependence in light-cone full Bondi} we notice that, in general:
\begin{align}
    \label{A_r(-n) dependence in light-cone full Bondi}
    A_r^{(-n)}& \quad \mathrm{depends} \ \mathrm{on} \quad A_r^{(-n+1)},\{A_\mu^{(-k)}\}_{k\leq n-2}\quad \text{and}\quad
    A_r^{(-n,0)}
\end{align}
Therefore, we conclude that knowing $\{A_\mu^{(-k)}\}_{k<n}$ is sufficient to determine $A_r^{(-n)}$ via equations \eqref{eq:A_r(-2) in light-cone gauge full Bondi} and \eqref{eq:A_r(-n) in light-cone gauge full Bondi}. Each function $\{A_r^{(-n,0)}\}_{n\geq 2}$ is the \textit{constant of integration} of the operator $\partial_u^{-1}$. 

Now, we solve for $A_z^{(-n)}$. At leading order, that is, for $n=0$, the $z$ and $\bar{z}$ components of the gauge field constitute the initial data, so we consider $n\geq1$. From \eqref{F_rz in light-cone full Bondi} we have:
\begin{align}
    \label{eq:A_z solution step 1}
    A_z^{(-n)}&=-\frac{1}{n}\left(\partial_zA_r^{(-n-1)}-i\sum_{k=2}^{n+1}[A_z^{(k-n-1)},A_r^{(-k)}]+F_{rz}^{(-n-1)}\right).
\end{align}

However, thanks to \eqref{F_rz dependence in light-cone full Bondi} we know that, in general, $F_{rz}^{(-n-1)}$ does depend on $A_z^{(-n)}$. Hence, we need an explicit expression for it, which is provided by the $z$ component \eqref{eq:EOM full Bondi light-cone z component} of the e.o.m.
\begin{equation}
    \label{eq:A_z solution step 2}
    \begin{aligned}
        F^{(-n-1)}_{rz}&=nA_z^{(-n)}-\partial_u^{-1}\bigg[nF^{(-n)}_{rz}+\partial_z(\gamma^{-1}F^{(1-n)}_{z\bar{z}}) \\
        &\quad -i\sum_{k=0}^{n-1}\left([A^{(k-n-1)}_{r},\partial_uA_z^{(-k)}-F^{(-k)}_{rz}]+\gamma^{-1}[A^{(1+k-n)}_{z},F^{(-k)}_{z\Bar{z}}]\right)\bigg] .
    \end{aligned}
\end{equation}
Putting \eqref{eq:A_z solution step 1} and \eqref{eq:A_z solution step 2} together we obtain (for $n\geq1$):
\begin{equation}
    \begin{aligned}
        A_z^{(-n)}&=\frac{1}{2n}\bigg\{-\partial_zA_r^{(-n-1)}+i\sum_{k=2}^{n+1}[A_z^{(k-n-1)},A_r^{(-k)}] \label{eq:A_z(-n) in light-cone gauge full Bondi} \\
        &\quad +\partial_u^{-1}\bigg[nF^{(-n)}_{rz}+\partial_z(\gamma^{-1}F^{(1-n)}_{z\bar{z}}) \\
        &\quad -i\sum_{k=0}^{n-1}\left([A^{(k-n-1)}_{r},\partial_uA_z^{(-k)}-F^{(-k)}_{rz}]+\gamma^{-1}[A^{(1+k-n)}_{z},F^{(-k)}_{z\Bar{z}}]\right)\bigg]\bigg\}.
    \end{aligned}
\end{equation}
The above equation together with the \eqref{light_cone_falloff}, \eqref{F_rz dependence in light-cone full Bondi}, \eqref{F_zzb dependence in light-cone full Bondi} and \eqref{A_r(-n) dependence in light-cone full Bondi} imply that:
\begin{align}
    A_z^{(-n)}& \quad \mathrm{depends} \ \mathrm{on} \quad A_r^{(-n-1)}, A_r^{(-n)},\{A_\mu^{(-k)}\}_{k<n}.
\end{align}

For instance, $A_z^{(-1)}$ depends just on the initial data:
\begin{equation}
    \label{A_z-1_full_Bondi}
    \begin{aligned}
        A_z^{(-1)}&=-\frac{1}{2}\left(\partial_zA_r^{(-2)}-i[A_z^{(0)},A_r^{(-2)}]\right) \\
        &\quad +\frac{1}{2}\partial_u^{-1}\left(\partial_z(\gamma^{-1}F^{(0)}_{z\bar{z}})-i[A^{(-2)}_{r},\partial_uA^{(0)}_{z}]-i\gamma^{-1}[A^{(0)}_{z},F^{(0)}_{z\Bar{z}}]\right).
    \end{aligned}
\end{equation}

Observe that since we start with the fall-offs \eqref{light_cone_falloff}, the inductive step is indeed satisfied. In other words, the solution for $A_z^{(-n)}$ is completely determined by $ A_r^{(-n-1)}, A_r^{(-n)}$, $\{ A_\mu^{(-k)} \}_{k<n}$ and the functions on the sphere $\{A_z^{(-n,0)}(z,\zb)\}$. Note that it is also possible to find $A_z^{(-n)}$ by integrating $F_{uz}^{(-n)}$, which in turn can be obtained through the Bianchi identity \eqref{eq:Bianchi light-cone full bondi 1}, achieving the same result.

In \autoref{fig:recursive_indices}, the relation between the indices is schematically depicted.

\begin{figure}[h]
    \centering
    \includegraphics[scale=0.7]{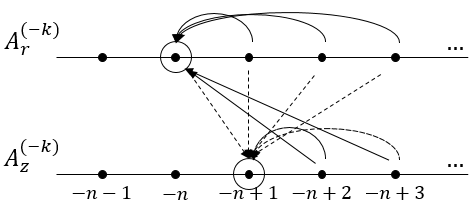}
    \caption{The solid arrows indicate which $A_r^{(-k)}$, $A_z^{(-k)}$ enter the recursion formula for $A_r^{(-n)}$, while the dotted arrows indicate which $A_r^{(-k)}$ and $A_z^{(-k)}$ enter the recursion formula for $A_z^{(-n)}$. The component $A_{\zb}^{(-n)}$ is not drawn, since the same result as 
    for $A_z^{(-n)}$ is valid.}
    \label{fig:recursive_indices}
\end{figure}

To conclude, the above procedure allows us to compute the solution to the e.o.m. $A_\mu^{(-n)}$ at an arbitrary order $n\in\mathbb{N}$ of the large $r$-expansion, by recursively constructing $A_\mu^{(-k)}$ for $k=0,\dots,n$, starting from the initial data $A_z^{(0)}$ and $A_{\bar{z}}^{(0)}$. Our recursive relation can be written in a more compact way as follows:
\begin{align}
A_r^{(-2)}&=2\gamma^{-1}\partial_u^{-2}\left(D_{(z|}\partial_u\mathcal{A}_{|\bar{z})}\right)^{(0)} \\
    A_r^{(-n)}&=\frac{1}{2-n}\partial_u^{-1}\left[2\gamma^{-1}\left(D_{(z}\mathcal{F}_{\bar{z})r}\right)^{(1-n)}+i\left([\mathcal{A}_r,\partial_u\mathcal{A}_r]\right)^{(-n-1)}\right], \qquad n\geq3 \\
    A_z^{(-n)}&=-\frac{1}{2n}\big\{(D_z\mathcal{A}_r)^{(-n-1)}-\partial_u^{-1}\big(nF_{rz}^{(-n)}+(D_z(\gamma^{-1}\mathcal{F}_{z\bar{z}}))^{(1-n)} \\
    &\quad -i([\mathcal{A}_r,\partial_u\mathcal{A}_z-\mathcal{F}_{rz}])^{(-n-1)}\big)\big\}, \qquad n\geq1 \nonumber
\end{align}

\textbf{Recursive formula for $F_{ur}^{(-n)}$:} As will be seen in \autoref{charges}, the $ru$ component of the field strength plays a crucial role in the construction of the charges. Thus it is useful to give a recursive formula for the different coefficients of the $r$-expansion of $F_{ur}^{(-n)}$, with $n\geq 2$.

Consider $\{F_{ur}^{(-n,0)} \in C^\infty(S^2)\}_{n\geq 0}$, a family of functions on the sphere, then, by \eqref{eq:EOM full Bondi light-cone u component}, we can obtain $F_{ur}^{-n}$ in terms of $F_{ur}^{(-n,0)}$ and $\{F_{ur}^{-k}\}_{k<n}$, by using the $\partial_u^{-1}$ operator.

It immediately follows from \eqref{eq:EOM full Bondi light-cone u component} that
\begin{equation}
    \label{F_ur_recursive_int}
    \begin{aligned}
        F_{ur}^{(-n)}&= \partial_u^{-1} \bigg( (3-n)F^{(-n+1)}_{ur} -2 \gamma^{-1} \partial_{(z} F^{(-n+2)}_{\Bar{z})u} \\
        & \quad -i\sum_{k=0}^{n-2}\left([A^{(k-n)}_r,F^{(-k)}_{ur}]  -2 \gamma^{-1} [A^{(2+k-n)}_{(z},F^{(-k)}_{\Bar{z})u}]\right) \bigg).
    \end{aligned}
\end{equation}
The constant of integration for the operator $\partial_u^{-1}$ is given by $F_{ur}^{(-n,0)}$, which we can introduce by fixing the lower limit in the anti-derivative as $-\infty$,

\begin{equation}
    \label{F_ur_recursive_int}
    \begin{aligned}
        F_{ur}^{(-n)}(u,z,\zb)&= F_{ur}^{(-n,0)} (z,\zb) + \int_{-\infty}^u du \bigg( (3-n)F^{(-n+1)}_{ur} -2 \gamma^{-1} \partial_{(z} F^{(-n+2)}_{\Bar{z})u} \\
        & \quad -i\sum_{k=0}^{n-2}\left([A^{(k-n)}_r,F^{(-k)}_{ur}]  -2 \gamma^{-1} [A^{(2+k-n)}_{(z},F^{(-k)}_{\Bar{z})u}]\right) \bigg).
    \end{aligned}
\end{equation}

\subsubsection{Decays on null infinity in the radiative phase space} \label{u_fall_offs_light_cone_gauge}

As it is well known in the sub$^n$-theorems literature (see for example \cite{Hamada:2018vrw, Li:2018gnc, Campiglia:2018dyi}, the $u$-behaviour of the fields translates as derivatives via a Fourier transform in $u$. Let us assume that we are working at tree-level. This  implies that the decay of the $u$-derivative of the initial data ($\partial_u A_z^{(0)}$) is faster than any polynomial in $u$. 

Let $\mathfrak{F}(A^{(0)}_z)(\omega, z,\zb)$ be the time-Fourier transform of $A_z(u, z,\zb)$. Then,  the $\omega \rightarrow 0$ expansion of the tree-level amplitude (see \cite{Hamada:2018vrw, Li:2018gnc}) implies the following fall off in $\omega$ for the radiative free data for the  gauge vector,
\be 
\mathfrak{F}(A^{(0)}_z)(\omega, \hat{x})=\omega^{-1} \mathfrak{F}(A^{(0)}_z)^{-1}(\hat{x})+\mathfrak{F}(A^{(0)}_z)^{0} (\hat{x})+\omega \mathfrak{F}(A^{(0)}_z)^{1}(\hat{x})+\ldots,
\ee

where the upper index outside the bracket indicates the order in the $\omega$-expansion. Undoing the Fourier transform, $A_z(u, z,\zb)$ behaves like
\be 
A^{(0)}_z(u, \hat{x})=\theta(u) A_z^{(\theta)}(\hat{x})+\delta(u) A_z^{(\delta)}(\hat{x})+\delta^{\prime}(u) A_z^{\left(\delta^{\prime}\right)}(\hat{x})+\ldots,
\ee

with $\theta$ a step function. Thus, the $u \rightarrow \pm \infty$ falloffs for the null infinity free data are as follows
\be 
A^{(0)}_z(u, \hat{x}) \stackrel{u \rightarrow \pm \infty}{=} A^{(0,0) \pm}_z (\hat{x})+ o(u^{-\infty}),
\ee
where $o(u^{-\infty})$ is a remainder that goes to zero faster than any power of $1 / u$. 

In what follows we compute the $u$-behaviour of the functions $A_r^{(-n)}, A_z^{(-n)}$ and $F_{ru}^{(-n)}$ as $u \rightarrow -\infty$. Let us start with the the following boundary conditions on the field strength at $\Ib^+_+$,
\be 
\lim_{u \rightarrow + \infty} F^{(-2)}_{ru} (u,z,\zb)= 0 , \quad \lim_{u \rightarrow + \infty} F_{rz}^{(-2)}(u,z,\zb) = 0 , \quad \lim_{u \rightarrow + \infty} F_{r\zb}^{(-2)}(u,z,\zb) = 0 ,
\ee
which arise due to the absence of massive colored fields. The standard $u \rightarrow -\infty$ fall-off for the radiative data $A_z^{(0)}$ is the following,
\be 
A_z^{(0)} (u,z,\zb)= A_z^{(0,0)}  (z,\zb) + o(u^{-\infty}).
\ee 
This translates to $F_{ru}^{(-2)}$ via \eqref{eq:A_r(-2) in light-cone gauge full Bondi} as follows
\be 
F_{ru}^{(-2)} (u,z,\zb) = F_{ru}^{(-2,0)}(z,\zb) + o(u^{-\infty}).
\ee
We can write
\be 
\partial_u F_{ru}^{(-2)} = 2 \gamma^{-1} \partial_{(z} F_{\zb) u}^{(0)} -2 i \gamma^{-1} [A_{(z}^{(0)} , F_{\zb) u}^{(0)}] =: 2 \gamma^{-1} D^{(0)}_{(z} F_{\zb) u}^{(0)},
\ee
where $D^{(0)}_z :=  \partial_z - i [A_{z}^{(0)}, \cdot]  $. Using the results in the previous subsections, we will show here how the integration from $A_z^{(0)}(u,z,\zb)$ can be carried out to all the coefficients (with the corresponding \textit{constants of integration} as functions on the sphere).

First, observe that $A_r^{(-2)} = O(u)$, since $\partial_u A_r^{(-2)} + F_{ru}^{(-2)} =0$. Then, if $A_r^{(-2 ,0)}$ is the finite value of the limit  
\be 
\lim_{u \rightarrow -\infty} A_r^{(-2)} + u F_{ru}^{(-2,0)} = A_r^{(-2 ,0)}(z,\zb),
\ee
where the second term on the LHS above has the role of substracting the linear dependence in $u$, then
\be 
A_r^{(-2)}(u,z, \zb) = - u F_{ru}^{(-2,0)}(z, \zb) + A_r^{(-2 ,0)}(z, \zb) + o(u^{-\infty}).
\ee
This is consistent with \eqref{eq:A_r(-2) in light-cone gauge full Bondi}. Next, the equation for $A_z^{(-1)}$ is given by \eqref{A_z-1_full_Bondi},
\beq \label{A_z_-1_computed}
A_z^{(-1)}&=& \frac{1}{2} u \left(\partial_z F_{ru}^{(-2,0)} -i[A_z^{(0,0)},F_{ru}^{(-2,0)}] +\partial_z(\gamma^{-1}F^{(0,0)}_{z\bar{z}})-i\gamma^{-1}[A^{(0,0)}_{z},F^{(0,0)}_{z\Bar{z}}]\right)   \\
&& -\frac{1}{2}\left(\partial_zA_r^{(-2,0)}-i[A_z^{(0,0)},A_r^{(-2,0)}]\right) + A_z^{(-1,0)} + o(u^{-\infty}). \nonumber
\eeq

Observe that $A_z^{(-1)}$ contains up to linear terms in $u$. The coefficients $F_{zu}^{(0)}$, $F_{zu}^{(-1)}$ and $F_{zr}^{(-2)}$ are found directly 
\begin{align}
F_{zu}^{(0)} & = - \partial_u A_z^{(0)} = o(u^{-\infty}) , \quad F_{zu}^{(-1)} = - \partial_u A_z^{(-1)} = O(1), \label{F_zu_0}\\
F_{zr}^{(-2)} & = \partial_z A_r^{(-2)} - A_z^{(-1)} -i [A_z^{(0)}, A_r^{(-2)}] = O(u).
\end{align}

The second equation is consistent with the definition of $A_z^{(-1)}$ via \eqref{eq:A_z solution step 1},
\be 
\partial_u F^{(-2)}_{rz} =\partial_u A_z^{(-1)}-F^{(-n)}_{rz}+\partial_z(\gamma^{-1}F^{(0)}_{z\bar{z}}) -i\left([A^{(-2)}_{r},\partial_uA_z^{(0)}-F^{(0)}_{rz}]+\gamma^{-1}[A^{(0)}_{z},F^{(0)}_{z\Bar{z}}]\right) .
\ee

From \eqref{F_ur_recursive_int} we have
\beq 
F_{ur}^{(-3)} &=& -\partial_u^{-1} \left(2 \gamma^{-1} D^{(0)}_{(z} F^{(-1)}_{\Bar{z})u} -2i \gamma^{-1} [A^{(-1)}_{(z},F^{(0)}_{\Bar{z})u}] \right) \label{F_ru_3}
\eeq

From equations \eqref{F_zu_0} and \eqref{F_ru_3}, we deduce that $F_{ru}^{(-3)}$ is $O(u)$. Indeed, the first and third terms within the brackets are $O(1)$ in $u$, while the second term contains a product of a linear term and $\partial_u A_z^{(0)}$, which is $o(u^{-\infty})$; therefore, after integration, the second term is $O(1)$. Then, we can proceed in an analogous way, since now we have $A_r^{(-3)} = O(u^2)$. Therefore, inductively from \eqref{eq:A_r(-n) in light-cone gauge full Bondi}, \eqref{eq:A_z(-n) in light-cone gauge full Bondi} and \eqref{F_ur_recursive_int} we have that each $F_{ru}^{(-n-2)}$ will be $O(u^n)$.

\subsection{Radial gauge}
\label{subsec:Radial gauge}

Next we study the radial gauge, characterized by
\begin{equation}
    \label{radial_gauge}
    \mathcal{A}_r=0\;,
\end{equation}
and let us consider the fall-off
\be 
\label{radial_fall_off}
A_u^{(0)}=0\;,
\ee 
which further specifies the gauge. Thus, this choice implies the following identities:
\begin{equation}
    F_{ur}^{(-n)}=F_{rz}^{(-n)}=F_{r\bar{z}}^{(-n)}=0 \qquad \mathrm{for} \ n=0,1.
\end{equation}
Equations \eqref{eom_full_u}-\eqref{eom_full_zb} result in
\begin{align}
    \label{eq:EOM full Bondi radial u component}
    E_u^{(-n)}&=\partial_{u}F^{(-n)}_{ur}+(n-3)F^{(1-n)}_{ur}+2\gamma^{-1}\partial_{(z}F^{(2-n)}_{\Bar{z})u} \\
    &\quad -i\sum_{k=0}^{n-1}[A^{(k-n)}_{u},F^{(-k)}_{ur}]-2i\gamma^{-1}\sum_{k=0}^{n-2}[A^{(2+k-n)}_{(z},F^{(-k)}_{\Bar{z})u}], \qquad n\geq2 \nonumber \\
    \label{eq:EOM full Bondi radial r component}
    E_r^{(-n)}&=(n-3)F^{(1-n)}_{ur}+2\gamma^{-1}\partial_{(z}F^{(2-n)}_{\Bar{z})r}-2i\gamma^{-1}\sum_{k=2}^{n-2}[A^{(2+k-n)}_{(z},F^{(-k)}_{\Bar{z})r}], \qquad n\geq4 \\
    \label{eq:EOM full Bondi radial z component}
    E_z^{(-n)}&=-\partial_{u}F^{(-n)}_{rz}+(n-1)(F^{(1-n)}_{uz}-F^{(1-n)}_{rz})-\partial_z(\gamma^{-1}F^{(2-n)}_{z\bar{z}}) \\
    &\quad +i\sum_{k=0}^{n-1}[A^{(k-n)}_{u},F^{(-k)}_{rz}]+i\gamma^{-1}\sum_{k=0}^{n-2}[A^{(2+k-n)}_{z},F^{(-k)}_{z\Bar{z}}], \qquad n\geq2 \nonumber\\
    \label{eq:EOM full Bondi radial zb component}
    E_{\bar{z}}^{(-n)}&=-\partial_{u}F^{(-n)}_{r\bar{z}}+(n-1)(F^{(1-n)}_{u\bar{z}}-F^{(1-n)}_{r\bar{z}})+\partial_{\bar{z}}(\gamma^{-1}F^{(2-n)}_{z\bar{z}}) \\
    &\quad +i\sum_{k=0}^{n-1}[A^{(k-n)}_{u},F^{(-k)}_{r\bar{z}}]-i\gamma^{-1}\sum_{k=0}^{n-2}[A^{(2+k-n)}_{\bar{z}},F^{(-k)}_{z\Bar{z}}], \qquad n\geq2\;. \nonumber
\end{align}

Observe the almost symmetric behaviour of $r$ and $u$ with respect to the previous case (eqn. \eqref{eq:EOM full Bondi light-cone u component}-\eqref{eq:EOM full Bondi light-cone zb component}), although there are differences which we will see below.

\subsubsection{Recursive formula for the gauge vector} \label{subsubsec:recursion_radial_full_bondi}

As in the previous section, we will prove that given the initial data $A_z^{(0)} (u,z,\zb) ,A_{\zb}^{(0)} (u,z,\zb)$ and $\{A_z^{(-n,0)},A_{\zb}^{(-n,0)},A_u^{(-n-1,0)}\}_{n\geq 1}$, it is possible to find the solution of the equation of motion at arbitrary order in the $1/r$-expansion. We present recursive formulas (similar to \eqref{eq:A_r(-n) in light-cone gauge full Bondi},\eqref{eq:A_z(-n) in light-cone gauge full Bondi}) that allow us to compute $A_\mu^{(-n)}$ for all $n\geq 1$, thus providing the explicit expression of the e.o.m. solution at each order.

First of all, it is useful to write explicitly the definition of the field strength tensor in radial gauge. From \eqref{radial_gauge} and \eqref{radial_fall_off}, at a given order $n$ we have
\begin{align}
    \label{F_ur in radial full Bondi}
    F_{ur}^{(-n)}&=(n-1)A_u^{(1-n)}, \qquad n\geq2, \\
    \label{F_uz in radial full Bondi}
    F_{uz}^{(-n)}&=2\partial_{[u}A_{z]}^{(-n)}-i\sum_{k=0}^{n-1}[A_u^{(k-n)},A_z^{(-k)}]=\partial_uA_z^{(-n)}-(D_z A_u)^{(-n)}, \\
    \label{F_rz in radial full Bondi}
    F_{rz}^{(-n)}&=(1-n)A_z^{(1-n)}, \qquad n\geq2, \\
    \label{F_zzb in radial full Bondi}
    F_{z\bar{z}}^{(-n)}&=2\partial_{[z}A_{\bar{z}]}^{(-n)}-i\sum_{k=0}^{n}[A_z^{(k-n)},A_{\bar{z}}^{(-k)}].
\end{align}
Recall that $F_{ur}^{(-n)}=F_{rz}^{(-n)}=0$ for $n=0,1$. From the above equations we notice that:
\begin{align}
    \label{F_rmu dependence in radial full Bondi}
    F_{r\mu}^{(-n)}& \quad \mathrm{depends} \ \mathrm{on} \quad A_\mu^{(1-n)}, \qquad (n\geq2) ,\\
    \label{F_rz dependence in radial full Bondi}
    F_{uz}^{(-n)}& \quad \mathrm{depends} \ \mathrm{on} \quad \{A_u^{(-k)},A_z^{(-k)}\}_{k\leq n}, \\
    \label{F_zzb dependence in radial full Bondi}
    F_{z\bar{z}}^{(-n)}& \quad \mathrm{depends} \ \mathrm{on} \quad \{A_z^{(-k)},A_{\bar{z}}^{(-k)}\}_{k\leq n}.
\end{align}

Suppose we know $\{A_z^{(-k)}(u,z,\zb),A_{\zb}^{(-k)}(u,z,\zb),A_u^{(-k-1)}(u,z,\zb)\}_{k<n}$ for some arbitrary $n\in\mathbb{N}$; we want to find $\{A_z^{(-n)}(u,z,\zb),A_{\zb}^{(-n)}(u,z,\zb),A_r^{(-n-2)}(u,z,\zb)\}$. First, we solve for the $u$ component of the gauge field. At leading order we know that $A_u^{(0)}$ vanishes. $A_u^{(-1)}$ can be obtained from \eqref{eq:EOM full Bondi radial u component} and \eqref{F_ur in radial full Bondi} with $n=2$:
\begin{align}
    \partial_{u}F^{(-2)}_{ur}&=-2\gamma^{-1}\partial_{(z}F^{(0)}_{\Bar{z})u}+2i\gamma^{-1}[A^{(0)}_{(z},F^{(0)}_{\Bar{z})u}], \\
    \label{eq:A_u(-1) in radial gauge full Bondi}
    \Longrightarrow A_u^{(-1)}&=2\gamma^{-1}\left(\partial_{(z}A_{\bar{z})}^{(0)}+i\partial_u^{-1}[\partial_uA_{(z}^{(0)},A_{\bar z)}^{(0)}]\right),
\end{align}
where in the second line we used $F_{uz}^{(0)}=\partial_uA_z^{(0)}$. For $n\geq2$ we make use of \eqref{eq:EOM full Bondi radial r component}, \eqref{F_ur in radial full Bondi} and \eqref{F_rz in radial full Bondi} to find
\begin{align}
    \label{eq:A_u(-n) in radial gauge full Bondi}
    A_u^{(-n)}&=-\frac{2\gamma^{-1}}{n}\left(\partial_{(z}A_{\Bar{z})}^{(1-n)}+\sum_{k=1}^{n-1}\frac{ik}{1-n}[A^{(1+k-n)}_{(z},A_{\Bar{z})}^{(-k)}]\right).
\end{align}
We notice that, in general:
\begin{align}
    \label{A_u(-n) dependence in radial full Bondi}
    A_u^{(-n)}& \quad \mathrm{depends} \ \mathrm{on} \quad \{A_z^{(-k)},A_{\bar{z}}^{(-k)}\}_{k<n} \quad \text{and} \quad A_u^{(-n,0)}.
\end{align}
Therefore, we conclude that knowing $\{A_z^{(-k)} , A_{\zb}^{(-k)}\}_{k<n}$ and the functions on the sphere $\{A_u^{(-n,0)}\}$ is sufficient to determine $A_u^{(-n)}$ via equations \eqref{eq:A_u(-1) in radial gauge full Bondi} and \eqref{eq:A_u(-n) in radial gauge full Bondi}. 

Now, we look at the $A_z^{(-n)}$ coefficients. At leading order, the $z$ and $\bar{z}$ components of the gauge field constitute the initial data. For $n\geq1$ we have $A_z^{(-n)}=-n^{-1}F_{rz}^{(-n-1)}$, from \eqref{F_rz in radial full Bondi}. However, thanks to \eqref{F_rz dependence in radial full Bondi} we know that, in general, $F_{rz}^{(-n-1)}$ does depend on $A_z^{(-n)}$. Hence, we need an explicit expression for it, which is provided by the $z$ component \eqref{eq:EOM full Bondi radial z component} of the e.o.m.,
\begin{align}
    \label{eq:A_z solution step 1 radial}
    F^{(-n-1)}_{rz}&=nA_z^{(-n)}-\partial_u^{-1}\bigg(n\partial_zA_u^{(-n)}+nF^{(-n)}_{rz}+\partial_z(\gamma^{-1}F^{(1-n)}_{z\bar{z}}) \\
    &\quad -in\sum_{k=1}^n[A_z^{(k-n)},A_u^{(-k)}]-i\sum_{k=0}^{n}[A^{(k-n-1)}_{u},F^{(-k)}_{rz}]-i\gamma^{-1}\sum_{k=0}^{n-1}[A^{(1+k-n)}_{z},F^{(-k)}_{z\Bar{z}}]\bigg). \nonumber
\end{align}
After some algebra we obtain, for $n\geq1$:
\begin{equation}
    \begin{aligned}
        A_z^{(-n)}&=\frac{1}{2}\partial_u^{-1}\bigg((1-n)A_z^{(1-n)}+ \partial_{z}\left( A_{u}^{(-n)}+\frac{1}{n}\gamma^{-1}F^{(1-n)}_{z\bar{z}} \right) \\
        &\quad -\frac{i}{n}\sum_{k=1}^{n}[A_z^{(k-n)},(2n-k)A_u^{(-k)}+\gamma^{-1}F^{(1-k)}_{z\Bar{z}}]\bigg).
    \end{aligned}
\end{equation}
The above equation together with \eqref{F_zzb dependence in radial full Bondi} and \eqref{A_u(-n) dependence in radial full Bondi} imply that,
\be
    A_z^{(-n)} \quad \mathrm{depends} \ \mathrm{on} \quad A_u^{(-n)},\{A_z^{(-k)},A_{\zb}^{(-k)},A_u^{(-k)}\}_{k<n} \quad \text{and} \quad A_z^{(-n,0)}.
\ee
As an example, $A_z^{(-1)}$ depends just on the initial data and on $A_u^{(-1)}$:
\begin{align}
\label{Az-1recursion}
    A_z^{(-1)}&=\frac{1}{2}\partial_u^{-1}\left(\partial_{z}A_{u}^{(-1)}+\partial_z(\gamma^{-1}F^{(0)}_{z\bar{z}})-i[A_z^{(0)},A_u^{(-1)}+\gamma^{-1}F^{(0)}_{z\Bar{z}}]\right)
\end{align}
In other words, the solution for $A_z^{(-n)}$ is completely determined by our initial hypothesis and by the above solution for $A_u^{(-n)}$. As in the light-cone gauge, it is also possible to compute $A_z^{(-n)}$ by integrating $F_{uz}^{(-n)}$, which in turn can be obtained through the Bianchi identity \eqref{eq:Bianchi radial gauge 1}. This shows that the results are consistent with the algebraic structure of the theory. \autoref{fig:recursive_indices_radial} shows the relation between the coefficients in the radial expansion.

\begin{figure}[h]
    \centering
    \includegraphics[scale=0.7]{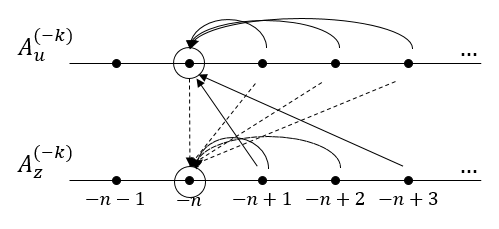}
    \caption{The solid arrows indicate which $A_u^{(-k)}$, $A_z^{(-k)}$ enter the recursion formula for $A_u^{(-n)}$, while the dotted arrows indicate which $A_u^{(-k)}$ and $A_z^{(-k)}$ enter the recursion formula for $A_z^{(-n)}$. The component $A_{\zb}^{(-n)}$ is not drawn, since the same result as 
    for $A_z^{(-n)}$ is valid.}
    \label{fig:recursive_indices_radial}
\end{figure}


\textbf{Recursive formula for $F_{ur}^{(-n)}$:} we can immediately write a recursive formula for the $u-r$component of the field strength, by virtue of \eqref{F_ur in radial full Bondi},
\begin{align}
    F_{ur}^{(-2)}&=-2\gamma^{-1}\partial_u^{-1}\left(D_{(z}\mathcal{F}_{\bar z)u}\right)^{(0)} \\
    F_{ur}^{(-n)}&=\frac{2\gamma^{-1}}{2-n} \left(D_{(z} \mathcal{F}_{\zb)r} \right)^{(1-n)} ,  \qquad n\geq3
\end{align}
Explicitly:
\begin{align}
   \label{Fur_recs} F_{ur}^{(-2)}&=2\gamma^{-1}\left(\partial_{(z}A_{\bar{z})}^{(0)}+i\partial_u^{-1}[\partial_uA_{(z}^{(0)},A_{\bar z)}^{(0)}]\right) \\
    F_{ur}^{(-n)}&=-2\gamma^{-1}\left(\partial_{(z}A_{\Bar{z})}^{(2-n)}+\sum_{k=1}^{n-2}\frac{ik}{2-n}[A^{(2+k-n)}_{(z},A_{\Bar{z})}^{(-k)}]\right), \qquad n\ge3
\end{align}
Formula for $F_{rz}^{(-n)}$ for $n\geq2$:
\begin{equation}
    \begin{aligned}
        F_{rz}^{(-n)}&=\frac{1}{2}\partial_u^{-1}\bigg(F_{rz}^{(1-n)}-\partial_z(F_{ur}^{(-n)}+\gamma^{-1}F_{z\bar z}^{(2-n)}) \\
        &\quad +i\sum_{k=2}^n[A_z^{(k-n)},\frac{2n-k-1}{k-1}F_{ur}^{(-k)}+\gamma^{-1}F_{z\bar z}^{(2-k)}]\bigg)
    \end{aligned}
\end{equation}

\subsubsection{Decays on null infinity in the radiative phase space}

As in the light-cone gauge, the standard $u \rightarrow -\infty$ fall-offs for the radiative data $A_z^{(0)}$,
\be 
A_z^{(0)} (u,z,\zb)= A_z^{(0,0)}  (z,\zb) + o(u^{-\infty}),
\ee 
translates to $F_{ru}^{(-2)}$ via \eqref{F_ur in radial full Bondi} and \eqref{eq:A_u(-1) in radial gauge full Bondi} as follows
\be 
F_{ru}^{(-2)} (u,z,\zb) = F_{ru}^{(-2,0)}(z,\zb) + o(u^{-\infty}).
\ee
We can organise the three non-vanishing components of $A_\mu$ by their $u$-decay,
\beq 
A_z^{(0)} ,A_{\zb}^{(0)} &=& O(1), \quad A_u^{(-1)} = O(1),\\
A_z^{(-1)} ,A_{\zb}^{(-1)} &=& O(u), \quad  A_u^{(-2)} = O(u), \\
A_z^{(-2)} ,A_{\zb}^{(-2)} &=& O(u^2), \quad A_u^{(-3)} = O(u^2),
\dots \nonumber
\eeq
which also translates into the six components of $F_{\mu \nu}$ as is shown in the following table,

\begin{table}[h] 
    \centering
    \begin{tabular}{c|cccccc} 
         $o(u^{-\infty})$ & $F_{uz}^{(0)}$ & $F_{u\zb}^{(0)}$ &   &  & \\
         \hline
         $O(1)$ & $F_{uz}^{(-1)}$ & $F_{u\zb}^{(-1)}$ & $F_{ru}^{(-2)}$ & $F_{z\zb}^{(0)}$  & \\
         \hline
         $O(u)$ & $F_{uz}^{(-2)}$ & $F_{u\zb}^{(-2)}$ & $F_{ru}^{(-3)}$ & $F_{z\zb}^{(-1)}$ & $F_{rz}^{(-2)}$ & $F_{r\zb}^{(-2)}$ \\
         \hline
         $O(u^2)$ & $F_{uz}^{(-3)}$ & $F_{u\zb}^{(-3)}$  & $F_{ru}^{(-4)}$ & $F_{z\zb}^{(-2)}$ & $F_{rz}^{(-3)}$ & $F_{r\zb}^{(-3)}$ \\
         \hline
         $\cdots$ &  &  &  &  & \\
         \hline
         $O(u^n)$ & $F_{uz}^{(-n-1)}$ & $F_{u\zb}^{(-n-1)}$ & $F_{ru}^{(-n-2)}$ & $F_{z\zb}^{(-n)}$ & $F_{rz}^{(-n-1)}$ & $F_{r\zb}^{(-n-1)}$ \\
    \end{tabular}
    \caption{}\label{table_fall_off}
\end{table}

\section{Charges} \label{charges}

In this section we compute the charges that arise when renormalizing the symplectic form, obtaining finite expressions at the corner $\Ib^+_-$. In the extended phase space presented in \eqref{extended_phase_space_YM}, the local gauge transformations are parametrized by a $\mathfrak{g}$-valued function $\Lambda$, giving the usual infinitesimal variation, \eqref{delta_Lambda_HatA}. Since we are working with field-independent variations, \eqref{stuec_gen} can be rewritten as 
\be \label{stueck_field_indep}
\tilde{\mathcal{A}}_\mu = e^{i\Psi}\mathcal{A}_\mu e^{-i\Psi}+ie^{i\Psi}\partial_\mu e^{-i\Psi},
\ee
and the strength tensor is given by,
\be \label{F_field_indep}
\tilde\F_{\mu\nu}=e^{i\Psi}\F_{\mu\nu}e^{-i\Psi}.
\ee
The \textit{bulk} action is
\be \label{tilded_Lagr}
S[\tilde{\A}_\mu] = \int_D \tr \left(  \tilde{\mathcal{F}}_{\mu\nu} \tilde{\mathcal{F}}^{\mu\nu}\right) dvol,
\ee 
where $D$ is some region with its boundary $\partial D := \Sigma_{t_0} \cup \Sigma_{t_1} \cup \mathcal{H}$, consisting in two Cauchy surfaces along Minkowski and a time-like hypersurface $\mathcal{H}$ at spatial infinity (see \autoref{scri_plus_diag}). 

We remark that the Lagrangian \eqref{tilded_Lagr}, constructed from $\tilde{\A}_\mu$ of \eqref{stuec_gen}, takes the same form as the YM Lagrangian constructed from $\A_\mu$. This can be seen immediately from the fact that $\tilde{\A}_\mu$ has the functional form of a gauge transformation, with the parameter replaced by our new field $\breve\Psi$, as is characteristic of the Stueckelberg procedure\footnote{The traditional form of the Stueckelberg trick is applied to an action where a gauge symmetry is broken in the Lagrangian, whereas in our case it manifests as a breaking of the asymptotic symmetries once the fall-offs have been imposed, and is thus not visible in the bulk action.}. In order to compute the relevant charges though, it is crucial to take the field $\tilde{\A}_\mu$ as our fundamental field in the covariant phase space formalism.

From \eqref{tilded_Lagr}, one can then compute the \textit{bulk} symplectic form straightforwardly,
\be \label{symplectic_form_bulk}
\Omega^{\text{bulk}}[\delta_1 , \delta_2] = - \int_\Sigma \tr \left( \delta_1 \tilde{\mathcal{F}}^{\mu \nu} \wedge \delta_2 \tilde{\Av}_\nu \right) dS_\mu - (1 \leftrightarrow 2),
\ee
for some Cauchy surface $\Sigma$ in Minkowski space. The symplectic form can be used to obtain the canonical charges associated with the infinitesimal symmetries \eqref{delta_Lambda_HatA}, by means of the relation
\be \label{def_int_charge}
\Omega^{\text{bulk}}[\delta_\Lambda , \delta] = \delta \tilde{Q}_\Lambda,
\ee 
giving the following result:
\be \label{charge_ext_YM}
\tilde{Q}_\Lambda = \int_\Sigma \partial_\nu \tr (\Lambda \tilde{\F}^{\mu \nu}) dS_\mu.
\ee

We stress that this charge is computed \textit{within the bulk}. As it stands, the expression for the charge is a total derivative, which implies that
\be 
\tilde{Q}_\Lambda = \int_{S^2} \tr (\Lambda \tilde{\F}^{\mu \nu}) dS_{\mu \nu},
\ee
where we are using that $\partial \Sigma = i_0 \simeq S^2$, with $i_0$ denoting spatial infinity. We can identify $\partial \Sigma$ with $\Ib^+_-$ using Bondi coordinates. Indeed, when working in these coordinates, the symplectic form at null infinity $\Ib^+$ is found by taking $\Sigma$ to be the hypersurface $\{t = r+u = constant\}$, and taking the limit $t\rightarrow +\infty$ at fixed $u$. In that way, $\Ib^+$ is parametrized by the null coordinate $u$. Therefore, by taking the limit $r \rightarrow +\infty$, while maintaining $u$ fixed, we have $\Sigma \rightarrow \Ib$ and $\partial \Sigma \rightarrow \Ib^+_-$ (see \autoref{scri_plus_diag}). The conservation of the charges from $\Ib^-_+$ to $\Ib^+_-$ is showed via the antipodal mapping at $\Ib^+_-$ (cf. \cite{Strominger:2013lka}). This antipodal map can be directly translated to $\Psi$. Indeed, if we defined $\Psi^+$ and $\Psi^-$ for the retarded and advanced coordinates respectively, then
\be 
\Psi^+(\hat{x}) = \Psi^-(-\hat{x}) ,
\ee
where $\hat{x}$ are the stereographic projection from the complex plane to the Riemann sphere.
Thus, on the celestial sphere (i.e., $\Ib^+_-$), the charge is the limit $r \rightarrow +\infty$ and $u \rightarrow -\infty$ in the quantity
\be \label{full_ch}
\tilde{Q}_\Lambda = \int_{S^2} \tr (\Lambda \tilde{\F}_{ru}) r^2 \gamma dz d\zb.
\ee
 Of course, this expression can diverge, and in general, a renormalization procedure has to be carried out to give meaning to the expression (see \cite{Freidel:2019ohg, Peraza:2023ivy}). 

\begin{figure}[h]
    \centering
    \includegraphics[scale=0.6]{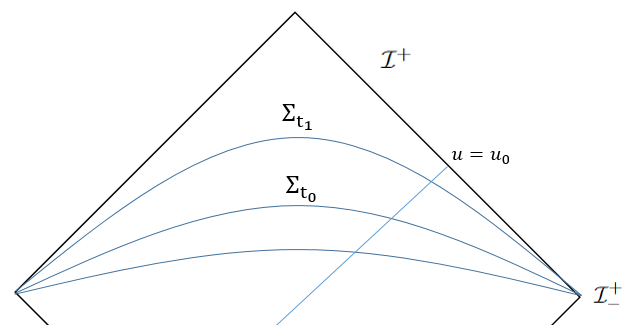}
    \caption{Cauchy Surfaces $\Sigma_{t_0}, \Sigma_{t_1}$, with $t_0 < t_1$, and constant $u$ ray. $\mathcal{H}$ is located at spatial infinity and collapsed to a point after compactification. The sets (topologically spheres) in the intersection of the hypersurfaces with $\{u = u_0\}$ correspond to the celestial sphere in the limit $t\rightarrow + \infty$ with $u$ fixed.}
    \label{scri_plus_diag}
\end{figure}

Assuming the renormalization procedure has been carried out (as is the case for Maxwell, \cite{Freidel:2019ohg, Peraza:2023ivy}), we will focus on the finite term in \eqref{full_ch}. Abusing notation, we will call such a finite part \textit{charge} and denote it also by $\tilde{Q}_\Lambda$. We will have the expansion in $r$, and the expansion in $\Psi$ and $\Lambda_+$. A careful counting of the orders of each field will indeed show that the first few orders in our result are consistent with previous works in the literature. 

\subsection{Charges in the light-cone gauge}

Let us start by computing each order in $r$ of $r^2 \tilde{\F}_{r u}$ \footnote{Recall that the $r^2$ comes from the volume element in $dS^2$ in standard Bondi coordinates.}. Note that we can re-write $\tilde{\F}_{\mu\nu}$ as
\be 
\begin{aligned}
\tilde\F_{\mu\nu}=&e^{i\Psi}\F_{\mu\nu}e^{-i\Psi}\\
=&e^{ad_{i \Psi}}\F_{\mu\nu}\\
=&\F_{\mu\nu}+\frac{e^{ad_{i \Psi}}-1}{ad_{i \Psi}}[i\Psi,\F_{\mu\nu}]\\
=&\F_{\mu\nu}+\frac{1-e^{ad_{i \Psi}}}{ad_{-i \Psi}}[i\Psi,\F_{\mu\nu}]\\
=&\F_{\mu\nu}+\Op_{-i\Psi } [i\Psi,\F_{\mu\nu}]
\end{aligned}
\ee 
where to get to the last line we used \eqref{O_def_gen}. Then we have
\be \label{r2_F_term_i}
\left( r^2 \tilde{\F}_{r u} \right)^{(-i)} =  F_{r u}^{(-2-i)}  + \left( \Op_{-i\Psi } \left( [i\Psi, \F_{ru}] \right) \right)^{(-2-i)}.
\ee
Explicitly, the first three coefficients are
\beq 
\left( r^2 \tilde{\F}_{r u} \right)^{(0)} &=&  F_{r u}^{(-2)} + [i\Psi^{(1)}, F_{r u}^{(-3)}] + \frac{1}{2} [i\Psi^{(1)}, [i\Psi^{(1)}, F_{r u}^{(-4)}] ]+ [i\Psi^{(2)}, F_{r u}^{(-4)}] + ..., \nonumber \\ 
\left( r^2 \tilde{\F}_{r u} \right)^{(-1)} &=&  F_{r u}^{(-3)} + [i\Psi^{(1)}, F_{r u}^{(-4)}] + \frac{1}{2} [i\Psi^{(1)}, [i\Psi^{(1)}, F_{r u}^{(-5)}] ]+ [i\Psi^{(2)}, F_{r u}^{(-5)}] + ..., \nonumber\\
\left( r^2 \tilde{\F}_{r u} \right)^{(-2)} &=&  F_{r u}^{(-4)} + [i\Psi^{(1)}, F_{r u}^{(-5)}]+ \frac{1}{2} [i\Psi^{(1)}, [i\Psi^{(1)}, F_{r u}^{(-6)}] ]+ [i\Psi^{(2)}, F_{r u}^{(-6)}] + ...,\nonumber
\eeq

Since we take $\Lambda(r,z,\zb) = \sum_{k=0}^{+\infty} r^{k} \Lambda^{(k)}(z,\zb)$, the finite term in the charge is given by,
\be 
\tilde{Q}_\Lambda = \int_{S^2} \tr \left( \sum_{l = 0}\Lambda^{(l)} \left( r^2 \tilde{\F}_{r u} \right)^{(-l)} \right) dS_{S^2}, 
\ee
where $dS_{S^2} = \gamma dz d\zb$. Then, in terms of expansion \eqref{r2_F_term_i}, it is, 
\beq 
\tilde{Q}_\Lambda &=& \int_{S^2} \tr \left( \sum_{l = 0}\Lambda^{(l)} \left(  F_{r u}^{(-2-l)} + [i\Psi^{(1)}, F_{r u}^{(-3-l)}] \right. \right. \nonumber\\
&& \left. \left. + \frac{1}{2} [i\Psi^{(1)}, [i\Psi^{(1)}, F_{r u}^{(-4-l)}] ] + [i\Psi^{(2)}, F_{r u}^{(-4-l)}] + ...\right) \right) dS_{S^2}.
\eeq

For each $l\geq 0$, consider the \textit{extended} charge density associated to $\Lambda^{(l)}$, defined as 
\beq 
\tilde{q}_{\Lambda^{(l)}} &: =&  \tr \left( \Lambda^{(l)} \left(  F_{r u}^{(-2-l)} + [i\Psi^{(1)}, F_{r u}^{(-3-l)}] \right. \right. \nonumber\\
&& +\left. \left. \frac{1}{2} [i\Psi^{(1)}, [i\Psi^{(1)}, F_{r u}^{(-4-l)}] ] + [i\Psi^{(2)}, F_{r u}^{(-4-l)}] + ...\right) \right),  \label{q_tilde_def}
\eeq
and the corresponding charge is denoted by
\be 
\tilde{Q}_{\Lambda^{(l)}} := \int_{S^2}  \tilde{q}_{\Lambda^{(l)}} dS_{S^2},
\ee 
while the total charge is given by
\be \label{total_caharge_as_sum_charges}
\tilde{Q}_{\Lambda} = \int_{S^2} \sum_{l=0}^{+\infty} \tilde{q}_{\Lambda^{(l)}} dS^2.
\ee

Observe that $\tilde{q}_{\Lambda^{(l)}}$ is linear in both $\Lambda^{(l)}$ and the coefficients $F_{r u}^{(-2-i-l)}$, for $i\geq 0$. The approximation up to $j$-th order in $\Psi$ to $\tilde{q}_{\Lambda^{(l)}}$ is defined as  
\be 
\overset{j}{q}_{\Lambda^{(l)}}:=  \tr \left( \Lambda^{(l)} \left(  F_{r u}^{(-2-l)} + [i\Psi^{(1)}, F_{r u}^{(-3-l)}] + ... + \mathfrak{ord}^{(j)}_{i \Psi} [F_{ru}^{(-2-j-l)}] \right) \right),
\ee 
where we have introduced the operator
\be 
\begin{aligned}
\mathfrak{ord}^{(j)}_{i \Psi}=&\sum_{k=1}^j \frac{1}{k!} \sum_{i_1 + ... + i_k = j} ad_{i \Psi^{(i_1)}} ( ... (ad_{i \Psi^{(i_k)}})...)[\cdot] \\
=&\left(e^{ad_{i\psi}} -1 \right)^{(j)}[\cdot] \ ,
\end{aligned}
\ee 
i.e. $\mathfrak{ord}^{(j)}_{i \Psi}$ denotes nested commutators of $i\Psi$ which add up to $\O(r^j)$.

We can re-write the charge $\overset{j}{q}_{\Lambda^{(l)}}$ by observing that $\mathfrak{ord}^{(j)}_\Psi$ is symmetric in $i_1 ,..., i_k$ and using iteratively the identity 
\be 
\tr(A [B,C]) = \tr([-B,A]C),
\ee
as follows,
\be 
\overset{j}{q}_{\Lambda^{(l)}}:=  \tr \left( F_{r u}^{(-2-l)} \Lambda^{(l)} - [i\Psi^{(1)}, \Lambda^{(l)} ] F_{r u}^{(-3-l)} + ... + \mathfrak{ord}^{(j)}_{-i \Psi}  [\Lambda^{(l)}] F_{ru}^{(-2-j-l)} \right).
\ee 
Similarly, we define the $j$-th approximation to the charge $\tilde{Q}_{\Lambda^{(l)}}$ as 
\be 
\overset{j}{Q}_{\Lambda^{(l)}} = \int_{S^2} \overset{j}{q}_{\Lambda^{(l)}} dS_{S^2}
\ee

In what follows we will define a hierarchy of sub$^n$-charges, labeled by $j+l \geq 0$ in the set of charge densities sequences $\{ \{ \overset{j}{q}_{\Lambda^{(l)}} \}_{l\geq 0}\}_{j\geq0}$, such that at each cut-off in the expansion of $\Psi$ we have a closed charge algebra, that approximates $\tilde{q}_{\Lambda^{(l)}}$. Each algebra will correspond to the radiative, linear, quadratic, ..., approximations.

\subsection{Leading Charges from zeroth order in fields}

We start with the zeroth order in the hierarchy, given by only one charge, $\{ \overset{0}{q}_{\Lambda^{(0)}} \}$. It corresponds to the usual charge derived from the radiative data \cite{Strominger:2013lka, Strominger:2017zoo}. It can be written simply as
\be 
\overset{0}{q}_{\Lambda^{(0)}} = \tr \left( \Lambda^{(0)} F_{r u}^{(-2)} \right) \ .
\ee

 This is an infinite dimensional space of charges, parametrized by $\Lambda^{(0)}$\footnote{Recall that $\Lambda^{(0)}$ is a $\mathfrak{g}$-valued function on the sphere, for the Lie algebra $\mathfrak{g}$ of the Lie group $G$.}. The definition is geometrically consistent with imposing the vanishing of terms containing $\Psi$ or any of the coefficients in $\Lambda_+$, eqn. \eqref{simple_lambda_plus}.
 In other words, this corresponds to a vanishing of the rest of the charge aspect approximations,
\be 
\overset{0}{q}_{\Lambda^{(l)}} = 0 \quad \forall l \geq 1.
\ee

\subsection{Sub-leading Charge from up to linear order in fields}

In the first truncation of the hierarchy that contains the sub-leading charge, we only consider linear terms in $\Psi$ or $\Lambda_0$ for the leading charge and the linear in $\Lambda_+$ and $\Psi$-independent terms for the sub-leading.

Again, this is consistent with imposing $\Lambda^{(l)} = 0$ for all $l\geq 2$ and taking only linear variations of the field $\Psi$. Then, 
\beq \label{q_1u_0d}
\overset{1}{q}_{\Lambda^{(0)}} &=& \tr \left( \Lambda^{(0)} \left( F_{r u}^{(-2)} + [i\Psi^{(1)}, F_{r u}^{(-3)}]\right) \right), \\
\overset{0}{q}_{\Lambda^{(1)}} &=& \tr \left( \Lambda^{(1)} F_{r u}^{(-3)} \right).
\eeq
We remark that this linear approximation in $\Psi$ and $\Lambda_+$ is equivalent to the linearization \cite{Campiglia:2021oqz}
\be 
\tilde{\A_\mu} = \A_\mu + D_\mu \Psi.
\ee 
Observe also that the functional dependence of $\overset{0}{q}_{\Lambda^{(1)}}$ on $\F_{r u}^{(-3)}$ is the same as the leading charge $\overset{0}{q}_{\Lambda^{(0)}}$ on $\F_{r u}^{(-2)}$. Thus, we can write the charge $\overset{1}{q}_{\Lambda^{(0)}}$ in the following way,
\be \label{linear_order_relation}
\overset{1}{q}_{\Lambda^{(0)}} = \overset{0}{q}_{\Lambda^{(0)}} - \overset{0}{q}_{[i\Psi, \Lambda^{(0)}]^{(1)}}.
\ee 
thanks to identity $\tr (A[B,C]) = - \tr([B,A]C)$. In \eqref{linear_order_relation}, we have used the notation
\be 
\label{comm_subscript_order}
[i\Psi, \Lambda^{(m)}]^{(n)}\equiv[i\Psi^{(n-m)}, \Lambda^{(m)}]
\ee 
The first term in \eqref{linear_order_relation} is in the zeroth order charge algebra, while the second one is a sub-leading contribution from the first order charge algebra. Notice that we can keep track of which algebra we are referring to with the charge $\overset{j}{q}_{\Lambda^{(l)}}$ by computing $j+l$.

These results can be connected with those in \cite{Campiglia2021}, in view of the recursion formulas found in \autoref{eom_scri_plus}. The details are provided in \autoref{sec_charges_corners}.

\subsection{Sub-sub leading Charge from up to second order in fields}

For the expressions in this subsection, we have up to quadratic terms in $\Psi$, and $\Lambda_+^{(2)}$. The charge aspects are
\beq 
\overset{2}{q}_{\Lambda^{(0)}} &=& \tr \left( \Lambda^{(0)} \left(  F_{r u}^{(-2)} + [i\Psi^{(1)}, F_{r u}^{(-3)}] + \frac{1}{2} [i\Psi^{(1)}, [i\Psi^{(1)}, F_{r u}^{(-4)}] ] + [i\Psi^{(2)}, F_{r u}^{(-4)}] \right) \right), \\
\overset{1}{q}_{\Lambda^{(1)}} &=& \tr \left(\Lambda^{(1)} \left( F_{r u}^{(-3)} + [i\Psi^{(1)}, F_{r u}^{(-4)}]\right) \right), \\
\overset{0}{q}_{\Lambda^{(2)}} &=& \tr \left( \Lambda^{(2)} F_{r u}^{(-4)} \right) .
\eeq

In defining the above charge aspects, we took the same approach as in the sub-leading case by considering a cut-off in the $r$-expansion of the fields, i.e.,
\be 
\Lambda^{(n)} = \Psi^{(n)} = 0 ,\quad \forall n \geq 3,
\ee 
which gives our candidate for the sub$^2$-leading charge. We can write a recursive expression of the charges in terms of the previous ones,  
\beq
\overset{2}{q}_{\Lambda^{(0)}}  &=& \overset{1}{q}_{\Lambda^{(0)}} -  \overset{0}{q}_{[i\Psi, \Lambda^{(0)}]^{(2)}} + \overset{0}{q}_{\frac{1}{2} [i\Psi,[i\Psi, \Lambda^{(0)}]]^{(2)}}, \\
\overset{1}{q}_{\Lambda^{(1)}}  &=& \overset{0}{q}_{\Lambda^{(1)}} -  \overset{0}{q}_{[i\Psi, \Lambda^{(1)}]^{(2)}},
\eeq
where observe that, as before, we are using the convention that $\overset{j}{q}_{\Lambda^{(l)}}$ is in the level $j+l$ of the hierarchy, for some LGT $\Lambda$. 

\subsection{Sub$^n$-leading charges from the truncation up to $n-$th order in the fields}

In view of the previous examples, given $n>2$, the charge aspects are defined recursively as follows,
\beq 
\overset{0}{q}_{\Lambda^{(n)}}  &=& \tr \left( \Lambda^{(n)} F_{r u}^{(-2-n)} \right) , \\
\overset{1}{q}_{\Lambda^{(n-1)}}  &=&  \overset{0}{q}_{\Lambda^{(n-1)}} +  \overset{0}{q}_{[-i\Psi, \Lambda^{(n-1)}]^{(n)}}, \\
&...& \\
\overset{n-1}{q}_{\Lambda^{(1)}}  &=& \overset{n-2}{q}_{\Lambda^{(1)}} + \sum_{k=1}^{n-1} \overset{0}{q}_{\left( \frac{1}{k!} ad^k_{-i\Psi}(\Lambda^{(1)})\right) ^{(n)}},\\
\overset{n}{q}_{\Lambda^{(0)}}  &=& \overset{n-1}{q}_{\Lambda^{(0)}} + \sum_{k=1}^{n} \overset{0}{q}_{\left( \frac{1}{k!} ad^k_{-i\Psi}(\Lambda^{(0)})\right)^{(n)}},
\eeq
 where, in order to be able to display the charges separately, we are identifying, by abuse of notation, that $\Lambda^{(n)}$ is the large gauge transformation whose only coefficient is at $r^n$ and its value is $\Lambda^{(n)}$. For an arbitrary large gauge transformation $\Lambda$, the total charge is given by the sum, cf. \eqref{total_caharge_as_sum_charges}.
 
 This way of writing the charges comes from a \textit{dressing} of the large gauge symmetry parameter by the Stueckelberg field, which can be seen directly from the charge expression \eqref{charge_ext_YM} as 
\be 
\tr \left( \Lambda e^{i \Psi} \F_{ru} e^{-i\Psi} \right) = \tr \left( e^{-i\Psi} \Lambda e^{i \Psi} \F_{ru} \right)
= \tr\left(e^{ad_{-i\Psi}}\Lambda\F_{ru} \right) .
\ee 

Now, trivially, we have that for any $0\leq j \leq n$, 

\be 
\overset{n-j}{q}_{\Lambda^{(j)}}  = \overset{0}{q}_{\Lambda^{(j)}} + \overset{0}{q}_{[-i\Psi,\Lambda^{(j)}]^{(j+1)}} + ...  + \sum_{k=1}^{n-j} \overset{0}{q}_{\left( \frac{1}{k!} ad^k_{-i\Psi}(\Lambda^{(j)})\right)^{(n)}}.
\ee 

Observe that we can rewrite the charges in terms of the exponential operator, simply by working up to its combined $n-$th order in $(\Psi, \Lambda)$, for $1\leq j \leq n$,
\beq 
\overset{0}{q}_{\Lambda^{(n)}}  &=& \tr \left( \Lambda^{(n)} F_{r u}^{(-2-n)} \right)  \\
\overset{n-j}{q}_{\Lambda^{(j)}}  &=& \overset{0}{q}_{[e^{ad_{-i\Psi}}(\Lambda^{(j)})]^{\{ n \}}}, \label{Op_expression_charges}
\eeq
where the superscript $\{ n \}$ in $[e^{ad_{-i\Psi}}(\Lambda^{(j)})]^{\{ n \}}$ denotes that we are considering the expression only up to the $n$-th order in combinations of $\Psi$ and $\Lambda^{(j)}$ (recall that $\Lambda^{(j)}$ has order $j$, fixed). Of course, this is consistent with \eqref{full_ch}, and is a compact expression that contains all the information up to $n$-th order. 

\section{Charge Algebra at each level} \label{charge_algebra}

In this section, we present the closure of each charge algebra in the hierarchy presented in the previous section. At each level of the hierarchy, we show that the charges act canonically, and the phase space has a natural structure. It is a non-trivial check since it has to be deduced from the definitions of each phase space. 

We are assuming we are working in a field-independent frame for the large gauge transformations (i.e., $\delta_{\Lambda_1} \Lambda_2=0$). Otherwise, a simple generalization can be carried by exchanging the bracket by a modified bracket \eqref{mod_bracket}.

Fixing $n>0$, the first question we have to address is how $\overset{n-j}{Q}_{\Lambda^{(j)}}$ acts on the extended phase space, in particular on $\Psi$, for some function on the sphere $\Lambda^{(j)}$ that generates a Large Gauge Transformation of order $j$ in the $r$-expansion. Now, by inspection of \eqref{Op_expression_charges}, we see that the expression involves the coefficients $\Psi^{(1)},...,\Psi^{(n-j)}$ of the Stueckelberg field. Therefore, its action on the phase space must involve only those fields. Moreover, for the action to be canonical, 
\be \label{up_to_n_symplectic_form}
\Omega^{\{n\}}[\{ \overset{n-j}{Q}_{\Lambda^{(j)}} , \cdot \} , \delta] = \delta \overset{n-j}{Q}_{\Lambda^{(j)}},
\ee 
where $\Omega^{\{n\}}$ is the symplectic form at $\Ib^+$ (i.e., the renormalized version of \eqref{symplectic_form_bulk}) up to $n-$th order in $\Psi$. Therefore, since we are allowing for up to $(n-j)$-th order in $\Psi$, then the action of the variation associated to $\{\overset{n-j}{Q}_{\Lambda^{(j)}} , \cdot\}$ on the coefficients of the Stueckelberg field $\Psi$ is given by the $(n-j)$-approximation of the variation by $\Lambda^{(l)}$,
\be \label{canonical_action_charge}
\{\overset{n-j}{Q}_{\Lambda^{(j)}} , \cdot\} [\Psi]= \overset{n-j}{\delta}_{\Lambda^{(j)}} [\Psi],
\ee 
where recall that $ \overset{j}{\delta}:= \sum_{k=0}^j \delta^{[k]}$, as defined in \autoref{sec_extended_phase_space_LC}. Observe that this is a minimal requirement for equation \eqref{up_to_n_symplectic_form} to hold consistently with \eqref{def_int_charge}.

In what follows, we consider only charge densities (the same results hold for the integrated charge).

\begin{figure}
    \centering
    \includegraphics[scale=0.6]{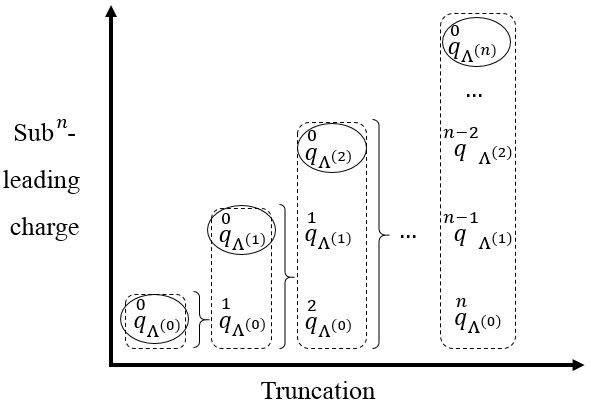}
    \caption{Diagram showing the organization of the subalgebras of charges. Along the horizontal axis, we increase in the power of $r$ involved in $\Psi$, while the vertical axis indicates the sub$^n-$leading charges. In circles, the first  sub$^n-$leading charge for each $n$. The curled brackets indicate the recursion \eqref{Op_expression_charges}.} 
    \label{fing_charge_hierarchy}
\end{figure}

\subsection{Zeroth order Charge algebra}

Recall that in the radiative phase space, we have $\Psi \equiv 0$, and therefore the action of a large gauge transformation $\Lambda$ is simply

\be 
\delta_\Lambda A_z^{(0)} = D^{(0)}_z \Lambda^{(0)}.
\ee

Due to the cut-off of the order in $\Psi$, observe that the action of the charge on the symplectic space is given by the identity,
\be 
\{ \overset{0}{q}_{\Lambda^{(0)}} , \cdot \} = \overset{0}{\delta}_{\Lambda^{(0)}}.
\ee
Here, again, we are identifying the large gauge transformation $\Lambda$ with only one coefficient at $r^0$ with its value, $\Lambda^{(0)}$, in order to ease the notation. Then, it is straightforward to prove that the zeroth-order charge algebra is closed, 
\be 
\{\overset{0}{q}_{\Lambda_1^{(0)}} , \overset{0}{q}_{\Lambda_2^{(0)}} \} = \overset{0}{q}_{- i[\Lambda_1 , \Lambda_2]^{(0)}} .
\ee
Let us denote this algebra as $\mathfrak{p}_0$, generated by all possible $\overset{0}{q}_{\Lambda^{(0)}}$, with $\Lambda^{(0)} : S^2 \rightarrow \R$.

\subsection{First order charge algebra}

Recall that by abuse of notation we are denoting $\Lambda^{(n)}$ as the large gauge transformation whose only coefficient is at $r^n$ and its value is $\Lambda^{(n)}$.

We have the identities,
\beq 
\{\overset{0}{q}_{\Lambda^{(1)}}  , \cdot \} &=& \delta^{[0]}_{\Lambda^{(1)}} , \\ 
\{ \overset{1}{q}_{\Lambda^{(0)}}  , \cdot \} &=& \overset{1}{\delta}_{\Lambda^{(0)}} =  \delta^{[0]}_{\Lambda^{(0)}} + \delta^{[1]}_{\Lambda^{(0)}}
\eeq

First, we compute the bracket for the radiative subalgebra,
\be \nonumber
\begin{aligned}
\{\overset{1}{q}_{\Lambda_1^{(0)}} , \overset{1}{q}_{\Lambda_2^{(0)}} \} &= \overset{1}{\delta}_{\Lambda_1^{(0)}} \tr \left( \Lambda_2^{(0)} F_{r u}^{(-2)} - [i\Psi^{(1)}, \Lambda_2^{(0)}] F_{r u}^{(-3)}\right) \\
&= \tr \left( \Lambda_2^{(0)} \delta^{[1]}_{\Lambda_1^{(0)}} F_{r u}^{(-2)} - [i 
 \overset{1}{\delta}_{\Lambda_1^{(0)}} \Psi^{(1)}, \Lambda_2^{(0)}] F_{r u}^{(-3)} - [i \Psi^{(1)}, \Lambda_2^{(0)}] \overset{1}{\delta}_{\Lambda_1^{(0)}} F_{r u}^{(-3)}\right) \\
&= \tr \left( \Lambda_2^{(0)} [ i \Lambda_1^{(0)}, F_{r u}^{(-2)}] + i [i 
 [ \Psi^{(1)}, \Lambda_1^{(0)}] , \Lambda_2^{(0)}] F_{r u}^{(-3)} \right. \\
&\quad  \left. -[i \Psi^{(1)}, \Lambda_2^{(0)}] [i \Lambda_1^{(0)}, F_{r u}^{(-3)}]  \right), 
\end{aligned}
\ee
where in the first line we used \eqref{q_1u_0d}, and in the last line we used that $\Lambda_1^{(0)}$ has only the zeroth-order coefficient in the $r$-expansion, simplifying \eqref{trans_stueck_perturbative_field}. Then, using identity \eqref{tr_formula}, we have, 

\beq 
\{\overset{1}{q}_{\Lambda_1^{(0)}} , \overset{1}{q}_{\Lambda_2^{(0)}}   \} &=& \tr \left( -i[ \Lambda_1^{(0)}, \Lambda_2^{(0)} ] F_{r u}^{(-2)} + [i \Psi^{(1)}, i [\Lambda^{(0)}_1, \Lambda_2^{(0)}]] F_{r u}^{(-3)} \right), \nonumber \\
&=&\overset{1}{q}_{ -i [\Lambda^{(0)}_1, \Lambda_2^{(0)}]} .\nonumber
\eeq
The previous step can be written in a clear form by observing that 
\be 
[\Lambda_1^{(0)},\Lambda_2^{(0)}] = [\Lambda_1,\Lambda_2]^{(0)}.
\ee

Next, we compute the mixed Poisson bracket between leading and sub-leading charges,

\be
\begin{aligned}
\{\overset{0}{q}_{\Lambda_1^{(1)}} , \overset{1}{q}_{\Lambda_2^{(0)}} \} =& \overset{0}{\delta}_{\Lambda_1^{(1)}} \tr \left( \Lambda_2^{(0)} F_{r u}^{(-2)} - [i\Psi^{(1)}, \Lambda_2^{(0)}] F_{r u}^{(-3)}\right) \\
=& \tr \left( \Lambda_2^{(0)} \overset{0}{\delta}_{\Lambda_1^{(1)}} F_{r u}^{(-2)} - [i 
\overset{0}{\delta}_{\Lambda_1^{(1)}} \Psi^{(1)}, \Lambda_2^{(0)}] F_{r u}^{(-3)} - [i \Psi^{(1)}, \Lambda_2^{(0)}] \overset{0}{\delta}_{\Lambda_1^{(1)}} F_{r u}^{(-3)}\right) \\
=& \tr \left( - i [\Lambda_1^{(1)} , \Lambda_2^{(0)}] F_{r u}^{(-3)} \right) \\
=& \overset{0}{q}_{-i[\Lambda_1^{(1)}, \Lambda_2^{(0)}]}  =\overset{0}{q}_{-i[\Lambda_1, \Lambda_2]^{(1)}}.
\end{aligned}
\ee
To get from the second to the third line we use that in general $\delta_{\Lambda^+} F_{r u}^{(-2-k)} = 0$ and the variation \eqref{zeroth_order_variation}.


Similarly, we do a consistency check for the anti-symmetric bracket, 
\beq 
\{ \overset{1}{q}_{\Lambda_2^{(0)}} ,\overset{0}{q}_{\Lambda_1^{(1)}} \} &=& \overset{1}{\delta}_{\Lambda_2^{(0)}} \tr \left( \Lambda_1^{(1)} F_{r u}^{(-3)} \right), \\
&=& \tr \left( - i \Lambda_1^{(1)} [ F_{r u}^{(-3)}, \Lambda_2^{(0)}] \right), \\
&=& \overset{0}{q}_{-i [\Lambda_2 , \Lambda_1]^{(1)}}.
\eeq
Finally, the sub-leading/sub-leading components of the Poisson bracket vanish,

\be
\{ \overset{0}{q}_{\Lambda_1^{(1)}} ,\overset{0}{q}_{\Lambda_2^{(1)}} \} = {\delta}^{[0]}_{\Lambda_1^{(1)}} \Tr \left( \Lambda_2^{(1)} F_{r u}^{(-3)} \right) = 0,
\ee
since $\overset{0}{\delta}_{\Lambda_1^{(1)}}$ acting on $\Lambda_2^{(1)} $ vanishes, as well as on $ F_{r u}^{(-3)}$.

Summarizing, 

\beq 
\{\overset{1}{q}_{\Lambda_1^{(0)}} , \overset{1}{q}_{\Lambda_2^{(0)}} \} &=& \overset{1}{q}_{-i[\Lambda_1^{(0)},\Lambda_2^{(0)}]} ,\\
\{\overset{0}{q}_{\Lambda_1^{(1)}} , \overset{1}{q}_{\Lambda_2^{(0)}} \} &=& - \{ \overset{1}{q}_{\Lambda_2^{(0)}} ,\overset{0}{q}_{\Lambda_1^{(1)}} \} = \overset{0}{q}_{-i[\Lambda_1 , \Lambda_2]^{(1)}}\\
\{ \overset{0}{q}_{\Lambda_1^{(1)}} ,\overset{0}{q}_{\Lambda_2^{(1)}} \} &=& 0.
\eeq

Therefore, we have a first order Poisson algebra, which we will denote by $\mathfrak{p}_1$, generated by functions $\Lambda^{(0)},\Lambda^{(1)}:S^2 \rightarrow \R$. Observe that $\mathfrak{p}_0 \subset \mathfrak{p}_1$.

\subsection{General case}

In this subsection, we are using the following standard identities,
\beq 
[A,[B,C]]  &=& [[A,B],C] + [B, [A,C]] ,\\
\tr \left( \left[ A, ad_{A_1}(...(ad_{A_n}(B))...) \right] \right) &=&  \tr \bigg( ad_{A_1}( ... ad_{A_n}([A,B])...)  \nonumber\\
&&  + \sum_{k=1}^n ad_{A_1} (... ad_{[A,A_k]} (... ad_{A_n}(B))...). \bigg) \label{Leibnitz_rule_nested}
\eeq  

\subsubsection{Brackets with leading charge}

First, we compute the brackets with the leading charge $\overset{n}{q}_{\Lambda^{(0)}}$. The action of $\delta_{\Lambda^{(0)}}$ on $F_{ru}^{(l)}$ is 
\be 
\delta_{\Lambda^{(0)}} F_{ru}^{(l)} = i[\Lambda^{(0)},F_{ru}^{(l)}],
\ee 
while its action on $\Psi$ is (cf. \eqref{trans_stueck_perturbative_field}) ,
\be 
\overset{n}{\delta}_{\Lambda^{(0)}} \Psi = i [\Lambda^{(0)},\Psi].
\ee
Then, 
\beq 
\overset{n}{\delta}_{\Lambda_1^{(0)}} \overset{n-j}{q}_{\Lambda_2^{(j)}} &=& \overset{n}{\delta}_{\Lambda_1^{(0)}}  \overset{0}{q}_{[e^{ad_{-i\Psi}}(\Lambda_2^{(j)})]^{\{ n \}}} \\
&=& \overset{n}{\delta}_{\Lambda_1^{(0)}} \left( \overset{0}{q}_{\Lambda_2^{(j)}} + \overset{0}{q}_{[-i\Psi,\Lambda_2^{(j)}]^{(j+1)}} +... + \sum_{k=1}^{n-j} \overset{0}{q}_{\left( \frac{1}{k!} ad^k_{-i\Psi}(\Lambda_2^{(j)})\right)^{(n)}} \right) \\
&=& \overset{0}{q}_{-i[\Lambda_1^{(0)},\Lambda_2^{(j)]}} + \overset{0}{q}_{[-i\Psi,-i[\Lambda_1^{(0)},\Lambda_2^{(j)]}]^{(j+1)}} +... + \sum_{k=1}^{n-j} \overset{0}{q}_{\left( \frac{1}{k!} ad^k_{-i\Psi}(-i[\Lambda_1^{(0)},\Lambda_2^{(j)}])\right)^{(n)}},
\eeq
where we use repeatedly \eqref{Leibnitz_rule_nested}.

Then, 
\be 
\overset{n}{\delta}_{\Lambda_1^{(0)}} \overset{n-j}{q}_{\Lambda_2^{(j)}} = \overset{n-j}{q}_{-i[\Lambda_1^{(0)},\Lambda_2^{(j)}]}.
\ee 

\subsubsection{Bracket with sub-leading charge}

Next, we compute the brackets with the sub-leading charge, $\overset{n-1}{q}_{\Lambda^{(1)}} $. In this case, we have
\beq \label{variation_order_1}
\overset{n-1}{\delta}_{\Lambda_1^{(1)}} \overset{n-j}{q}_{\Lambda_2^{(j)}} &=& \overset{n-1}{\delta}_{\Lambda_1^{(1)}} \left(\overset{0}{q}_{\Lambda_2^{(j)}} + \overset{0}{q}_{[-i\Psi,\Lambda_2^{(j)}]^{(j+1)}} +... + \sum_{k=1}^{n-j} \overset{0}{q}_{\left( \frac{1}{k!} ad^k_{-i\Psi}(\Lambda_2^{(j)})\right)^{(n)}}  \right) 
\eeq 
Since $\delta_{\Lambda^+} F_{ru}=0$, the first term vanishes. In particular, 
\beq 
\overset{n-1}{\delta}_{\Lambda_1^{(1)}} \overset{0}{q}_{\Lambda_2^{(n)}} = 0,
\eeq 
which is consistent with the cut-off at order $n$ in $\Psi$. For the second term, we have,
\beq 
\overset{n-1}{\delta}_{\Lambda_1^{(1)}} \overset{0}{q}_{[-i\Psi,\Lambda_2^{(j)}]^{(j+1)}} &=& \overset{n-1}{\delta}_{\Lambda_1^{(1)}} \overset{0}{q}_{[-i\Psi^{(1)},\Lambda_2^{(j)}]} \\
&=& \overset{0}{q}_{-i[\Lambda_1^{(1)},\Lambda_2^{(j)}]}.
\eeq 

As a intermediate step, we show the third term,
\beq 
\overset{n-1}{\delta}_{\Lambda_1^{(1)}} \overset{0}{q}_{[-i\Psi,\Lambda_2^{(j)}]^{(j+2)} + \frac{1}{2}[-i\Psi,[-i\Psi,\Lambda_2^{(j)}]]^{(j+2)}} &=& \overset{n-1}{\delta}_{\Lambda_1^{(1)}} \overset{0}{q}_{[-i\Psi^{(2)},\Lambda_2^{(j)}] + \frac{1}{2}[-i\Psi^{(1)},[-i\Psi^{(1)},\Lambda_2^{(j)}]]} \nonumber\\
&=& \overset{0}{q}_{[-\frac{i}{2}[-i\Psi^{(1)} , \Lambda_1^{(1)}],\Lambda_2^{(j)}] + \frac{1}{2}[-i\Lambda_1^{(1)},[-i\Psi^{(1)},\Lambda_2^{(j)}]]}  \nonumber \\
&& +\overset{0}{q}_{\frac{1}{2}[-i\Psi^{(1)},[-i\Lambda_1^{(1)},\Lambda_2^{(j)}]]} \nonumber\\
&=&  \overset{0}{q}_{[-i\Psi^{(1)} ,-i[\Lambda_1^{(1)},\Lambda_2^{(j)}]] }. \nonumber
\eeq 

In general, we have that each $\Psi^{(1)},...,\Psi^{(n)}$ transforms accordingly to, 
\beq 
\Psi^{(1)} &\mapsto & \Lambda_1^{(1)}, \nonumber\\
\Psi^{(2)} &\mapsto & -\frac{i}{2}[\Psi^{(1)} ,\Lambda_1^{(1)}],\nonumber\\
\Psi^{(3)} &\mapsto & -\frac{i}{2}[\Psi^{(2)} ,\Lambda_1^{(1)}] - \frac{1}{12} [\Psi^{(1)} , [\Psi^{(1)}, \Lambda_1^{(1)}]],\nonumber\\
&& ...,\nonumber
\eeq 
and therefore each term in \eqref{variation_order_1} moves on to the left (using \eqref{Leibnitz_rule_nested} and Bianchi identity). Then:
\begin{equation}
    \overset{n-1}{\delta}_{\Lambda_1^{(1)}} \overset{n-j}{q}_{\Lambda_2^{(j)}} =\begin{cases}
        \overset{n-j-1}{q}_{-i[\Lambda_1^{(1)},\Lambda_2^{(j)}]} \quad \text{if } j<n \\
        \quad 0 \qquad\qquad\qquad\ \, \text{if } j=n
    \end{cases}\;.
\end{equation}

\subsubsection{General expression}

As it was already discussed above, in the general case the action of $\overset{n-k}{q}_{\Lambda^{(k)}}$ on the phase space is
\be 
\{\overset{n-k}{q}_{\Lambda^{(k)} }, \cdot \} = \overset{n-k}{\delta}_{\Lambda^{(k)}}.
\ee

Consider the bracket with $\overset{n-j}{q}_{\Lambda_2^{(j)}}$, for some $0\leq j \leq n$,
\beq \label{variation_order_k}
\overset{n-k}{\delta}_{\Lambda_1^{(k)}} \overset{n-j}{q}_{\Lambda_2^{(j)}} &=& \overset{n-k}{\delta}_{\Lambda_1^{(k)}} \left( \overset{0}{q}_{\Lambda_2^{(j)}} + \overset{0}{q}_{[-i\Psi,\Lambda_2^{(j)}]^{(j+1)}} +... + \sum_{k=1}^{n-j} \overset{0}{q}_{\left( \frac{1}{k!} ad^k_{-i\Psi}(\Lambda_2^{(j)})\right)^{(n)}} \right) 
\eeq 
By the identities in \autoref{sec_ext_phase_space}, each $\Psi^{(1)},...,\Psi^{(n)}$ transforms under $\overset{n-k}{\delta}_{\Lambda_1^{(k)}}$ accordingly to
\beq 
\Psi^{(1)} &\mapsto & 0, \\
\Psi^{(2)} &\mapsto & 0, \\
&& ... \\
\Psi^{(k)} &\mapsto & \Lambda_1^{(k)}, \\
\Psi^{(k+1)} &\mapsto & -\frac{i}{2}[\Psi^{(1)} ,\Lambda_1^{(k)}], \\
\Psi^{(k+2)} &\mapsto & -\frac{i}{2}[\Psi^{(2)} ,\Lambda_1^{(k)}] - \frac{1}{12} [\Psi^{(1)} , [\Psi^{(1)}, \Lambda_1^{(k)}]],\\
&&...,
\eeq 
and therefore each term in \eqref{variation_order_k} moves to the left $k$ steps. Then, we arrive at the following result, for $0\leq j,k \leq n$
\begin{equation}
\label{charge_algebra_up_to_n}
    \{\overset{n-k}{q}_{\Lambda_1^{(k)}} , \overset{n-j}{q}_{\Lambda_2^{(j)}} \} =\begin{cases}
        \overset{n-k-j}{q}_{- i[\Lambda^{(k)}, \Lambda^{(j)}]} \quad \text{if } j+k \leq n \\
        \quad 0 \qquad\qquad\qquad\;\; \text{otherwise}
    \end{cases}\;.
\end{equation}
With this result, we have a closed Poisson algebra $\mathfrak{p}_n$ such that the action of the charges is canonical and we have the following chain
\be 
\mathfrak{p}_0 \subset ... \subset \mathfrak{p}_{n-1} \subset \mathfrak{p}_n .
\ee 
 
\section{Charges as corner terms - comparison with previous results} \label{sec_charges_corners}

In this brief section, we are connecting the recursive formulas in \autoref{eom_scri_plus} with the charges constructed in \autoref{charges}. We show that subsets of our charges match with previous results. In particular, we can match our $n=1$ level with the one in the literature for sub-leading charges in Yang-Mills \cite{Casali:2014xpa,Campiglia:2021oqz}, and the complete hierarchy of charges for the abelian case \cite{Campiglia:2018dyi,Hamada:2018vrw,Li:2018gnc, Peraza:2023ivy}.

We first provide the usual corner interpretation for the leading charge. By the fall-offs in \autoref{u_fall_offs_light_cone_gauge}, $F_{ru}^{(-2)}(u,z,\zb) = F_{ru}^{(-2)} + 0(u^{-\infty})$, we have that,
\beq
\overset{0}{Q}_{\Lambda^{(0)}} = \int_{S^2} \overset{0}{q}_{\Lambda^{(0)}}dS_{S^2} &=&  \int_{S^2}\tr\left( \Lambda^{(0)} F_{r u}^{(-2)} \right) dS_{S^2} \\
&=& \int_{S^2}  \tr \left( \Lambda^{(0)} F_{r u}^{(-2,0)} \right) dS_{S^2}.
\eeq
This is the standard leading charge (cf. \cite{Strominger:2013lka}).  

 Next, we compare the results for the sub-leading charge with the ones obtained in \cite{Campiglia:2021oqz}. We rewrite equations \eqref{q_1u_0d} as follow,
\beq
\overset{0}{q}_{\Lambda^{(1)}} &=& \tr \left( \Lambda^{(1)} F_{r u}^{(-3)} \right), \\
\overset{1}{q}_{\Lambda^{(0)}} &=& \overset{0}{q}_{\Lambda^{(0)}} + \overset{0}{q}_{[-i \Psi^{(1)}, \Lambda^{(0)}]}. \label{rec_q_from_first_order}
\eeq

In the prescription given in \cite{Campiglia:2021oqz}, $\Psi^{(1)}$ acts as the generator of the linear extension for the radiative phase space, 
\be 
\Gamma^{ext}_1 := \{ \tilde{\A}_\mu = \A_\mu + D_\mu \Psi^{(1)}, \quad \A_\mu \in \Gamma^{rad} , \quad  \Psi^{(1)} \in C^{\infty}(S^2) \},
\ee
where $\Gamma^{rad}$ is the radiative phase space. Observe that the leading charge, \eqref{rec_q_from_first_order}, has the same expression as in equation (4.21) in \cite{Campiglia:2021oqz}, by restoring the coefficients $i$ in the reference. For the sub-leading charge, we have to use the recursive relations for $F_{ru}^{(-3)}$.

Let denote $D^-_z := \partial_z - i[A_z^{(0,0)}, \cdot ]$, which is the covariant derivative with respect to the $u\rightarrow -\infty$ limit in $A_z^{(0)}$. Then, $A_z^{(-1)}$ can be written as (cf. \eqref{A_z_-1_computed}),
\be 
A_z^{(-1)}= A_z^{(-1,0)} -\frac{1}{2} D^-_z A_r^{(-2,0)} + \frac{1}{2}u D^-_z(F_{ru}^{(-2,0)} + \gamma^{-1} F_{z\zb }^{(0,0)}) + o(u^{-\infty}).
\ee
From here, a direct computation gives,
\be 
F_{uz}^{(-1)} = \frac{1}{2} D^-_z(F_{ru}^{(-2,0)} + \gamma^{-1} F_{z\zb }^{(0,0)}) + o(u^{-\infty}),
\ee
and analogous for $F_{u\zb}^{(-1)}$. From \eqref{F_ru_3}, we have,
\be 
\partial_u F_{ru}^{(-3)} = \gamma^{-1} D^-_{(z} D^-_{\zb)} (F_{ru}^{(-2,0)} + \gamma^{-1} F_{z\zb }^{(0,0)}) + o(u^{-\infty}).
\ee
This is exactly the integrated (modulus the weight u) quantity in equation (4.17) in \cite{Campiglia:2021oqz}, where the authors use the \textit{integrated} version of the charges, along $\Ib^+$. To relate the integrated charges with the ones here, one should impose the boundary conditions 
\be 
F_{ru}^{(-2)}(u,z,\zb),F_{z\zb }^{(0)}(u,z,\zb) \sim o(u^{-\infty}), \quad \text{as} \quad u \rightarrow + \infty,
\ee
and integrate by parts using $\partial_u^{-1} (u\partial_u X) = u X - \partial_u^{-1} X $. Therefore, we arrive at the same result for the sub-leading charge.

As an interesting corollary of our construction, observe that for the abelian case we can immediately see the coefficients of the field $\Psi$ as the Goldstone modes associated to the sub-leading charges, by constructing a symplectic form which satisfies \eqref{up_to_n_symplectic_form} and such that the modes $\{\Psi^{(i)}\}_{i\geq 1}$ live on the boundary $\Ib^+_-$. In accordance with \cite{Peraza:2023ivy}, for a fixed $n\geq0$ we have to evaluate the abelian version of \eqref{symplectic_form_bulk} at null infinity and after a renormalization in both $r$ and $u$,
\be 
\Omega^{n}(\delta_1 , \delta_2) = \Omega^{rad}(\delta_1 , \delta_2) + \sum_{i=1}^{n}\int_{S^2}  \delta_1 F_{ru}^{(-2-i,0)} \wedge \delta_2 \Psi^{(i)} \gamma dz d\zb,
\ee
where $\Omega^{rad}(\delta_1 , \delta_2) = \int_{\Ib^+}  \delta_1 \partial_u A^{(z} \wedge \delta_2 A^{\zb)} du dz d\zb $ is the usual radiative phase space symplectic form. Of course, the analogue of this symplectic form in Yang-Mills theory also implies \eqref{up_to_n_symplectic_form}, but $F_{ru}^{(-2-i,0)}$ and $\Psi^{(i)}$ are no longer canonical conjugates (due to the action of the variations on $\Psi$).

\section{Relations with infinite-dimensional algebras}
\label{Relations with infinite algebras}

In \cite{Freidel:2023gue}, a subset of the charge aspects of YM theory were shown, up to quadratic order in the creation and annihilation operators, to satisfy the infinite dimensional YM version of the $w_{1+\infty}$ algebra \cite{Strominger:2021mtt}
\be \label{S_inf_alg_YM}
\left[S_{m}^{p,a}(\zb),S_{n}^{q,b}(\zb)\right]
=-i\ f^{ab}_{\ \ c}  S_{m+n}^{p+q-1,c}(\zb)
\ee 
where $1-p\leq m \leq p-1$, $p,q$ are half-integers with $p,q>1$, and $m+p\in \mathbb{Z}$ (similarly for $(q,n)$).

Crucial to find \eqref{S_inf_alg_YM} were the recursion relations for the charge aspects. Working in flat Bondi coordinates for simplicity (see \autoref{YM e.o.m. in flat Bondi coordinates}), and in radial gauge with 
\be 
\A_r=0, \quad \text{and \quad} A_u^{(0)}=0, 
\ee
these relations can be written as\footnote{The difference in sign is due to converting to our conventions for the commutator.}: 
\be 
\label{Freid_rec}
\partial_u \mathcal{R}_s=\partial_z\mathcal{R}_{s-1} 
-i[A_z^{(0)},\mathcal{R}_{s-1} ]=D_z^{(0)}\mathcal{R}_{s-1} 
\ee 
where $\mathcal{R}_s$ are spin charge aspects, the first few given explicitly by:
\be \label{Freid_charge_aspects}
\mathcal{R}_{-1}=F_{\zb u}^{(0)},\quad
\mathcal{R}_{0}= \tfrac{1}{2}\left(F_{ru}^{(-2)}+F_{\zb z}^{(0)}\right),\quad
\mathcal{R}_{1}=F_{rz}^{(-2)}
\ee 
The generators in \eqref{S_inf_alg_YM} are constructed, after a renormalisation procedure, by smearing the charge aspects with a Lie algebra valued function on the sphere\footnote{See \cite{Freidel:2023gue} for details.}. 

First, observe that we can consider $F_{uz}^{(0)}$ as a lowest \textit{weight} field on the sphere, since its limit when $u\rightarrow -\infty$ vanishes, see \autoref{table_fall_off}. The next weight corresponds to $O(1)$, and this is indeed the $u$ fall-off of $F_{ru}^{(-2)}$ and $F_{\zb z}^{(0)}$. Finally $F_{rz}^{(-2)}$ will be quadratic in $u$, as detailed in \autoref{table_fall_off}\footnote{The fall-off in $u$ is the same in both standard Bondi and flat Bondi coordinates.}. 

We note that the recursions above arise as a special case of the recursion relations we give in \autoref{YM e.o.m. in flat Bondi coordinates}. At $s=0$, we have, from \eqref{Freid_rec} and \eqref{Freid_charge_aspects}: 
\be 
\label{Freidel s=0}
\partial_u F_{ru}^{(-2)}+\partial_u F_{\zb z}^{(0)}=2\left(
\partial_z F_{\zb u}^{(0)}-i[A_z^{(0)},F_{\zb u}^{(0)}]\right)
\ee 
This follows from our relation \eqref{Fur_recs flat}, can be rearranged in the form of \eqref{Freidel s=0} as below:  
\begin{equation}
    \begin{aligned}
        \partial_uF_{ur}^{(-2)}&=2\left(\partial_{(z}\partial_uA_{\bar{z})}^{(0)}-i[A_{(z}^{(0)},\partial_uA_{\bar{z})}^{(0)}]\right)\\
        &=2\left(\partial_{z}F_{u\bar z}^{(0)}-i[A_{z}^{(0)},F_{u\bar z}^{(0)}]\right)-\partial_uF_{z\bar z}^{(0)}
    \end{aligned}
\end{equation}
At $s=1$ we have 
\be 
\partial_u F_{rz}^{(-2)}=\tfrac{1}{2}\left(\partial_z F_{ru}^{(-2)}
+\partial_z F_{\zb z}^{(0)}-i[A_z^{(0)},F_{ru}^{(-2)}+F_{\zb z}^{(0)}]\right)
\ee 
This immediately follows from our equation \eqref{F_rz recursion flat} with $n=2$ 
\begin{equation}
    \begin{aligned}
        2\partial_uF_{rz}^{(-2)}&=-\partial_z(F_{ur}^{(-2)}+F_{z\bar z}^{(0)})+i[A_z^{(0)},F_{ur}^{(-2)}+F_{z\bar z}^{(0)}] \ .
    \end{aligned}
\end{equation}

Higher orders are then reached via repeatedly acting with the operator $\partial_u^{-1} \cdot D_z^{(0)}$, see  \eqref{Freid_rec}. We remark that this tower of charges, which are linear in the anti-holomorphic component  $A_{\zb}^{(0)}$ and of arbitrary order in the holomorphic components $A_z^{(0)}$ are closely related to a self-dual subsector of the theory. One could of course also construct the anti-holomorphic tower, by exchanging $z$ and $\zb$
\be \label{Freid_charge_aspects_conj}
\bar{\mathcal{R}}_{-1}=F_{z u}^{(0)},\quad
\bar{\mathcal{R}}_{0}= \tfrac{1}{2}\left(F_{ru}^{(-2)}+F_{z \zb}^{(0)}\right),\quad
\bar{\mathcal{R}}_{1}=F_{r\zb}^{(-2)}
\ee 
and postulating the recursion relation 
\be 
\partial_u \bar{\mathcal{R}}_s=\partial_{\zb}\bar{\mathcal{R}}_{s-1} 
-i[A_{\zb}^{(0)},\bar{\mathcal{R}}_{s-1} ]=D_{\zb}^{(0)}\bar{\mathcal{R}}_{s-1} 
\ee 
Again, the first two steps will follow from our \eqref{Fur_recs flat} and the conjugate of \eqref{F_rz recursion flat}, with all higher orders following from the action of $D_{\zb}$. Of course this (anti-self-dual) subsector will also possess a copy of the infinite-dimensional symmetry algebra. If would be interesting to see whether the commutation of charges between the towers could belong to some  deformation of these infinite algebras. We note that similar considerations apply for the gravity theory, see \cite{Geiller:2024bgf}.

\section{Applications to higher derivative interactions} \label{higher_oerder_der_int}

In this section we present the sub-leading expression when a term containing higher order derivatives of the vector field is present in the Lagrangian. The discussion of this type of interactions follows from the work of Elvang and Jones \cite{Elvang:2016qvq}, and is motivated by the study of loop corrections arising from massive particles. 

It is shown in \cite{Elvang:2016qvq} that a small set of effective (local) operators can modify the soft photon and graviton theorems. This introduces \textit{quasi-universal} terms, in the sense that the soft factors are described in terms of a small set of theory-dependent spin-shifting operators, with a theory-independent kinematic factor. The connection with the asymptotic symmetries was studied by Laddha and Mitra \cite{Laddha:2017vfh}, where they extend the known equivalence between asymptotic symmetries and sub-leading soft photon theorem to the family of effective operators found in \cite{Elvang:2016qvq}. 

A divergent gauge parameter is used in \cite{Laddha:2017vfh} to derive both the universal and quasi-universal terms of \cite{Elvang:2016qvq}. Here we extend the phase space in order to accommodate the symmetry transformation associated with this parameter. As a straightforward exercise, we also present the modifications to the charge from the quasi-universal parts in the sub-leading soft gluon theorem using our prescription for deriving the charges. Starting with the lagrangian from \cite{Elvang:2016qvq,Laddha:2017vfh}, we use our prescription to dress all the fields
\be 
\mathcal{L}[\tilde{\A},\tilde{\phi}, ...] = - \frac{1}{4} \tr \left( \tilde{\F}_{\mu \nu} \tilde{\F}^{\mu \nu} \right) + \mathcal{L}^m_{kin}[\tilde{\phi},...] + \mathcal{L}_{int}[\tilde{\A},\tilde{\phi},...],
\ee
where $\tilde{\phi}$ are the dressed matter fields with suppressed indices. The dressing of the matter fields is done in analogous way as in \eqref{stueck_field_indep}, i.e. we do a finite gauge transformation on $\phi$ and promote the parameter to our Stueckelberg field $\Psi$, resulting in the new field $\tilde{\phi}$. As an example, for a field in the adjoint representation we have
\be 
\phi \mapsto \tilde{\phi} = e^{i \Psi} \phi e^{- i \Psi}.
\ee
The kinematic terms $\mathcal{L}^m_{kin}$ depend only on the matter fields, and the interaction between the gauge field and the matter fields is given by
\be 
\mathcal{L}_{int}[\tilde{\A},\tilde{\phi}, ...] = - \tr \left(  \tilde{\A}^\mu J_\mu^m [\tilde{\A},\tilde{\phi}] \right) + \mathcal{L}_{int}^{HO}[\tilde{\A},\tilde{\phi}, ...],
\ee
where $\mathcal{L}_{int}^{HO}$ denotes the higher order derivative interaction part. The complete list of possible effective field operators that can enter this term is given in \cite{Elvang:2016qvq}. The extended phase space that we are considering is \eqref{extended_phase_space_YM} times the particular phase space coming from the matter fields in $\mathcal{L}^m_{kin}$,
\be 
\Gamma^{\text{ext}}_{\infty} \times \Gamma^{\text{matter} }.
\ee

By taking the variation of the action and computing the Euler-Lagrange equations, we have the general form
\be 
\tilde{D}^{\mu} \tilde{\F}_{\mu \nu}  + \tilde{D}^{\mu} G_{\mu\nu} [\tilde{\A},\tilde{\phi}]= J^m_{\nu}[\tilde{\A},\tilde{\phi}],
\ee
where $\tilde{D}^{\mu}$ is the $\tilde{A}_\mu$ associated gauge derivative and $G_{\mu\nu}[\tilde{\A}, \tilde{\phi}]$ is the term resulting from varying $\mathcal{L}_{int}^{HO}$.

The symplectic potential for a Cauchy slice $\Sigma_t$ is given by 
\be 
\Theta [\delta] = \int_{\Sigma_t} \tr \left( (\tilde{\F}^{\mu\nu} + G^{\mu\nu}[\tilde{\A},\tilde{\phi}]) \delta \tilde{\A}_\nu \right) dS_\mu + \Theta_m [\delta],
\ee
where $\Theta_m$ is the symplectic potential from the matter sector. The symplectic form is
\be 
\Omega [\delta,\delta'] = \int_{\Sigma_t} \tr \left(\delta(\tilde{\F}^{\mu\nu} + G^{\mu\nu}[\tilde{\A},\tilde{\phi}]) \wedge \delta' \tilde{\A}_\nu \right) dS_\mu + \Omega_m [\delta, \delta'],
\ee
with $\Omega_m$ coming from $\Theta_m$. Given an arbitrary gauge transformation $\delta_\Lambda$, its associated charge is computed via the equation, cf. \eqref{def_int_charge}, 
\be 
\delta  \tilde{Q}_\Lambda = \Omega [\delta_\Lambda , \delta].
\ee
In general, the $\delta$-integrability of the previous equation is not guaranteed \cite{Brown:1986ed}, since the specific form of $G[\tilde{\A},\tilde{\phi}]$ depends on the particular effective operator on $\mathcal{L}^{HO}_{int}$. For simplicity, we consider the case
\be 
\mathcal{L}^m_{kin} = - \frac{1}{2} \partial_\mu \phi \partial^\mu \phi ,\quad \mathcal{L}_{int}^{HO} = - \frac{1}{4} \phi \tr \left( \tilde{\F}_{\mu \nu} \tilde{\F}^{\mu \nu} \right),
\ee
where the scalar field $\phi$ is in the trivial representation (i.e., $\tilde{\phi} = \phi$), which could correspond, for example, to a first order effective action of a dilaton coupled to Yang-Mills (see e.g.  \cite{Lavrelashvili:1992cp,Gross:1986mw}). Then, we have a simple expression for $G_{\mu \nu}$,
\be 
G_{\mu \nu} [\tilde{\A},\phi]=  \phi \tilde{\F}_{\mu \nu},
\ee
and also for $\Theta_m[\delta]$,
\be 
\Theta_m [\delta] = \int_{\Sigma_t} \partial^\mu \phi \delta \phi dS_\mu.
\ee
Now, we can use the equations of motion to compute the charge, $\delta$-integrate its expression, and obtain
\be 
\tilde{Q}_\Lambda = \int_{\Sigma_t} \partial_\nu \tr \left( \Lambda (\tilde{\F}^{\mu \nu}  +  G^{\mu\nu}[\tilde{\A},\phi]) \right)   dS_{\mu}.
\ee
Then we have the same expression as in \eqref{full_ch} plus a term which represents the contribution from the higher order derivatives,
\be 
\tilde{Q}^{new}_\Lambda = \int_{\Sigma_t} \partial_\nu  \tr \left( \Lambda  G^{\mu\nu}[\tilde{\A},\phi]  \right)  dS_{\mu}.
\ee
We want to re-obtain the sub-leading corrections, for which we need to consider variations with $\Lambda^{(1)} \neq 0$ and $\Lambda^{(l)} = 0$ for $l\geq 2$, see \cite{Laddha:2017vfh}. Correspondingly, we have to take the linear extension of the phase space with $\Psi^{(1)} \neq 0$ and $\Psi^{(l)}=0$ for $l \geq 2$. The extension allows us to interpret the $\Lambda^{(1)}$ as the parameter of a genuine symmetry transformation, realised canonically on $\Psi^{(1)}$, and acting as
\be 
\delta_{\Lambda^{(0)}} \Psi^{(1)} = -i[\Psi^{(1)}, \Lambda^{(0)}], \quad \delta_{\Lambda^{(1)}} \Psi^{(1)} = \Lambda^{(1)}.
\ee
The radiative fall-off for the scalar field is $\phi = \frac{1}{r} \phi^{(-1)} + o(r^{-1})$, and therefore the first terms in the expansion for $G_{ru}$ are
\be 
G_{ru} = \frac{1}{r^2} \phi^{(-1)} \tilde{\F}^{(-1)}_{ru} + \frac{1}{r^3} \left( \phi^{(-1)} \tilde{\F}^{(-2)}_{ru} + \phi^{(-2)} \tilde{\F}^{(-1)}_{ru} \right) + o(r^{-3}).
\ee
We have two contributions to the new sub-leading charge, from $\Lambda^{(0)}$ and $\Lambda^{(1)}$. Since $\tilde{\F}^{(-1)}_{ru}$ is already linear in $\Psi^{(1)}$, for the $\Lambda^{(1)}$ contribution those terms are discarded. Then, upon renormalization, we are left with 

\be 
Q^{new}_{\Lambda^{(1)}} = \int_{\Ib} \tr  \left( \Lambda^{(1)}  \partial_u (\phi^{(-1)} F^{(-2)}_{ru})  \right)  dS^2,
\ee
\be 
Q^{new}_{\Lambda^{(0)}} = \int_{\Ib} \partial_u \tr  \left( \Lambda^{(0)}  \phi^{(-1)} [i\Psi^{(1)} ,F^{(-2)}_{ru}] \right)  dS^2,
\ee
where we use that $\Lambda^{(1)}$ is independent of $u$. We can identify the first term above as the contribution to the soft theorem from the higher order derivative effective operators found in \cite{Laddha:2017vfh} in electrodynamics, but written in Yang-Mills theory. The explicit map from our expression to that in  \cite{Laddha:2017vfh} is via the equations of motion. The second term above corresponds to a new contribution to the sub-leading charge for Yang-Mills coming from our extended phase space.

We note that effective actions with higher derivative terms also arise when considering loops of the gluon itself, and these were studied in the simplified context of SDYM in \cite{Costello:2021bah,Costello:2022wso,Monteiro:2022nqt}. It would be interesting to apply our formalism to their work. We leave this for future study.

\section{Conclusions} \label{Conclusions}
 
Expanding on our proposal in \cite{Nagy:2024dme}, we have shown how the phase space at null infinity can be enlarged in order to accommodate sub$^n$-leading charges, associated to the soft theorems for Yang-Mills theory. The procedure is independent of gauge and coordinate choices, and we have allowed for a very general fall-off in the expansion coordinate, thus making our procedure potentially applicable to both tree- and loop-level effects.   

We also provided expressions for the equations of motion for Yang-Mills theory in the $r$-expansion, order by order, in standard Bondi coordinates. As an example, we gave the explicit equations in both radial and light-cone gauges using the usual radiative fall-offs compatible with tree-level. In both cases, we showed the existence of a recursive relation for the sub-leading coefficients to all orders in the $r$-expansion for the field strength and the gauge vector in terms of the leading component, which establishes the usual convention of taking $A_z^{(0)}(u,z,\zb)$ as free initial data. Of course, it is not clear whether it is possible to obtain such recursive formulas for other gauge choices. We leave this question for future work.

We generalized the Stueckelberg procedure\footnote{See e.g. \cite{Stueckelberg:1938hvi,Henneaux:1989zc,Nagy:2019ywi,Bansal:2020krz}.} from our previous construction in the context of self-dual Yang-Mills (SDYM) \cite{Nagy:2022xxs} to full Yang-Mills. The Stueckelberg trick has previously been employed in seemingly very different contexts, but the unifying principle is the presence of a local broken symmetry, though in our case this is subtly related to the radial expansion at null infinity. In particular, the Goldstone-like modes that constitute the expansion of the Stueckelberg field naturally serve as the indicators of the power in the radial expansion, since they serve as the target of the action of over-leading large gauge transformations acting on the whole phase space. This, in turn, allowed us to construct an extended phase space that contains a clear hierarchy of subspaces. The tangent subspaces (where variations live) are determined by the successive terms in the perturbative expansion of the characteristic power series generating the so-called Todd polynomial. An interesting insight here is that the canonical action of the charges within the hierarchy \eqref{canonical_action_charge} is controlled by the Bernoulli numbers $B_n$. Since $B_{2k+1}$ are vanishing for $k>0$, it appears that the canonical action in higher levels is controlled by the even sub$^n$-leading charges. We will explore this further in future work. 

The charges are related to the $r$-expansion of the field strength (we assume renormalization, cf. (\cite{Peraza:2023ivy}), and can be computed directly. The charge algebra is closed for each phase space. Using the recursion formulas derived in \autoref{eom_scri_plus}, we compared the formulas of the sub-leading charge with those in the literature \cite{Campiglia:2021oqz, Casali:2014xpa}. The charges can also be constructed via an algebraically simple recursive relation. 

Several interesting future lines of work emerge from our results. On the one hand, the relation between over-leading large gauge transformations and the symplectic structure of the phase space could shed light on the symmetries of the theories, and therefore provide a deep understanding of the classical starting point for a possible quantization of the theory. A natural question is whether this extends to gravity, where work at the first few sub-leading orders already exists in some contexts (e.g., \cite{Cachazo:2014fwa, Campiglia:2021oqz, Freidel:2021dfs, Fuentealba:2022xsz}), see also results in the Newman-Penrose formalism \cite{Geiller:2024bgf}. Our construction is particularly well-suited for gravity, via the finite action of a ``Stueckelberg'' diffeomorphism $\xi$ acting on the metric $g$,
\be 
\breve{g} = e^{\mathcal{L}_\xi} g.
\ee
We expect this extension to be more involved than the Yang-Mills case, due to the non-linearities in both $\xi$ and $g$ in the definitions of the gravitational charges.

 In \cite{Nagy:2022xxs}, the self-dual sector of gravity was considered at all sub-leading orders, and a simple \textit{double copy} dictionary was constructed from Yang-Mills to gravity, for the extended radiative phase space, and the tower of symmetries. A natural question prompted by our results here is whether this can be extended to the full YM and gravitational theories\footnote{For some previous work on double copy relations at null infinity, see also \cite{Campiglia:2021srh,Adamo:2021dfg,Ferrero:2024eva,Mao:2021kxq,Godazgar:2021iae}}.

 On the other hand, understanding the symmetries that govern the properties of the S-matrix is crucial for establishing the conjectured holographic principle for asymptotically flat spacetimes (e.g., \cite{Pasterski:2021raf,Raclariu:2021zjz} and references therein). Recently, studies in some sub-sectors of gravity (e.g. \cite{Freidel:2021ytz}) and in the celestial holography program (e.g. \cite{Guevara:2021abz, Strominger:2021mtt}) have shown the presence of infinite dimensional algebras, known as $w_{1+\infty}$ algebras. It would be interesting to see how these fit in the extended gravity case via our proposal. In this direction, from a geometric point of view and on more speculative grounds, it would be interesting to explore the connection between the action of large diffeomorphisms and the universal corner algebra $\mathfrak{diff}(S^2) \times \mathfrak{gl}(2,\R) \times \R^2  $ \cite{Ciambelli:2022cfr}. In particular, a possible extended graded algebraic structure with the subalgebra $\mathfrak{diff}(S^2) \times \mathfrak{gl}(2,\R)$ as base level could determine the emergence of the $w_{1 + \infty}$ algebra in a clear geometric picture.
 
In the context of scattering amplitudes, the next step is to directly apply the ideas in this article to the calculation of sub$^n$-leading soft theorems at loop level \cite{He:2014bga, Bianchi:2014gla,Sahoo:2018lxl, Pasterski:2022djr, Donnay:2022hkf,Agrawal:2023zea,Choi:2024ygx,Campiglia:2019wxe,AtulBhatkar:2019vcb}, via the Ward identities. If we only consider loops arising from massive matter particles, these can be encoded in higher derivative terms appearing in the effective action \cite{Elvang:2016qvq}. These lead to \emph{quasi-universal} corrections to the sub-leading soft theorem. In QED, these corrections were shown to arise from the Ward identity of an over-leading gauge parameter \cite{Laddha:2017vfh}. We extended this analysis to Yang-Mills, and also provided the extended phase space on which the symmetry responsible for these corrections acts canonically. 

Of course, a more interesting question involves considering loops of the gauge fields themselves. A simplified set-up which holds promise as a starting point for going to arbitrary orders in the soft expansion is the self-dual sector. This also has the benefit of being one-loop exact, for both Yang-Mills and gravity \cite{Bern:1993qk,Mahlon:1993si,Bern:1996ja,Bern:1998xc,Monteiro:2022nqt}, and it preserves the infinite dimensional algebras above at loop-level \cite{Ball:2021tmb,Mago:2021wje,Banerjee:2023jne,Monteiro:2022xwq}. A promising direction here could be to use the results of \cite{Monteiro:2022nqt}, which give a quantum corrected action encoding the effective vertices after loop integration. This is based on the earlier observation in \cite{Costello:2021bah,Costello:2022wso} that the contribution of diagrams with gluons in the loops in SDYM can be cancelled by the tree-level contribution arising from the addition of an axion interaction; interestingly, the interaction term is of the form we studied in \autoref{higher_oerder_der_int}. 

More generally, the study of loop effects from asymptotic symmetries requires the introduction of a formalism for massive particles (see e.g. \cite{Campiglia:2019wxe}), based on hyperbolic coordinates (namely Euclidean AdS$_{3}$) allowing us to approach future/past timelike infinity $i^{\pm}$. We remark that the general procedure in \autoref{sec_ext_phase_space}, by virtue of being coordinate and gauge independent, should be straightforwardly adaptable to massive particles. Finally, for gravity and QED, the leading order (logarithmic) loop correction turns out to be controlled by the same asymptotic symmetry parameter as the sub-leading tree level correction \cite{Sahoo2019,Choi:2024ygx, Campiglia:2019wxe}, which has the potential to simplify the extension of our construction to loop corrections.

\acknowledgments
 
We thank Miguel Campiglia, Marc Geiller and Sam Wikeley for useful discussions. We also want to thank the anonymous referees of the short paper \cite{Nagy:2024dme} for their useful comments, which allowed us to connect our work with different parts of the vast literature in asymptotic symmetries and soft theorems. G.P. is funded by STFC Doctoral Studentship 2023. S.N. is supported in part by STFC consolidated grant T000708. J.P. was partially funded by Fondo Clemente Estable Project FCE\_1\_2023\_1\_175902 and by CSIC Group 883174. Currently funded by a Postdoctoral Fellowship at Concordia University, Montreal.

\appendix

\section{Extended phase space formulae}\label{Extended phase space calculations}
Equation \eqref{LHS_consistency}:
\begin{equation}
\label{deltaLambda_A}
    \begin{aligned}
        \delta_{\breve\Lambda}\tilde{\mathcal{A}}_\mu &= \delta_{\breve\Lambda}(e^{i\breve\Psi}\mathcal{A}_\mu e^{-i\breve\Psi}+ie^{i\breve\Psi}\partial_\mu e^{-i\breve\Psi}) \\
        &= \delta_{\breve\Lambda}(e^{-X}\mathcal{A}_\mu e^X+ie^{-X}\partial_\mu e^X) \\
        &= -\mathcal{O}_X(\delta_{\breve\Lambda} X)e^{-X}\mathcal{A}_\mu e^X+e^{-X}(\delta_{\breve\Lambda}\mathcal{A}_\mu)e^X+e^{-X}\mathcal{A}_\mu e^X\mathcal{O}_X(\delta_{\breve\Lambda} X) \\
        &\quad -i\mathcal{O}_X(\delta_{\breve\Lambda} X)e^{-X}\partial_\mu e^X+ie^{-X}\partial_\mu\left[e^X\mathcal{O}_X(\delta_{\breve\Lambda} X)\right] \\
        &= e^{-X}\big\{\delta_{\breve\Lambda}\mathcal{A}_\mu+\mathcal{A}_\mu e^X\mathcal{O}_X(\delta_{\breve\Lambda} X)e^{-X}-e^X\mathcal{O}_X(\delta_{\breve\Lambda} X)e^{-X}\mathcal{A}_\mu \\
        &\quad \quad\quad +ie^{X}\mathcal{O}_X (\delta_{\breve\Lambda} X)\partial_\mu e^{-X}+i\partial_\mu\left[e^X\mathcal{O}_X(\delta_{\breve\Lambda} X)\right]e^{-X}\big\}e^X \\
        &= e^{i\breve\Psi}\left\{\delta_{\breve\Lambda}\mathcal{A}_\mu+D_\mu\left[e^{-i\breve\Psi}\mathcal{O}_{-i\breve\Psi}(\delta_{\breve\Lambda}\breve\Psi)e^{i\breve\Psi}\right]\right\}e^{-i\breve\Psi}\;,
    \end{aligned}
\end{equation}
where in the second line we defined $X:=-i\breve\Psi$ and we used
\ba
\delta e^{-X} &=& -\mathcal{O}_X(\delta X)e^{-X} \label{delta_e_neg}\\
e^{-X}\partial_\mu e^{X} &=& -(\partial_\mu e^{-X})e^X\;.
\ea
Equation \eqref{RHS_consistency}:
\begin{equation}
\label{hatD_Lambda}
    \begin{aligned}
        \tilde{D}_\mu\breve\Lambda &= \partial_\mu\breve\Lambda-i[\tilde{\mathcal{A}}_\mu,\breve\Lambda] \\
        &= \partial_\mu\breve\Lambda-i[e^{i\breve\Psi}\mathcal{A}_\mu e^{-i\breve\Psi}+ie^{i\breve\Psi}\partial_\mu e^{-i\breve\Psi},\breve\Lambda] \\
        &= e^{i\breve\Psi}\partial_\mu(e^{-i\breve\Psi}\breve\Lambda e^{i\breve\Psi})e^{-i\breve\Psi}-ie^{i\breve\Psi}[\mathcal{A}_\mu,e^{-i\breve\Psi}\breve\Lambda e^{i\breve\Psi}]e^{-i\breve\Psi} \\
        &= e^{i\breve\Psi}D_\mu(e^{-i\breve\Psi}\breve\Lambda e^{i\breve\Psi})e^{-i\breve\Psi}\;.
    \end{aligned}
\end{equation}

Next, we provide some of the formulae which are used throughout the text. The proof of each identity is straightforward. First, the usual trace identities,
\beq 
\tr\left( A [B,C]\right) &=& \tr \left( [A, B] C\right),\label{tr_formula} \\
\tr \left( A ad_B^k (C)\right) &=& (-1)^k \tr \left( ad_B^k (A) C\right) \label{ad_formula}.
\eeq

Next, we provide a few representations for the operator $\G$ when evaluated at $\Av_\mu$, i.e., $\G_{\breve\Psi} (\Av_\mu)$.
\beq
\G_{\breve\Psi} (\Av_\mu) &=& e^{i\breve\Psi} \Av_\mu e^{- i \breve\Psi} + i e^{i\breve\Psi} \partial_\mu e^{- i \breve\Psi} \\
&=& \Av_\mu - D_\mu i\breve\Psi + \frac{1}{2} [D_\mu i\breve\Psi , i\breve\Psi] - \frac{1}{3!} [[D_\mu i\breve\Psi, i\breve\Psi], i \breve\Psi] + \dots \\
&=& \Av_\mu + \sum_{k=1}^{+\infty} \frac{1}{k!} ad^{k-1}_{i\breve\Psi} (D_\mu i\breve\Psi) \\
&=& \Av_\mu + \Op_{-i\breve\Psi }(D_\mu i\breve\Psi)\;.
\eeq

\section{YM e.o.m. in radial expansion} \label{app_YM_eq_r_exp}

In order to study \(\mathcal{T}\) at a generic order \(n\) in the \(1/r\)-expansion, it is useful to introduce the following map:
\begin{equation}
\label{def:projector P}
\begin{aligned}
\mathcal{P}_{-n}\colon C_r^\infty(M)&\longrightarrow C_r^\infty(M)\\
f&\longmapsto r^{-n}f^{(-n)}
\end{aligned}
\end{equation}
It is easy to see that \(\mathcal{P}_{-m}\circ\mathcal{P}_{-n}(f)=\delta_{mn}\mathcal{P}_{-n}(f)\), for all \(m,n\in\mathbb{Z}\) and for all \(f\in C^\infty_r(M)\), meaning that the map \(\mathcal{P}_{-n}\) is a projector. Moreover, as a consequence of the definition of the set \(C^\infty_r(M)\), we have that \(\mathcal{P}_{-n}(f)=0\) for all \(n<0\), which is consistent with equation \eqref{eq:generic tensor expansion}.

\begin{lemma}
    \(\forall m,n\in\mathbb{N}\) and \(\forall f,g\in C^\infty_r(M)\) the following identities hold:
    \begin{align}
    \label{prop:projector_1}
    \mathcal{P}_{-n}(r^{-m}f)&=r^{-m}\mathcal{P}_{m-n}(f)\\
    \label{prop:projector_2}
    \mathcal{P}_{-n}(\partial_r f)&=\partial_r\mathcal{P}_{1-n}(f) \\
    \label{prop:projector_3}
    \mathcal{P}_{-n}(fg)&=\sum_{k\in\mathbb{N}}\mathcal{P}_{k-n}(f)\mathcal{P}_{-k}(g)
    \end{align}
\end{lemma}
\begin{proof}
    \begin{align}
        \mathcal{P}_{-n}(r^{-m}f)&=\mathcal{P}_{-n}\sum_{k\in\mathbb{Z}}r^{-m-k}f^{(-k)}=r^{-n}f^{(m-n)}=r^{-m}\mathcal{P}_{m-n}(f) \\
        \mathcal{P}_{-n}(\partial_rf)&=\mathcal{P}_{-n}\sum_{k\in\mathbb{Z}}(-k)r^{-k-1}f^{(-k)}=(1-n)r^{-n}f^{(1-n)}=\partial_r\mathcal{P}_{1-n}(f) \\
        \mathcal{P}_{-n}(fg)&=\mathcal{P}_{-n}\sum_{k,l\in\mathbb{Z}}r^{-k-l}f^{(-k)}g^{(-l)}=r^{-n}\sum_{k\in\mathbb{N}}f^{(k-n)}g^{(-k)}=\sum_{k\in\mathbb{N}}\mathcal{P}_{k-n}(f)\mathcal{P}_{-k}(g)
    \end{align}
\end{proof}

Using this lemma, one can prove equation \eqref{eq:YM EOM at every order in general Bondi}:

\begin{proof}
    From \eqref{def:EOM}:
    \begin{align}
        \mathcal{E}_\mu&=g^{\nu\rho}D_\rho\mathcal{F}_{\nu\mu} \\
        &=-2D_{(u}\mathcal{F}_{r)\mu}+D_r\mathcal{F}_{r\mu}+2r^{-2}\gamma^{-1}D_{(z}\mathcal{F}_{\Bar{z})\mu} \\
        &=2\left(-\partial_{(u}\mathcal{F}_{r)\mu}+\Gamma^\rho_{ur}\mathcal{F}_{\rho\mu}+\Gamma^\rho_{\mu(u}\mathcal{F}_{r)\rho}+i[\mathcal{A}_{(u},\mathcal{F}_{r)\mu}]\right) \\
        &\quad +\partial_r\mathcal{F}_{r\mu}-\Gamma^\rho_{rr}\mathcal{F}_{\rho\mu}-\Gamma^\rho_{r\mu}\mathcal{F}_{r\rho}-i[\mathcal{A}_r,\mathcal{F}_{r\mu}] \nonumber \\
        &\quad +2r^{-2}\gamma^{-1}\left(\partial_{(z}\mathcal{F}_{\Bar{z})\mu}-\Gamma^\rho_{z\Bar{z}}\mathcal{F}_{\rho\mu}-\Gamma^\rho_{\mu(z}\mathcal{F}_{\Bar{z})\rho}-i[\mathcal{A}_{(z},\mathcal{F}_{\Bar{z})\mu}]\right) \nonumber\\
        &=-2\left(\partial_{(u}\mathcal{F}_{r)\mu}-i[\mathcal{A}_{(u},\mathcal{F}_{r)\mu}]\right)+r^{-1}(\mathcal{F}_{uz}\delta_{z\mu}+\mathcal{F}_{u\bar{z}}\delta_{\bar{z}\mu}) \\
        &\quad +\partial_r\mathcal{F}_{r\mu}-r^{-1}(\mathcal{F}_{rz}\delta_{z\mu}+\mathcal{F}_{r\bar{z}}\delta_{\bar{z}\mu})-i[\mathcal{A}_r,\mathcal{F}_{r\mu}] \nonumber \\
        &\quad +2r^{-2}\gamma^{-1}\{\partial_{(z}\mathcal{F}_{\Bar{z})\mu}-r\gamma(\mathcal{F}_{u\mu}-\mathcal{F}_{r\mu})-r^{-1}\mathcal{F}_{(z\bar{z})}\delta_{r\mu}-i[\mathcal{A}_{(z},\mathcal{F}_{\Bar{z})\mu}]\} \nonumber\\
        &\quad +r^{-1}(\mathcal{F}_{uz}-\mathcal{F}_{rz})\delta_{z\mu}+r^{-1}(\mathcal{F}_{u\bar{z}}-\mathcal{F}_{r\bar{z}})\delta_{\bar{z}\mu}+(r\gamma)^{-2}(\partial_z\gamma\delta_{z\mu}-\partial_{\bar{z}}\gamma\delta_{\bar{z}\mu})\mathcal{F}_{z\bar{z}} \nonumber \\
        &=-2\left(\partial_{(u}\mathcal{F}_{r)\mu}-i[\mathcal{A}_{(u},\mathcal{F}_{r)\mu}]\right)
        +\partial_r\mathcal{F}_{r\mu}-i[\mathcal{A}_r,\mathcal{F}_{r\mu}]-2r^{-1}\mathcal{F}_{ur}(\delta_{u\mu}+\delta_{r\mu}) \\
        &\quad +2r^{-2}\gamma^{-1}\left(\partial_{(z}\mathcal{F}_{\Bar{z})\mu}-i[\mathcal{A}_{(z},\mathcal{F}_{\Bar{z})\mu}]\right)+(r\gamma)^{-2}(\partial_z\gamma\delta_{z\mu}-\partial_{\bar{z}}\gamma\delta_{\bar{z}\mu})\mathcal{F}_{z\bar{z}} \nonumber \\
    \end{align}
    Thus, \(\forall n\in\mathbb{N}\):
    \begin{align}
        E_\mu^{(-n)}&=r^{n}\mathcal{P}_{-n}(\mathcal{E}_\mu) \\
        &=r^n\left\{-\partial_{u}\mathcal{P}_{-n}(\mathcal{F}_{r\mu})-\partial_{r}\mathcal{P}_{1-n}(\mathcal{F}_{u\mu}-\mathcal{F}_{r\mu})+i\mathcal{P}_{-n}(2[\mathcal{A}_{(u},\mathcal{F}_{r)\mu}]-[\mathcal{A}_r,\mathcal{F}_{r\mu}])\right\} \\
        &\quad -2r^{n-1}\mathcal{P}_{1-n}(\mathcal{F}_{ur})(\delta_{u\mu}+\delta_{r\mu})+2r^{n-2}\gamma^{-1}\left\{\partial_{(z}\mathcal{P}_{2-n}(\mathcal{F}_{\Bar{z})\mu})-i\mathcal{P}_{2-n}([\mathcal{A}_{(z},\mathcal{F}_{\Bar{z})\mu}])\right\} \nonumber \\
        &\quad +r^{n-2}\gamma^{-2}(\partial_z\gamma\delta_{z\mu}-\partial_{\bar{z}}\gamma\delta_{\bar{z}\mu})\mathcal{P}_{2-n}(\mathcal{F}_{z\bar{z}}) \nonumber \\
        &=-\partial_{u}F^{(-n)}_{r\mu}+(n-1)(F^{(1-n)}_{u\mu}-F^{(1-n)}_{r\mu})-2F^{(1-n)}_{ur}(\delta_{u\mu}+\delta_{r\mu}) \\
        &\quad +2\gamma^{-1}\partial_{(z}F^{(2-n)}_{\Bar{z})\mu}+\gamma^{-2}(\partial_z\gamma\delta_{z\mu}-\partial_{\bar{z}}\gamma\delta_{\bar{z}\mu})F^{(2-n)}_{z\bar{z}} \nonumber \\
        &\quad +i\sum_{k=0}^n\left(2[A^{(k-n)}_{(u},F^{(-k)}_{r)\mu}]-[A^{(k-n)}_r,F^{(-k)}_{r\mu}]\right)-2i\gamma^{-1}\sum_{k\in\mathbb{N}}[A^{(2+k-n)}_{(z},F^{(-k)}_{\Bar{z})\mu}]\\
        &=-\partial_{u}F^{(-n)}_{r\mu}+\tilde{\delta}_\mu^\rho(n)F^{(1-n)}_{u\rho}-\tilde{\delta}_\mu^\rho(n)F^{(1-n)}_{r\rho} \\
        &\quad +2\gamma^{-1}\partial_{(z}F^{(2-n)}_{\Bar{z})\mu}+\gamma^{-2}(\partial_z\gamma\delta_{z\mu}-\partial_{\bar{z}}\gamma\delta_{\bar{z}\mu})F^{(2-n)}_{z\bar{z}} \nonumber \\
        &\quad +i\sum_{k=0}^n\left(2[A^{(k-n)}_{(u},F^{(-k)}_{r)\mu}]-[A^{(k-n)}_r,F^{(-k)}_{r\mu}]-2\gamma^{-1}[A^{(k-n)}_{(z},F^{(2-k)}_{\Bar{z})\mu}]\right) \nonumber
    \end{align}
    with \(\tilde{\delta}_\mu^\nu(n)\) as in \eqref{def:delta_tilde}.
\end{proof}


    


\subsection{Bianchi identities}

For the recursive formulas it is useful to have explicitly written Bianchi identities for certain gauge and fall-offs conditions (for example, \eqref{light_cone_gauge} and \eqref{light_cone_falloff}. We can think of Bianchi identities as the vanishing of a $(0,3)-$tensor,
\begin{align}
    \mathcal{B}_{\mu\nu\rho}:=D_\mu\mathcal{F}_{\nu\rho}+\circlearrowleft\;=0\;,
\end{align}
where $\circlearrowleft$ indicates a cyclic permutation of the indices. Then, by defining the $r-$expansion of $\mathcal{B}_{\mu\nu\rho}$ as 
\be 
\mathcal{B}_{\mu\nu\rho} = \sum_{n \in \N} r^{-n} B^{(-n)}_{\mu\nu\rho},
\ee 
we can directly compute Bianchi identities in light-cone gauge at a given order \(n\in\mathbb{N}\):
\begin{align}
    \label{eq:Bianchi light-cone full bondi 1}
    B_{urz}^{(-n)}&=\partial_uF_{rz}^{(-n)}+(n-1)F_{uz}^{(1-n)}+\partial_zF_{ur}^{(-n)} \\
    &\quad +i\sum_{k=0}^{n-2}[A_r^{(k-n)},F_{uz}^{(-k)}]-i\sum_{k=2}^{n}[A_z^{(k-n)},F_{ur}^{(-k)}], \qquad n\geq2 \nonumber\\
    \label{eq:Bianchi light-cone full bondi 2}
    B_{ur\bar{z}}^{(-n)}&=\partial_uF_{r\bar{z}}^{(-n)}+(n-1)F_{u\bar{z}}^{(1-n)}+\partial_{\bar{z}}F_{ur}^{(-n)} \\
    &\quad +i\sum_{k=0}^{n-2}[A_r^{(k-n)},F_{u\bar{z}}^{(-k)}]-i\sum_{k=2}^{n}[A_{\bar{z}}^{(k-n)},F_{ur}^{(-k)}], \qquad n\geq2 \nonumber\\
    \label{eq:Bianchi light-cone full bondi 3}
    B_{uz\bar{z}}^{(-n)}&=\partial_uF_{z\bar{z}}^{(-n)}-\partial_zF_{u\bar{z}}^{(-n)}+\partial_{\bar{z}}F_{uz}^{(-n)}-2i\sum_{k=0}^n[A_{[z}^{(k-n)},F_{\bar{z}]u}^{(-k)}] \\
    \label{eq:Bianchi light-cone full bondi 4}
    B_{rz\bar{z}}^{(-n)}&=(1-n)F_{z\bar{z}}^{(1-n)}-\partial_zF_{r\bar{z}}^{(-n)}+\partial_{\bar{z}}F_{rz}^{(-n)} \\
    &\quad -i\sum_{k=0}^{n-2}[A_r^{(k-n)},F_{z\bar{z}}^{(-k)}]-2i\sum_{k=2}^n[A^{(k-n)}_{[z},F_{\bar{z}]r}^{(-k)}], \qquad n\geq2 \nonumber
\end{align}

Bianchi identities in radial gauge at a given order \(n\in\mathbb{N}\):
\begin{align}
    \label{eq:Bianchi radial gauge 1}
    B_{urz}^{(-n)}&=\partial_uF_{rz}^{(-n)}+(n-1)F_{uz}^{(1-n)}+\partial_zF_{ur}^{(-n)} \\
    \quad &-i\sum_{k=0}^{n-1}[A_u^{(k-n)},F_{rz}^{(-k)}]-i\sum_{k=2}^n[A_z^{(k-n)},F_{ur}^{(-k)}], \qquad n\geq2 \nonumber \\
    \label{eq:Bianchi radial gauge 2}
    B_{ur\bar{z}}^{(-n)}&=\partial_uF_{r\bar{z}}^{(-n)}+(n-1)F_{u\bar{z}}^{(1-n)}+\partial_{\bar{z}}F_{ur}^{(-n)} \\
    \quad &-i\sum_{k=0}^{n-1}[A_u^{(k-n)},F_{r\bar{z}}^{(-k)}]-i\sum_{k=2}^n[A_{\bar{z}}^{(k-n)},F_{ur}^{(-k)}], \qquad n\geq2 \nonumber \\
    \label{eq:Bianchi radial gauge 3}
    B_{uz\bar{z}}^{(-n)}&=\partial_uF_{z\bar{z}}^{(-n)}-\partial_zF_{u\bar{z}}^{(-n)}+\partial_{\bar{z}}F_{uz}^{(-n)} \\
    \quad &-i\sum_{k=0}^{n-1}[A_u^{(k-n)},F_{z\bar{z}}^{(-k)}]-2i\sum_{k=0}^n[A_{[z}^{(k-n)},F_{\bar{z}]u}^{(-k)}] \nonumber \\
    \label{eq:Bianchi radial gauge 4}
    B_{rz\bar{z}}^{(-n)}&=(1-n)F_{z\bar{z}}^{(1-n)}-\partial_zF_{r\bar{z}}^{(-n)}+\partial_{\bar{z}}F_{rz}^{(-n)}-2i\sum_{k=2}^n[A_{[z}^{(k-n)},F_{\bar{z}]r}^{(-k)}], \qquad n\geq2
\end{align}

\section{YM e.o.m. and recursive relations in flat Bondi coordinates}
\label{YM e.o.m. in flat Bondi coordinates}
Metric:
\begin{equation}
    ds^2=-2dudr+2r^2dzd\bar{z}
\end{equation}
Equations of motion:
\begin{equation}
    \mathcal{E}_\mu:=D^\nu\mathcal{F}_{\nu\mu}=\nabla^\nu\mathcal{F}_{\nu\mu}-i[\mathcal{A}^\nu,\mathcal{F}_{\nu\mu}]\in C^\infty_r(M)
\end{equation}
The YM equations of motion in flat Bondi at order \(1/r^n\), with \(n\in\mathbb{N}\), are:
\begin{equation}
    \label{eq:YM EOM at every order in flat Bondi}
    \begin{aligned}
    E_\mu^{(-n)}&=-\partial_{u}F^{(-n)}_{r\mu}+\tilde{\delta}_\mu^\rho(n-1)F^{(1-n)}_{u\rho}+2\partial_{(z}F_{\Bar{z})\mu}^{(2-n)} \\
    &\quad +2i\sum_{k=0}^n\left([A^{(k-n)}_{(u},F^{(-k)}_{r)\mu}]-[A^{(k-n)}_{(z},F^{(2-k)}_{\Bar{z})\mu}]\right),
    \end{aligned}
\end{equation}
where \(\tilde{\delta}_\mu^\nu(n)\) is defined in \eqref{def:delta_tilde}.
\begin{proof}
    From \eqref{def:EOM}:
    \begin{align}
        \mathcal{E}_\mu&=g^{\nu\rho}D_\rho\mathcal{F}_{\nu\mu} \\
        &=-2D_{(u}\mathcal{F}_{r)\mu}+2r^{-2}D_{(z}\mathcal{F}_{\Bar{z})\mu} \\
        &=2\left(-\partial_{(u}\mathcal{F}_{r)\mu}+\Gamma^\rho_{ur}\mathcal{F}_{\rho\mu}+\Gamma^\rho_{\mu(u}\mathcal{F}_{r)\rho}+i[\mathcal{A}_{(u},\mathcal{F}_{r)\mu}]\right) \\
        &\quad +2r^{-2}\left(\partial_{(z}\mathcal{F}_{\Bar{z})\mu}-\Gamma^\rho_{z\Bar{z}}\mathcal{F}_{\rho\mu}-\Gamma^\rho_{\mu(z}\mathcal{F}_{\Bar{z})\rho}-i[\mathcal{A}_{(z},\mathcal{F}_{\Bar{z})\mu}]\right) \nonumber\\
        &=2\left\{-\partial_{(u}\mathcal{F}_{r)\mu}+r^{-1}/2(\mathcal{F}_{uz}\delta_{z\mu}+\mathcal{F}_{u\bar{z}}\delta_{\bar{z}\mu})+i[\mathcal{A}_{(u},\mathcal{F}_{r)\mu}]\right\} \\
        &\quad +2r^{-2}\left\{\partial_{(z}\mathcal{F}_{\Bar{z})\mu}-r\mathcal{F}_{u\mu}-r^{-1}\mathcal{F}_{(z\Bar{z})}\delta_{r\mu}-r/2(\mathcal{F}_{zu}\delta_{z\mu}+\mathcal{F}_{\bar{z}u}\delta_{\bar{z}\mu})-i[\mathcal{A}_{(z},\mathcal{F}_{\Bar{z})\mu}]\right\} \nonumber\\
        &=-2\left(\partial_{(u}\mathcal{F}_{r)\mu}-i[\mathcal{A}_{(u},\mathcal{F}_{r)\mu}]\right)-2r^{-1}\mathcal{F}_{ur}\delta_{r\mu}+2r^{-2}\left(\partial_{(z}\mathcal{F}_{\Bar{z})\mu}-i[\mathcal{A}_{(z},\mathcal{F}_{\Bar{z})\mu}]\right)
    \end{align}
    Thus, \(\forall n\in\mathbb{N}\):
    \begin{align}
        E_\mu^{(-n)}&=r^{n}\mathcal{P}_{-n}(\mathcal{E}_\mu) \\
        &=-r^{n}\left\{\partial_u\mathcal{P}_{-n}(\mathcal{F}_{r\mu})+\partial_r\mathcal{P}_{1-n}(\mathcal{F}_{u\mu})-2i\mathcal{P}_{-n}([\mathcal{A}_{(u},\mathcal{F}_{r)\mu}])\right\} \\
        &\quad -2r^{n-1}\mathcal{P}_{1-n}(\mathcal{F}_{ur})\delta_{r\mu}+2r^{n-2}\left\{\partial_{(z}\mathcal{P}_{2-n}(\mathcal{F}_{\Bar{z})\mu})-i\mathcal{P}_{2-n}([\mathcal{A}_{(z},\mathcal{F}_{\Bar{z})\mu}])\right\} \nonumber\\
        &=-\partial_u F_{r\mu}^{(-n)}+(n-1)F_{u\mu}^{(1-n)}-2F_{ur}^{(1-n)}\delta_{r\mu}+2\partial_{(z}F_{\Bar{z})\mu}^{(2-n)} \\
        &\quad +2i\sum_{k=0}^n[A_{(u}^{(k-n)},F_{r)\mu}^{(-k)}]-2i\sum_{k\in\mathbb{N}}[A_{(z}^{(2+k-n)},F_{\Bar{z})\mu}^{(-k)}] \nonumber\\
        &=-\partial_u F_{r\mu}^{(-n)}+(n-3)F_{ur}^{(1-n)}\delta_{r\mu}+(n-1)F_{uz}^{(1-n)}\delta_{z\mu}+(n-1)F_{u\bar{z}}^{(1-n)}\delta_{\bar{z}\mu} \\
        &\quad +2\partial_{(z}F_{\Bar{z})\mu}^{(2-n)}+2i\sum_{k=0}^n\left([A_{(u}^{(k-n)},F_{r)\mu}^{(-k)}]-[A_{(z}^{(k-n)},F_{\Bar{z})\mu}^{(2-k)}]\right) \nonumber
    \end{align}
\end{proof}
The components are:
\begin{align}
    E_u^{(-n)}&=\partial_u F_{ur}^{(-n)}+2\partial_{(z}F_{\Bar{z})u}^{(2-n)}-i\sum_{k=0}^n\left([A_{u}^{(k-n)},F_{ur}^{(-k)}]+2[A_{(z}^{(k-n)},F_{\Bar{z})u}^{(2-k)}]\right) \\
    E_r^{(-n)}&=(n-3)F_{ur}^{(1-n)}+2\partial_{(z}F_{\Bar{z})r}^{(2-n)}+i\sum_{k=0}^n\left([A_{r}^{(k-n)},F_{ur}^{(-k)}]-2[A_{(z}^{(k-n)},F_{\Bar{z})r}^{(2-k)}]\right) \\
    E_z^{(-n)}&=-\partial_u F_{rz}^{(-n)}+(n-1)F_{uz}^{(1-n)}-\partial_{z}F_{z\Bar{z}}^{(2-n)} \\
    &+i\sum_{k=0}^n\left(2[A_{(u}^{(k-n)},F_{r)z}^{(-k)}]+[A_{z}^{(k-n)},F_{z\Bar{z}}^{(2-k)}]\right) \\
    E_{\bar{z}}^{(-n)}&=-\partial_u F_{r\bar{z}}^{(-n)}+(n-1)F_{u\bar{z}}^{(1-n)}+\partial_{\bar{z}}F_{z\Bar{z}}^{(2-n)} \\
    &+i\sum_{k=0}^n\left(2[A_{(u}^{(k-n)},F_{r)\bar{z}}^{(-k)}]-[A_{\bar{z}}^{(k-n)},F_{z\Bar{z}}^{(2-k)}]\right)
\end{align}
The first three orders are:
\begin{align}
    E_\mu^{(0)}&=-\partial_u F_{r\mu}^{(0)}+2i[A_{(u}^{(0)},F_{r)\mu}^{(0)}] \\
    E_\mu^{(-1)}&=-\partial_u F_{r\mu}^{(-1)}-2F_{ur}^{(0)}\delta_{r\mu}+2i\left([A_{(u}^{(-1)},F_{r)\mu}^{(0)}]+[A_{(u}^{(0)},F_{r)\mu}^{(-1)}]\right) \\
    E_\mu^{(-2)}&=-\partial_u F_{r\mu}^{(-2)}+\tilde{\delta}_\mu^\rho(1) F_{u\rho}^{(-1)}+2\partial_{(z}F_{\Bar{z})\mu}^{(0)}-2i[A_{(z}^{(0)},F_{\Bar{z})\mu}^{(0)}]+2i\sum_{k=0}^2[A_{(u}^{(k-n)},F_{r)\mu}^{(-k)}]
\end{align}

\subsection{Light-cone gauge}
E.o.m. in light-cone gauge \eqref{light_cone_gauge} with fall-off \eqref{light_cone_falloff}:
\begin{align}
    \label{eq:EOM flat Bondi light-cone u component}
    E_u^{(-n)}&=\partial_u F_{ur}^{(-n)}+2\partial_{(z}F_{\Bar{z})u}^{(2-n)}-2i\sum_{k=0}^{n-2}[A_{(z}^{(2+k-n)},F_{\Bar{z})u}^{(-k)}], \qquad n\geq2 \\
    \label{eq:EOM flat Bondi light-cone r component}
    E_r^{(-n)}&=(n-3)F_{ur}^{(1-n)}+2\partial_{(z}F_{\Bar{z})r}^{(2-n)} \\
    &\quad +i\sum_{k=2}^{n-2}\left([A_{r}^{(k-n)},F_{ur}^{(-k)}]-2[A_{(z}^{(2+k-n)},F_{\Bar{z})r}^{(-k)}]\right), \qquad n\geq4  \nonumber \\
    \label{eq:EOM flat Bondi light-cone z component}
    E_z^{(-n)}&=-\partial_u F_{rz}^{(-n)}+(n-1)F_{uz}^{(1-n)}-\partial_{z}F_{z\Bar{z}}^{(2-n)} \\
    &\quad +i\sum_{k=0}^{n-2}\left([A_{r}^{(k-n)},F_{uz}^{(-k)}]+[A_{z}^{(2+k-n)},F_{z\Bar{z}}^{(-k)}]\right), \qquad n\geq2 \nonumber \\
    \label{eq:EOM flat Bondi light-cone zb component}
    E_{\bar{z}}^{(-n)}&=-\partial_u F_{r\bar{z}}^{(-n)}+(n-1)F_{u\bar{z}}^{(1-n)}+\partial_{\bar{z}}F_{z\Bar{z}}^{(2-n)} \\
    &\quad +i\sum_{k=0}^{n-2}\left([A_{r}^{(k-n)},F_{u\bar{z}}^{(-k)}]-[A_{\bar{z}}^{(2+k-n)},F_{z\Bar{z}}^{(-k)}]\right), \qquad n\geq2 .\nonumber
\end{align}

\textbf{Recursive formula for the gauge vector in Light-cone gauge:} Starting from the equations above, it is possible to find recursive formulas for $A_\mu^{(-n)}$ by following the very same steps described in paragraph \ref{subsubsec:recursive_formula_A_light_cone_full}. Since the comments and considerations presented there still hold, we simply list the results:
\begin{align}
    A_r^{(-2)}&=2\partial_u^{-1}\left(\partial_{(z}A_{\bar{z})}^{(0)}+i\partial_u^{-1}[\partial_uA_{(z}^{(0)},A_{\bar z)}^{(0)}]\right) \\
    A_r^{(-n)}&=\frac{1}{2-n}\partial_u^{-1}\left[2\partial_{(z}F_{\bar{z})r}^{(1-n)}+i\sum_{k=2}^{n-1}\left([A_r^{(k-n-1)},F_{ur}^{(-k)}]-2[A_{(z}^{(1+k-n)},F_{\bar{z})r}^{(-k)}]\right)\right], \quad n\geq3 \\
    A_z^{(-n)}&=\frac{1}{2n}\bigg\{-\partial_zA_r^{(-n-1)}+i\sum_{k=2}^{n+1}[A_z^{(k-n-1)},A_r^{(-k)}] \\
        &\quad +\partial_u^{-1}\bigg[\partial_zF^{(1-n)}_{z\bar{z}}-i\sum_{k=0}^{n-1}\left([A^{(k-n-1)}_{r},\partial_uA_z^{(-k)}]+[A^{(1+k-n)}_{z},F^{(-k)}_{z\Bar{z}}]\right)\bigg]\bigg\}, \quad n\geq1 \nonumber
\end{align}
that, in a more compact form, read
\begin{align}
A_r^{(-2)}&=-2\partial_u^{-2}\left(D_{(z}\mathcal{F}_{\bar z)u}\right)^{(0)} \\
    A_r^{(-n)}&=\frac{1}{2-n}\partial_u^{-1}\left[2\left(D_{(z}\mathcal{F}_{\bar{z})r}\right)^{(1-n)}+i\left([\mathcal{A}_r,\partial_u\mathcal{A}_r]\right)^{(-n-1)}\right], \qquad n\geq3 \\
    A_z^{(-n)}&=-\frac{1}{2n}\big\{(D_z\mathcal{A}_r)^{(-n-1)}-\partial_u^{-1}\big((D_z\mathcal{F}_{z\bar{z}})^{(1-n)}-i([\mathcal{A}_r,\partial_u\mathcal{A}_z])^{(-n-1)}\big)\big\}, \qquad n\geq1.
\end{align}

\subsection{Radial gauge}
E.o.m. in radial gauge \eqref{radial_gauge} with fall-off \eqref{radial_fall_off}:
\begin{align}
    \label{eq:EOM flat Bondi radial u component}
    E_u^{(-n)}&=\partial_u F_{ur}^{(-n)}+2\partial_{(z}F_{\Bar{z})u}^{(2-n)} \\
    &\quad -i\sum_{k=0}^{n-1}[A_{u}^{(k-n)},F_{ur}^{(-k)}]-2i\sum_{k=0}^{n-2}[A_{(z}^{(2+k-n)},F_{\Bar{z})u}^{(-k)}], \qquad n\geq2 \nonumber\\
    \label{eq:EOM flat Bondi radial r component}
    E_r^{(-n)}&=(n-3)F_{ur}^{(1-n)}+2\partial_{(z}F_{\Bar{z})r}^{(2-n)}-2i\sum_{k=2}^{n-2}[A_{(z}^{(2+k-n)},F_{\Bar{z})r}^{(-k)}], \qquad n\geq4 \\
    \label{eq:EOM flat Bondi radial z component}
    E_z^{(-n)}&=-\partial_u F_{rz}^{(-n)}+(n-1)F_{uz}^{(1-n)}-\partial_{z}F_{z\Bar{z}}^{(2-n)} \\
    &\quad +i\sum_{k=0}^{n-1}[A_{u}^{(k-n)},F_{rz}^{(-k)}]+i\sum_{k=0}^{n-2}[A_{z}^{(2+k-n)},F_{z\Bar{z}}^{(-k)}], \qquad n\geq2 \nonumber\\
    \label{eq:EOM flat Bondi radial zb component}
    E_{\bar{z}}^{(-n)}&=-\partial_u F_{r\bar{z}}^{(-n)}+(n-1)F_{u\bar{z}}^{(1-n)}+\partial_{\bar{z}}F_{z\Bar{z}}^{(2-n)} \\
    &\quad +i\sum_{k=0}^{n-1}[A_{u}^{(k-n)},F_{r\bar{z}}^{(-k)}]-i\sum_{k=0}^{n-2}[A_{\bar{z}}^{(2+k-n)},F_{z\Bar{z}}^{(-k)}], \qquad n\geq2 .\nonumber
\end{align}

\textbf{Recursive formula for the gauge vector in radial gauge:} Following the same line of reasoning of paragraph \ref{subsubsec:recursion_radial_full_bondi} we obtain
\begin{align}
    A_u^{(-1)}&=2\partial_{(z}A_{\bar{z})}^{(0)}+2i\partial_u^{-1}[\partial_uA_{(z}^{(0)},A_{\bar z)}^{(0)}] \nonumber\\
    A_u^{(-n)}&=-\frac{2}{n}\left(\partial_{(z}A_{\Bar{z})}^{(1-n)}+\sum_{k=1}^{n-1}\frac{ik}{1-n}[A^{(1+k-n)}_{(z},A_{\Bar{z})}^{(-k)}]\right), \quad n\geq2 \nonumber \\
    A_z^{(-n)}&=\frac{1}{2}\partial_u^{-1}\bigg(\partial_{z}A_{u}^{(-n)}+\frac{1}{n}\partial_zF^{(1-n)}_{z\bar{z}}-\frac{i}{n}\sum_{k=1}^{n}[A_z^{(k-n)},(2n-k)A_u^{(-k)}+F^{(1-k)}_{z\Bar{z}}]\bigg), \quad n\geq1. \nonumber
\end{align}
Some recursive formulas for the field strength are
\begin{align}
   \label{Fur_recs flat}   F_{ur}^{(-2)}&=2\left(\partial_{(z}A_{\bar{z})}^{(0)}+i\partial_u^{-1}[\partial_uA_{(z}^{(0)},A_{\bar z)}^{(0)}]\right) \\
    F_{ur}^{(-n)}&=-2\left(\partial_{(z}A_{\Bar{z})}^{(2-n)}+\sum_{k=1}^{n-2}\frac{ik}{2-n}[A^{(2+k-n)}_{(z},A_{\Bar{z})}^{(-k)}]\right), \qquad n\ge3 \nonumber \\
    \label{F_rz recursion flat}
    F_{rz}^{(-n)}&=-\frac{1}{2}\partial_u^{-1}\left(\partial_z(F_{ur}^{(-n)}+F_{z\bar z}^{(2-n)})-i\sum_{k=2}^n[A_z^{(k-n)},\frac{2n-k-1}{k-1}F_{ur}^{(-k)}+F_{z\bar z}^{(2-k)}]\right), \quad n\geq2. 
\end{align}
The first two equations above can be recast in the more compact form
\begin{align}
F_{ur}^{(-2)}&=-2\partial_u^{-1}\left(D_{(z}\mathcal{F}_{\bar z)u}\right)^{(0)} \\
    F_{ur}^{(-n)}&=\frac{2}{2-n}\left(D_{(z}\mathcal{F}_{\bar z)r}\right)^{(1-n)}, \qquad n\geq3.
\end{align}
Finally, it is easy to prove that
\begin{equation}
    F_{uz}^{(-n)}=\frac{1}{n}\left[(D_u\mathcal{F}_{rz})^{(-n-1)}+(D_z\mathcal{F}_{z\bar z})^{(1-n)}\right].
\end{equation}

\section{Sub$^n$-leading charge $\I^+$ limit} \label{sub_n_leading_charge_limit}

In this appendix, we will explicitly show how the charges defined on a spatial Cauchy slice can be ``pushed'' to $\Ib^+$ as depicted in \autoref{scri_plus_diag}, via the limiting procedure $t \rightarrow +\infty$ at fixed $u$ and a renormalization on $t$.

Start with the charge defined in the standard way,
\be 
\begin{aligned}
\tilde{\rho}_\lambda
&=-\partial_\mu\tr\left(\sqrt{g}\lambda\mathcal{F}^{t\mu}\right)\\
&=-\partial_\mu\tr\left(\sqrt{g}\lambda\left(\mathcal{F}^{r\mu}+\mathcal{F}^{u\mu} \right) \right)\\
&=\tr\left\{\gamma \left[\partial_r(r^2\lambda\F_{ur}) 
-\partial_u(r^2\lambda\F_{ur})\right]
+r^2 \partial_z(\gamma\lambda\F_u^{\ z})
+r^2 \partial_{\zb}(\gamma\lambda\F_u^{\ \zb})\right\},
\end{aligned}
\ee
where $\lambda$ is the large gauge parameter. We will derive the relation between $\lambda$ and our $\Lambda$ (see e.g. \eqref{simple_lambda_plus}) below.

The last two terms will drop out when integrating over the sphere. Let us then normalise by a factor of $\gamma$ and define our density as 
\be
\label{eq:normalized_charge}
\rho=\tr\left[\partial_r(r^2\lambda\F_{ur}) 
-\partial_u(r^2\lambda\F_{ur})\right]
\ee
Assume a tree-level type fall-off for $\lambda$ and $\F_{ur}$:
\be
\label{eq:lambda_F_falloff}
\lambda(r,u,z,\zb) =\sum_{k=0}^{n}r^k \lambda^{(k)} (u,z,\zb)\ ,\qquad 
\F_{ur}=\sum_{m=0}^{\infty}\frac{F_{ur}^{(-2-m)}(u,z,\zb) }{r^{2+m}}
\ee 
then, keeping only the positive powers of $r$ in $\rho$, we have
\be
\label{eq:rho}
\begin{aligned}
\rho&=\tr\left[\partial_r\left(\sum_{k\geq m\geq 0}^n\lambda^{(k)}F_{ur}^{(-2-m)}r^{k-m}\right)-\partial_u\left(\sum_{k\geq m\geq 0}^n\lambda^{(k)}F_{ur}^{(-2-m)}r^{k-m}\right)\right]\\
&=\tr\sum_{k\geq m\geq 0}^n (k-m)\lambda^{(k)}
F_{ur}^{(-2-m)}r^{k-m-1}
-\tr\sum_{k\geq m\geq 0}^n\partial_u\left(\lambda^{(k)}F_{ur}^{(-2-m)} \right)r^{k-m}
\end{aligned}
\ee
Now we recall that $r=t-u$. We plug this into the above, and take 
the limit $t\to\infty$, keeping $u$ fixed. 

We discard the terms divergent in $t$, as they correspond to sub$^k$-leading theorems, for $k< n$ (by virtue of Lemma \autoref{lemma:tdivergent terms}, see \autoref{subsec:t divergent terms}). These terms can be written as contributing to the symplectic potential as total variations of the Lagrangian or a partial divergence ($\partial^i f_i$ over spatial indexes). These types of terms were shown to be renormalizable, see e.g. \cite{Peraza:2023ivy, Freidel:2019ohg}. Thus, we are left with
\be
\label{eq:rho_finite}
\begin{aligned}
\rho_{finite}&=\tr\sum_{k\geq m\geq 0}^n
(k-m)\lambda^{(k)}F_{ur}^{(-2-m)}(-1)^{k-m-1}u^{k-m-1} \\
&\quad -\tr\sum_{k\geq m\geq 0}^n \partial_u\left(\lambda^{(k)}F_{ur}^{(-2-m)}\right)(-1)^{k-m}u^{k-m}\\
&=\partial_u\tr \left(\sum_{k\geq m\geq 0}^n \lambda^{(k)}F_{ur}^{(-2-m)}(-1)^{k-m-1}u^{k-m} \right)\\
&=\partial_u\tr \left(\sum_{m=0}^n\left(\sum_{k=m}^n
\lambda^{(k)}(-1)^{k-m-1}u^{k-m} 
\right) F_{ur}^{(-2-m)} \right)
\end{aligned}
\ee 
Then we take our charge density constructed in \autoref{charges} and write it as a charge density on $\I^+$, keeping only the 0th order in the Stueckelberg fields:
\be 
\overset{0}{q}_{\I^+}=\partial_u\tr\left(\sum_{m=0}^n \Lambda^{(m)} F_{ru}^{(-2-m)}\right)
\ee 
Then we can match the above via
\be 
\Lambda^{(m)}=\sum_{k=m}^n
\lambda^{(k)}(-1)^{k-m}u^{k-m}
\ee
At sub-leading order (i.e. $n=1$), we then have
\be 
\Lambda^{(0)}=\lambda^{(0)}-u\lambda^{(1)},\qquad
\Lambda^{(1)}=\lambda^{(1)}
\ee 
The above gives the explicit relation between our definition of the large gauge parameter, and that in \cite{Campiglia2016a}, working to first sub-leading order.

\subsection{$t$-divergent terms in the charge}
\label{subsec:t divergent terms}
\begin{lemma}
\label{lemma:tdivergent terms}
    Let $\rho$ be the charge \eqref{eq:normalized_charge} with gauge parameter $\lambda$ and field strength $\mathcal{F}_{ur}$ as in \eqref{eq:lambda_F_falloff}. Moreover, assume $\lambda^{(k)}\propto u^{n-k}$ for all $k\in\{0,\dots,n\}$. Then, for each integer $s\in\{0,\dots,n\}$, the coefficient of $t^s$ in the divergent part of $\lim_{t\rightarrow\infty}\rho$ is equal to the finite part of the charge in a sub$^{(n-s)}$-leading theorem.
\end{lemma}
\begin{proof}
    In order to prove this result, it is useful to introduce some notation: we add the superscript $[n]$ to the charge in \eqref{eq:normalized_charge} and to the gauge parameter in \eqref{eq:lambda_F_falloff}, to show explicitly that they refer to a sub$^n$-theorem:
    \begin{align}
        \rho^{[n]}=\partial_r\tr(r^2\lambda^{[n]}\F_{ur})-\partial_u\tr(r^2\lambda^{[n]}\F_{ur}) \ , \qquad
        \lambda^{[n]}=\sum_{k=0}^{n}r^k \lambda^{[n](k)} \ .
    \end{align}
    By assumption, the coefficients $\lambda^{[n](k)}$ in the expansion of $\lambda^{[n]}$ have the form
    \begin{equation}
        \lambda^{[n](k)}(u,z,\bar z)=f_{n,k}(z,\bar z)u^{n-k} \ ,
    \end{equation}
    where $f_{n,k}$ are arbitrary functions on the sphere. From equation \eqref{eq:rho} and $r=t-u$ it follows that
    \begin{equation}
        \begin{aligned}
            \rho^{[n]}&=-\partial_u\tr\sum_{m=0}^n\sum_{k=m}^n\lambda^{[n](k)}F_{ur}^{(-2-m)}(t-u)^{k-m} \\
            &=\partial_u\tr\sum_{m=0}^n\sum_{k=m}^nf_{n,k}(z,\bar z)F_{ur}^{(-2-m)}\sum_{s=0}^{k-m}\binom{k-m}{s}(-1)^{k-m-s-1}\,t^su^{n-m-s} \\
            &=\partial_u\tr\sum_{s=0}^n\sum_{m=0}^{n-s}\sum_{k=m+s}^n(-1)^{k-m-s-1}\binom{k-m}{s}f_{n,k}(z,\bar z)F_{ur}^{(-2-m)}\,t^su^{n-m-s} \; ,
        \end{aligned}
    \end{equation}
    where $\binom{k-m}{s}$ is the binomial coefficient\footnote{In general, the binomial coefficient $\binom{n}{k}$ is well defined only for $n\geq k$. Here, we extend its definition as follows:
    \begin{equation}
        \binom{n}{k}:=\begin{cases}
            \frac{n!}{k!(n-k)!} \quad n\geq k \\
            0 \quad\quad\quad\ \ \, n<k
        \end{cases} \; .
    \end{equation}
    }. In the last step, we make use of
    \begin{equation}
        \begin{aligned}
            &\sum_{m=0}^n\sum_{k=m}^n\sum_{s=0}^{k-m}\binom{k-m}{s}\alpha_{n,s,m,k}=\sum_{m,s=0}^n\sum_{k=m}^n\binom{k-m}{s}\alpha_{n,s,m,k} \\
            &=\sum_{m,s=0}^n\sum_{k=m+s}^n\binom{k-m}{s}\alpha_{n,s,m,k}=\sum_{s=0}^n\sum_{m=0}^{n-s}\sum_{k=m+s}^n\binom{k-m}{s}\alpha_{n,s,m,k} \; ,          
        \end{aligned}
    \end{equation}
    which is true for any arbitrary function $\alpha_{n,s,m,k}$. At this point, it is convenient to define the following objects:
    \begin{align}
        \label{def:rhon_ftilde}
        \rho^{[n]}_s(u,z,\bar z)&:=\partial_u\tr\sum_{m=0}^{n-s}\tilde f_{n,s,m}(z,\bar z)F_{ur}^{(-2-m)}u^{n-m-s} \\
        \label{def:tildef}
        \tilde f_{n,s,m}(z,\bar z)&:=\sum_{k=m+s}^n(-1)^{k-m-s-1}\binom{k-m}{s}f_{n,k}(z,\bar z) \quad \mathrm{for} \ m+s\leq n \;,
    \end{align}
    so that the charge $\rho^{[n]}$ can be written as
    \begin{equation}
    \label{pf:rho_n}
        \begin{aligned}
            \rho^{[n]}=\sum_{s=0}^nt^s\rho_s^{[n]}=\underbrace{\rho_0^{[n]}}_{\mathrm{finite}}+\underbrace{\sum_{s=1}^nt^s\rho_s^{[n]}}_{\mathrm{divergent}} \;.
        \end{aligned}
    \end{equation}
    When taking the limit $t\rightarrow\infty$, the expression above diverge, except for the term with $s=0$, that corresponds to \eqref{eq:rho_finite}. Looking at the expression \eqref{def:rhon_ftilde} for the coefficient of the $t^s$-divergent term, we can see that it has the same $u$-expansion as the finite term in the charge of a sub$^{(n-s)}$-leading theorem, which is
    \begin{equation}
        \label{def:rhonminuss_ftilde}
        \rho^{[n-s]}_0(u,z,\bar z)=\partial_u\tr\sum_{m=0}^{n-s}\tilde f_{n-s,0,m}(z,\bar z)F_{ur}^{(-2-m)}u^{n-m-s} \;.
    \end{equation}
    In particular, given $\rho^{[n]}$, it is always possible to find a choice of the gauge parameter $\lambda^{[n-s]}$ for the sub$^{(n-s)}$-leading theorem such that $\rho_s^{[n]}=\rho_0^{[n-s]}$. Indeed, \eqref{def:rhon_ftilde} and \eqref{def:rhonminuss_ftilde} coincide if $\tilde f_{n-s,0,m}=\tilde f_{n,s,m}$ for all $m=0,\dots,n-s$. Using the definition \eqref{def:tildef}, and after some algebra, one can show that this condition corresponds to
    \begin{equation}
        \sum_{k=m+s}^{n}(-1)^{k-m-s}\left[f_{n-s,k-s}-\binom{k-m}{s}f_{n,k}\right]=0\;, \quad \forall m\in\{0,\dots,n-s\}\;,
    \end{equation}
    which is a system of $n-s+1$ equations. Solving it recursively for for $f_{n-s,k-s}$ we obtain
    \begin{equation}
        \label{eq:solution_system_tildefnminuss}
        \begin{aligned}
            f_{n-s,n-s}&=f_{n,n} \\
            f_{n-s,n-s-j}&=f_{n,n-j}+\sum_{k=n-j+1}^n(-1)^{k+j-n-1}\left[f_{n-s,k-s}-\binom{k+s+j-n}{s}f_{n,k}\right]\;,
        \end{aligned}
    \end{equation}
    with $j\in\{1,\dots,n-s\}$. Thus, choosing a gauge parameter $\lambda^{[n-s]}$ whose expansion is given by functions $f_{n-s,k}$ as in \eqref{eq:solution_system_tildefnminuss} gives $\rho_s^{[n]}=\rho_0^{[n-s]}$.
    
\end{proof}

\section{Sub-leading charges from the Ward identity perspective} \label{last_appendix}

In this section we derive the explicit form of the gauge transformation $\Lambda$ under the hypothesis that its associated charge matches a Ward identity-type expression.

More specifically, we assume that the charge splits into two contributions, a soft part that is linear in both $A_z^{(0)}$ and $A_{\zb}^{(0)}$ and a hard part with quadratic and higher terms in $A_z^{(0)}$ and $A_{\zb}^{(0)}$. In particular, the modes should correspond to a vanishing energy limiting process of a soft gluon,
\be 
\lim_{\omega \rightarrow 0} \int du  e^{iu\omega} F_{uz}^{(0)} \rightarrow  \int du  F_{uz}^{(0)}.
\ee

We impose radial gauge on the gauge field, but recall that we do not restrict the over-leading gauge parameters, as they only act on the Stueckelberg fields (see \autoref{sec_ext_phase_space} for more details).

\subsection{Leading order charge}

As a warm up, we first re-derive the $u$ dependence of the large gauge transformation corresponding to the leading charge. Let us start with a general parameter, $\Lambda=\Lambda^{(0)}(u,z,\bar z)$.
\begin{equation}
    \begin{aligned}
        Q_\Sigma[\Lambda]&=\frac{1}{e^2}\int_{\partial\Sigma}\tr(\Lambda * F)^{(0)}=\frac{1}{e^2}\int_{\partial\Sigma}dz^2\tr(\sqrt{g}\Lambda F_{ru})^{(0)} \\
        &=\frac{1}{e^2}\int_{\partial\Sigma}dz^2\tr\left(\gamma\Lambda^{(0)}F_{ru}^{(-2)}\right)=\frac{1}{e^2}\int_{\mathcal{I}^+}dudz^2\partial_u\tr\left(\gamma\Lambda^{(0)}F_{ru}^{(-2)}\right) \\
        &=\frac{1}{e^2}\int_{\mathcal{I}^+}dudz^2\gamma\tr\left(\partial_u\Lambda^{(0)}F_{ru}^{(-2)}\right)+\frac{1}{e^2}\int_{\mathcal{I}^+}dudz^2\gamma\tr\left(\Lambda^{(0)}\partial_uF_{ru}^{(-2)}\right)
    \end{aligned}
\end{equation}

From the e.o.m. \eqref{eq:EOM full Bondi radial u component} at order $n=2$ we obtain
\begin{equation}
\label{pf:F_ru-2}
    F_{ru}^{(-2)}=2\gamma^{-1}\int_{-\infty}^ud\tilde{u}\left(\partial_{(z}F_{\bar z)u}^{(0)}-i[A_{(z}^{(0)},F_{\bar z)u}^{(0)}]\right)\;,
\end{equation}
so that
\begin{equation}
    \begin{aligned}
        Q_\Sigma[\Lambda]&=\frac{2}{e^2}\int_{\mathcal{I}^+}dudz^2\tr\left[\partial_u\Lambda^{(0)}\int_{-\infty}^ud\tilde{u}\left(\partial_{(z}F_{\bar z)u}^{(0)}-i[A_{(z}^{(0)},F_{\bar z)u}^{(0)}]\right)\right] \\
        &\quad+\frac{2}{e^2}\int_{\mathcal{I}^+}dudz^2\tr\left[\Lambda^{(0)}\left(\partial_{(z}F_{\bar z)u}^{(0)}-i[A_{(z}^{(0)},F_{\bar z)u}^{(0)}]\right)\right]\;.
    \end{aligned}
\end{equation}
The top line in the previous equation does not have the right form corresponding to the insertion of a standard soft mode. Therefore, we need to set $\partial_u\Lambda^{(0)}=0$, so
\begin{equation}
    \begin{aligned}
        Q_\Sigma[\Lambda]&=-\frac{2}{e^2}\int_{\mathcal{I}^+}dudz^2\tr\left[\partial_{(z}\Lambda^{(0)}F_{\bar z)u}^{(0)}+i\Lambda^{(0)}[A_{(z}^{(0)},F_{\bar z)u}^{(0)}]\right]\;.
    \end{aligned}
\end{equation}

\subsection{sub-leading charge}
Next, we write $\Lambda=r\Lambda^{(1)}(u,z,\bar z)+\Lambda^{(0)}(u,z,\bar z)$, and we want to obtain the expression for $\Lambda^{(0)}$ in terms of $\Lambda^{(1)}$.

The sub-leading charge is
\begin{equation}
    \begin{aligned}
        Q_\Sigma^\mathrm{sub}[\Lambda]&=\frac{1}{e^2}\int_{\partial\Sigma}dz^2\;\tr\left(r^2\gamma\Lambda F_{ru}\right)^{(0)}=\frac{1}{e^2}\int_{\partial\Sigma}dz^2\gamma\tr\left(\Lambda^{(0)}F_{ru}^{(-2)}+\Lambda^{(1)}F_{ru}^{(-3)}\right) \\
        &=\frac{1}{e^2}\int_{\mathcal{I}^+}dudz^2\gamma\tr\left(\partial_u\Lambda^{(0)}F_{ru}^{(-2)}+\Lambda^{(0)}\partial_uF_{ru}^{(-2)}+\partial_u\Lambda^{(1)}F_{ru}^{(-3)}+\Lambda^{(1)}\partial_uF_{ru}^{(-3)}\right)\;.
    \end{aligned}
\end{equation}
The first two terms in the last line have been already studied for the leading order charge. In order to rewrite the other two terms, consider e.o.m. \eqref{eq:EOM full Bondi radial u component} with $n=3$:
\begin{equation}
\label{pf:partial_u_F_ru-3}
    \begin{aligned}
        \partial_{u}F^{(-3)}_{ru}&=2\gamma^{-1}\partial_{(z}F^{(-1)}_{\Bar{z})u}-i[A^{(-1)}_{u},F^{(-2)}_{ur}]-2i\gamma^{-1}\sum_{k=0}^{1}[A^{(k-1)}_{(z},F^{(-k)}_{\Bar{z})u}]\;.
    \end{aligned}
\end{equation}
Now, subtract Bianchi identity \eqref{eq:Bianchi radial gauge 1} from e.o.m. \eqref{eq:EOM full Bondi radial z component} with $n=2$, to find
\begin{equation}
\label{pf:F_uz_-1}
    \begin{aligned}
        F^{(-1)}_{uz}&=\frac{1}{2}\partial_z\left(\gamma^{-1}F^{(0)}_{z\bar{z}}+F_{ru}^{(-2)}\right)-\frac{i}{2}[A^{(0)}_{z},\gamma^{-1}F^{(0)}_{z\Bar{z}}+F_{ru}^{(-2)}]\;;
    \end{aligned}
\end{equation}
similarly, subtract \eqref{eq:Bianchi radial gauge 2} from \eqref{eq:EOM full Bondi radial zb component} with $n=2$, obtaining
\begin{equation}
\label{pf:F_uzb_-1}
    \begin{aligned}
        F^{(-1)}_{u\bar{z}}&=-\frac{1}{2}\partial_{\bar{z}}\left(\gamma^{-1}F^{(0)}_{z\bar{z}}-F_{ru}^{(-2)}\right)+\frac{i}{2}[A^{(0)}_{\bar{z}},\gamma^{-1}F^{(0)}_{z\bar{z}}-F_{ru}^{(-2)}]\;.
    \end{aligned}
\end{equation}
Use equations \eqref{pf:F_uz_-1} and \eqref{pf:F_uzb_-1} to rewrite \eqref{pf:partial_u_F_ru-3} as
\begin{equation}
    \begin{aligned}
        \partial_{u}F^{(-3)}_{ru}&=-\gamma^{-1}\partial_{z}\partial_{\bar{z}}F_{ru}^{(-2)}+\mathfrak{nl}\;,
    \end{aligned}
\end{equation}
where we defined
\begin{equation}
    \begin{aligned}
           \mathfrak{nl}:&=-i\gamma^{-1}\Big(\partial_{[z}[A^{(0)}_{\bar{z}]},\gamma^{-1}F^{(0)}_{z\bar{z}}]-\partial_{(z}[A^{(0)}_{\bar{z})},F_{ru}^{(-2)}]+\gamma[A^{(-1)}_{u},F^{(-2)}_{ur}] \\
           &\quad +2\sum_{k=0}^{1}[A^{(k-1)}_{(z},F^{(-k)}_{\Bar{z})u}]\Big)\;;
    \end{aligned}
\end{equation}
then we obtain:
\begin{equation}
    \begin{aligned}
        F^{(-3)}_{ru}&=-\gamma^{-1}\int_{-\infty}^ud\tilde u\left(\partial_{z}\partial_{\bar{z}}F_{ru}^{(-2)}+\mathfrak{nl}\right)\;.
    \end{aligned}
\end{equation}
At this point, the sub-leading charge is
\begin{equation}
    \begin{aligned}
        Q_\Sigma^\mathrm{sub}[\Lambda]&=\frac{1}{e^2}\int_{\mathcal{I}^+}dudz^2\tr\Big[  -\partial_u\Lambda^{(1)}\int_{-\infty}^ud\tilde u\left(\partial_{z}\partial_{\bar{z}}F_{ru}^{(-2)}+\mathfrak{nl}\right) \\
        &\quad +\gamma\partial_u\Lambda^{(0)}F_{ru}^{(-2)}+\gamma\Lambda^{(0)}\partial_uF_{ru}^{(-2)}-\Lambda^{(1)}\left(\partial_{z}\partial_{\bar{z}}F_{ru}^{(-2)}+\mathfrak{nl}\right)\Big]\;.
    \end{aligned}
\end{equation}
The term in the first line, when written in terms of $F_{uz}^{(0)}$ and $F_{u\bar z}^{(0)}$, is of the form \; \; \; $\int_{-\infty}^{+\infty} du \int_{-\infty}^ud\tilde u\int_{-\infty}^{\tilde u}d\tilde{\tilde u}$ (recall equation \eqref{pf:F_ru-2} in the computation of the leading charge above), which as in the previous subsection, does not correspond to a standard soft gluon insertion. The only way to make it vanish is to impose the condition
\begin{equation}
\label{pf:Lambda_condition_1}
    \partial_u\Lambda^{(1)}=0\;.
\end{equation}
We are now left with
\begin{equation}
    \begin{aligned}
        Q_\Sigma^\mathrm{sub}[\Lambda]&=\frac{1}{e^2}\int_{\mathcal{I}^+}dudz^2\tr\left[\left(\gamma\partial_u\Lambda^{(0)}-\partial_{z}\partial_{\bar{z}}\Lambda^{(1)}\right)F_{ru}^{(-2)}+\gamma\Lambda^{(0)}\partial_uF_{ru}^{(-2)}-\Lambda^{(1)}\mathfrak{nl}\right]\;,
    \end{aligned}
\end{equation}
where the first term is of the form $\int_{-\infty}^{+\infty} du \int_{-\infty}^ud\tilde u$, so again is of the wrong form, and we require it to vanish, thus arriving at
\begin{equation}
\label{pf:Lambda_condition_2}
    \partial_u\Lambda^{(0)}=\gamma^{-1}\partial_{z}\partial_{\bar{z}}\Lambda^{(1)}\;.
\end{equation}
Solving the previous equation, using condition \eqref{pf:Lambda_condition_1}, we finally conclude that
\begin{equation}
    \begin{cases}
        \Lambda^{(1)}=\mu(z,\bar z) \\
        \Lambda^{(0)}= \lambda(z,\zb) + u\gamma^{-1}\partial_z\partial_{\bar z}\mu(z,\bar z)=\lambda(z,\zb) + \frac{u}{2}\Delta\mu(z,\bar z)
    \end{cases}\;,
\end{equation}
where $\lambda(z,\zb)$ is a constant of integration from \eqref{pf:Lambda_condition_2}. This parameter contributes to the leading charge and therefore we discard it for the sub-leading charge.

Finally, using \eqref{pf:F_ru-2} again, we write the sub-leading charge as
\begin{equation}
    \begin{aligned}
        Q_\Sigma^\mathrm{sub}[\Lambda]&=\frac{1}{e^2}\int_{\mathcal{I}^+}dudz^2\tr\left(2\Lambda^{(0)}\partial_{(z}F_{\bar z)u}^{(0)}-2i\Lambda^{(0)}[A_{(z}^{(0)},F_{\bar z)u}^{(0)}]-\Lambda^{(1)}\mathfrak{nl}\right)\;.
    \end{aligned}
\end{equation}
Via this procedure we obtain a simple form of the sub-leading charges, avoiding the intricacies that follow from gauge fixing the gauge parameters, in particular the field dependence that can arise in certain gauges (e.g. Lorentz gauge).

\providecommand{\href}[2]{#2}\begingroup\raggedright\endgroup

\end{document}